\documentclass[11pt,a4paper,reqno]{amsart}%
\usepackage{amsthm,amsmath,amsfonts,amssymb,xcolor,amsxtra,appendix,bookmark,dsfont,bm,mathrsfs}
\usepackage[latin1]{inputenc}
\usepackage{color} 
\usepackage{amssymb}\usepackage{graphicx}
\usepackage{tikz}
\usepackage{hyperref}
\theoremstyle{plain}
\newtheorem{theorem}{Theorem}
\newtheorem{lemma}[theorem]{Lemma}

\newtheorem{proposition}[theorem]{Proposition}

\newtheorem{definition}{Definition}

\theoremstyle{remark}
\newtheorem{remark}[theorem]{Remark}

\allowdisplaybreaks

\title[Finite time energy cascade]{Finite time energy cascade for mixed $3-$ and $4-$wave kinetic equations}

\author[G. Staffilani]{Gigliola Staffilani
}
\address{Department of Mathematics, Massachusetts Institute of Technology, Cambridge, MA 02139, USA}
\email{gigliola@math.mit.edu} 
\thanks{G.S. is  funded in part by  the NSF grants DMS-2052651, DMS-2306378 and the Simons Foundation through the Simons Collaboration on Wave Turbulence.}

\author[M.-B. Tran]{Minh-Binh Tran}
\address{Department of Mathematics, Texas A\&M University, College Station, TX 77843, USA}
\email{minhbinh@tamu.edu} 
\thanks{M.-B. T is  funded in part by  a   Humboldt Fellowship,   NSF CAREER  DMS-2303146, and NSF Grants DMS-2204795, DMS-2305523,  DMS-2306379.}

\begin{document}
\date{\today}

\begin{abstract} 
In this work we study a kinetic equation whose collision operator comprises three distinct wave interaction mechanisms: one representing a 3-wave process, and two corresponding to 4-wave processes. This wave kinetic equation describes the temporal evolution of the density function of the thermal cloud of a finite temperature trapped Bose gas. We establish that, for a broad class of initial data, solutions exhibit an immediate cascade of energy towards arbitrarily large frequencies. Furthermore, for other classes of initial conditions, we demonstrate that the energy is transferred to infinity in finite time.

\end{abstract}

\maketitle

 \tableofcontents

\section{Introduction}\label{intro} 

The experimental realization of Bose-Einstein condensation (BEC) in dilute atomic gases~\cite{WiemanCornell, Ketterle, bradley1995evidence} marked a  milestone in quantum statistical physics, leading to extensive theoretical and more experimental investigations. In these experiments, rapid evaporative cooling is employed to lower the temperature of a Bose gas below its critical threshold, the BEC transition temperature, thereby initiating the emergence of a  condensate. One of the central challenges that has since emerged concerns the nonequilibrium dynamics of the condensate, including the condensed  and the non-condensed components,  at finite temperatures. \footnote{Below the BEC transition temperature, the system's temperature is very low. However, in practical implementations, the temperature cannot reach absolute zero. Therefore, the finite-temperature case is of greater importance.
 }

In this work, we study the  evolution of the thermal cloud (the non-condensed component), modeled by the density function \( f(t,k) \), which satisfies the kinetic equation
\begin{equation}
	\label{4wave}
	\partial_t f(t,k) = \mathbb{Q}[f](t,k) := C_{12}[f](t,k) + C_{22}[f](t,k) + C_{31}[f](t,k), \qquad f(0,k) = f_0(k),
\end{equation}
where the operators \( C_{12} \), \( C_{22} \), and \( C_{31} \) account for different types of collisions between particles, described as follows.

\medskip
\noindent\textit{3-wave collision operator:}
\begin{align}
	\label{C12}
	\begin{split}
		C_{12}[f] :=\ &\mathfrak{c}_{12} \iint_{\mathbb{R}^3 \times \mathbb{R}^3} \mathrm{d}k_1\, \mathrm{d}k_2 \; \mathcal{W}_{12}(|k|,|k_1|,|k_2|)\, \delta(\omega - \omega_1 - \omega_2)\, \delta(k - k_1 - k_2) \left[ f_1 f_2 - (f_1 + f_2)f \right] \\
		& - 2 \iint_{\mathbb{R}^3 \times \mathbb{R}^3} \mathrm{d}k_1\, \mathrm{d}k_2 \; \mathcal{W}_{12}(|k_1|,|k|,|k_2|)\, \delta(\omega_1 - \omega - \omega_2)\, \delta(k_1 - k - k_2) \left[ f f_2 - (f + f_2)f_1 \right],
	\end{split}
\end{align}

\medskip
\noindent\textit{4-wave collision operator (pair interactions):}
\begin{equation}
	\label{C22}
	\begin{aligned}
		C_{22}[f] :=\ &\mathfrak{c}_{22} \iiint_{\mathbb{R}^{3 \times 3}} \mathrm{d}k_1\, \mathrm{d}k_2\, \mathrm{d}k_3 \; \mathcal{W}_{22}(|k|,|k_1|,|k_2|,|k_3|)\, \delta(k + k_1 - k_2 - k_3)\, \delta(\omega + \omega_1 - \omega_2 - \omega_3) \\
		& \quad \times \left[ f_2 f_3 (f_1 + f) - f f_1 (f_2 + f_3) \right],
	\end{aligned}
\end{equation}

\medskip
\noindent\textit{4-wave collision operator (triplet interactions):}
\begin{align}
	\label{C31}
	\begin{split}
		C_{31}[f] :=\ &\mathfrak{c}_{31} \iiint_{\mathbb{R}^{3 \times 3}} \mathrm{d}k_1\, \mathrm{d}k_2\, \mathrm{d}k_3 \; \mathcal{W}_{31}(|k|,|k_1|,|k_2|,|k_3|)\, \delta(\omega - \omega_1 - \omega_2 - \omega_3)\, \delta(k - k_1 - k_2 - k_3) \\
		& \quad \times \left[ f_1 f_2 f_3 - (f_1 f_2 + f_2 f_3 + f_1 f_3)f \right] \\
		& - 3 \iiint_{\mathbb{R}^{3 \times 3}} \mathrm{d}k_1\, \mathrm{d}k_2\, \mathrm{d}k_3 \; \mathcal{W}_{31}(|k_1|,|k|,|k_2|,|k_3|)\, \delta(\omega_1 - \omega - \omega_2 - \omega_3)\, \delta(k_1 - k - k_2 - k_3) \\
		& \quad \times \left[ f f_2 f_3 - (f f_2 + f_2 f_3 + f f_3)f_1 \right],
	\end{split}
\end{align}

where, \( \mathfrak{c}_{12}, \mathfrak{c}_{22}, \mathfrak{c}_{31} > 0 \) are constants, \( t \in \mathbb{R}_+ \) denotes time, \( k \in \mathbb{R}^3 \) is the momentum variable, and \( \omega(k) = \omega(|k|) \) is the dispersion relation of the quasiparticles. The initial data \( f_0(k) \geq 0 \) denotes the initial thermal distribution.

\textbf{Assumption $X$:}

The interaction kernels are assumed to take the following factorized forms:
\begin{equation}
	\label{W21}
	\mathcal{W}_{12}(|k|,|k_1|,|k_2|) = \bar{\mathfrak{P}}(\omega)\, \bar{\mathfrak{P}}(\omega_1)\, \bar{\mathfrak{P}}(\omega_2),
\end{equation}
\begin{equation}
	\label{W22}
	\mathcal{W}_{22}(|k|,|k_1|,|k_2|,|k_3|) = \bar{\mathfrak{R}}(\omega)\, \bar{\mathfrak{R}}(\omega_1)\, \bar{\mathfrak{R}}(\omega_2)\, \bar{\mathfrak{R}}(\omega_3)\, \bar{\mathfrak{R}}_o(\omega,\omega_1,\omega_2,\omega_3),
\end{equation}
\begin{equation}
	\label{W31}
	\mathcal{W}_{31}(|k|,|k_1|,|k_2|,|k_3|) = \frac{1}{\left|\, |k| - |k_1| - |k_2| - |k_3| \,\right|} \bar{\mathfrak{Q}}(\omega)\, \bar{\mathfrak{Q}}(\omega_1)\, \bar{\mathfrak{Q}}(\omega_2)\, \bar{\mathfrak{Q}}(\omega_3).
\end{equation}

Detailed assumptions on the functions \( \bar{\mathfrak{P}}, \bar{\mathfrak{Q}}, \bar{\mathfrak{R}} \) are provided below.
\begin{itemize}

		\item[(X1)] There exist constants \( 1 <   1/\theta   \) and \( 0 < C_{\omega} \) such that the dispersion relation satisfies the bounds
		\begin{equation}\label{X1}
			  \omega(k) = C_{\omega} |k|^{1/\theta}.
		\end{equation}
		The function \(\omega(|k|)\) is continuous and non-decreasing, i.e., \(\omega'(|k|) \ge 0\) for all \(k \in \mathbb{R}^3\). In addition, it satisfies the  condition
		\begin{equation}\label{X1:1}
			\omega(|k_1| + |k_2|) \ge \omega(|k_1|) + \omega(|k_2|), \quad \forall\, k_1, k_2 \in \mathbb{R}^3,
		\end{equation}
		and $\omega(0)=0$. 
		Moreover, for each value of \(\omega\), there exists a unique corresponding \(|k|\) such that \(\omega(|k|) = \omega\). We denote this inverse relationship explicitly as \(|k| = |k|(\omega)\).

	\item[(X2)] We assume that the solution \( f \) is radial, meaning \( f(k) = f(|k|) \). Using a change of variables, the integral of $f$ becomes
	\begin{equation}\label{X2}
		\int_{\mathbb{R}^3} \mathrm{d}k\, f(k) = 2\pi^2 \int_{[0,\infty)} \mathrm{d}\omega\, \frac{|k|^2}{\omega'(|k|)} f(\omega).
	\end{equation}
	Hence, we can identify \( f(k) \) with \( f(\omega) \), and write \( f(k) = f(\omega) \) accordingly.

	Assuming that $f(k)$ is radial, we  adopt the shorthand notation:
	\begin{equation}
		\label{Shorthand}
		\begin{aligned}
			& f = f(k) =  f(\omega),\quad f_1 = f(k_1) = f(\omega_1),\quad f_2 = f(k_2) = f(\omega_2), \quad f_3 = f(k_3) = f(\omega_3)\\
			& \omega = \omega(k),\quad \omega_1 = \omega(k_1),\quad \omega_2 = \omega(k_2), \quad \omega_3 = \omega(k_3).	\end{aligned}
	\end{equation}
	\item[(X3)] We define
	\begin{equation}\label{X3}
	 \Gamma(k) =	\Gamma(\omega) := \frac{|k|^2}{\omega'(|k|)},
	\end{equation}
	\begin{equation}\label{X3:1}
		\mathfrak{P}(\omega) := \Gamma(\omega)\bar{\mathfrak{P}}(\omega), \quad
		\mathfrak{Q}(\omega) := \Gamma(\omega)\bar{\mathfrak{Q}}(\omega), \quad
		\mathfrak{R}(\omega) := \Gamma(\omega)\bar{\mathfrak{R}}(\omega)/|k|,
	\end{equation}
	and
	\begin{equation}\label{X3:2}
		\tilde{\mathfrak{P}}(\omega) := \frac{\bar{\mathfrak{P}}(\omega)}{\omega}, \quad
		\tilde{\mathfrak{Q}}(\omega) := \frac{\bar{\mathfrak{Q}}(\omega)}{\omega}, \quad
		\tilde{\mathfrak{R}}(\omega) := \frac{\bar{\mathfrak{R}}(\omega)}{\omega}.
	\end{equation}
	
Denoting the derivative  of   \( {\mathfrak{R}}(\omega) \) by    \( {\mathfrak{R}}'(\omega) \),  we assume that \( {\mathfrak{P}}(\omega) \), \( {\mathfrak{Q}}(\omega) \), \( {\mathfrak{R}}(\omega) \),     \( \bar{\mathfrak{P}}(\omega) \), \( \bar{\mathfrak{Q}}(\omega) \), and \( \bar{\mathfrak{R}}(\omega) \) are continuous and non-decreasing functions of \( \omega \), and that there exist constants \( C_{\mathfrak{P}}, C_{\mathfrak{Q}}, C_{\mathfrak{R}} > 0 \) and exponents \( -1\le \varpi_1, \varpi_2, \varpi_3 \le 0 \) such that

\begin{equation}\label{X3:3}
	\begin{aligned}
 C_{\mathfrak{P}}' \omega^{\varpi_1} 	\ge\ 	& \tilde{\mathfrak{P}}(\omega) \ge C_{\mathfrak{P}} \omega^{\varpi_1}, \\
	 C_{\mathfrak{R}}'  \omega^{\varpi_2}	\ge	\ & \tilde{\mathfrak{R}}(\omega) \ge C_{\mathfrak{R}} \omega^{\varpi_2}, \\
 C_{\mathfrak{Q}}' \omega^{\varpi_3}	\ge	\	 & \tilde{\mathfrak{Q}}(\omega) \ge C_{\mathfrak{Q}} \omega^{\varpi_3}, \quad \forall \omega \ge 1.
	\end{aligned}
\end{equation}
Moreover, \(  \bar{\mathfrak{P}}, \bar{\mathfrak{Q}}, \bar{\mathfrak{R}} \) are continuous functions, \( \mathfrak{P}(0) = \mathfrak{Q}(0) = \mathfrak{R}(0) = 0 \) and there exist constants \( C_{\mathfrak{P}'}, C_{\mathfrak{Q}'}, C_{\mathfrak{R}'} > 0 \) and exponents \( -1\le \kappa_2 \) such that
\begin{equation}\label{X3:4}
	\begin{aligned}
C_{\mathfrak{R}'}' \omega^{\kappa_2}	\ge	& {\mathfrak{R}'}(\omega) \ge C_{\mathfrak{R}'} \omega^{\kappa_2}.
	\end{aligned}
\end{equation}

Moreover,
\begin{equation}\label{X3:5}
\bar{\mathfrak{R}}_o(\omega,\omega_1,\omega_2,\omega_3) \ = \ \max\{\omega,\omega_1,\omega_2,\omega_3\}^{\gamma},
\end{equation}
for $0\le \gamma\le 1$.
	
Fix $\alpha \in(0,1)$. The parameters satisfy the following relation:
	\begin{equation}\label{X4}
	 \begin{aligned}
			0\  < & \ 4\varpi_{3}+3\theta+\alpha,\\ 	
		0\  < & \  3\varpi_{1}+3\theta+\alpha,\\ 	
		0\  < & \ 4\varpi_2+3+ \alpha+\gamma,\\  	
		0\  < & \ 3\varpi_2 + 2 - 2\theta\\
		\gamma+\kappa_2 \ \ge &\ 0,\\
		3\theta + 2\varpi_1\  \le & \ 0, \\ 2\theta + 2\varpi_2 \ \le &\ 0,\\ \ 3\theta + 2\varpi_3\ \le &\ 0, \\
		2\varpi_2+\theta+\gamma\ \ge & \ 0.\end{aligned}
	\end{equation}

Similar to~\eqref{Shorthand}, we also employ the shorthands $\mathfrak{P}_1$, $\mathfrak{P}_2$, $\mathfrak{R}_1$, $\mathfrak{R}_2$, $\mathfrak{R}_3$, $\mathfrak{Q}_1$, $\mathfrak{Q}_2$, and $\mathfrak{Q}_3$.

	 We define the function
	\begin{equation} \label{FDefinition}
		\mathfrak{F}(t,k) := f(t,k)\, \Gamma(k), \quad t \geq 0.
	\end{equation}
\end{itemize}

\begin{remark}
As an example, we consider the dispersion relation \( \omega(k) = |k|^4 \).  
In this case, we have \( \theta =  \tfrac{1}{4} \).  
We choose \( \kappa_2 = 2\theta + \varpi_2 = \varpi_2 + \tfrac{1}{2} \).  
Therefore, condition~\eqref{X4} becomes
\begin{equation}\label{X4:1}
	\begin{aligned}
		0 &< 4\varpi_{3} + \tfrac{3}{4} + \alpha, \quad 
		0 < 3\varpi_{1} + \tfrac{3}{4} + \alpha,\\ 	
		0 &< 4\varpi_2 +  \alpha + \gamma, \quad  	
		0 < 3\varpi_2 + \tfrac{3}{2},\\
		\tfrac{3}{4} + 2\varpi_1 &\le 0, \quad 
		\tfrac{1}{2} + 2\varpi_2 \le 0, \quad 
		\tfrac{3}{4} + 2\varpi_3 \le 0,\quad \gamma +  \varpi_2 + \tfrac{1}{2}  \ge 0, \quad	2\varpi_2+\tfrac{1}{4} +\gamma\ \ge  \ 0.
	\end{aligned}
\end{equation}
We can choose \( 1>\alpha > \tfrac34,\ \varpi_{1} = -\tfrac{3}{8}, \ \varpi_{2} = -\tfrac{1}{4}, \ \varpi_{3} = -\tfrac{3}{8}, \  \kappa_2 = \tfrac{1}{4}, \gamma= \tfrac{1}{4} \) and
\begin{equation*}
	\begin{aligned}
	& \tilde{\mathfrak{P}}(\omega) =  \omega^{\varpi_1}, \ \
	 \tilde{\mathfrak{R}}(\omega) = \omega^{\varpi_2}, \ \
	 \tilde{\mathfrak{Q}}(\omega) = \omega^{\varpi_3}.
	\end{aligned}
\end{equation*}

\end{remark}
We define the energy of the solution as
$$\int_{[0,\infty)} \, \mathrm{d}\omega  f(t,\omega)\, \omega \Gamma(\omega),$$
and set the maximal time of energy conservation by
\begin{equation} \label{T0}
	T^* := \sup\left\{ T \,\middle|\, \int_{[0,\infty)} \, \mathrm{d}\omega  f(t,\omega)\, \omega \Gamma(\omega)
	= \int_{[0,\infty)} \, \mathrm{d}\omega f(0,\omega)\, \omega \Gamma(\omega) \quad \text{for all } t \in [0, T) \right\}.
\end{equation}

{\it The primary objective of this work is to demonstrate that, under certain assumptions on the initial conditions, the  time \( T^* \) can be finite or even equal to zero. This result indicates that the energy of the solution is transferred to arbitrarily large frequencies either instantaneously or within a finite time interval. We refer to Remark~\ref{RemarkMainTheo:1} for the discussion on \( T^* \) and the energy cascade phenomenon.
}

\subsection{Physical context }\label{Subs:PhysicalContext}
The kinetic theory of Bose gases at finite temperatures was first systematically developed by Kirkpatrick and Dorfman~\cite{KD1, KD2}, who laid the groundwork for deriving kinetic equations governing the behavior of particles outside the condensate. Building on this foundation, Zaremba, Nikuni, and Griffin formulated in~\cite{ZarembaNikuniGriffin:1999:DOT} a more comprehensive model that couples a quantum Boltzmann equation, describing the dynamics of the thermal (non-condensed) component, with a Gross-Pitaevskii equation that governs the evolution of the Bose-Einstein condensate (BEC) itself. Independently, a similar kinetic framework was proposed by Pomeau, Brachet, Metens, and Rica in~\cite{PomeauBrachetMetensRica}.

These kinetic models incorporate two primary types of collision processes, each represented by an  operator whose detailed forms will be introduced later:
\begin{itemize}
	\item \( \mathscr{C}_{22} \): models binary (\(2 \leftrightarrow 2\)) collisional events between excited particles.
	\item \( \mathscr{C}_{12} \): describes interactions of the form \(1 \leftrightarrow 2\), involving the condensate and two thermal excitations.
\end{itemize}

A significant development in this theory came with the work of Reichl and Gust~\cite{ReichlGust:2012:CII}, who proposed a third collision operator, \( \mathscr{C}_{31} \), capturing \(1 \leftrightarrow 3\) collisional events among excitations. However,  the  mathematical confirmation of \( \mathscr{C}_{31} \) was achieved only recently in~\cite{tran2020boltzmann}. Its physical relevance has since been supported by experimental observations reported in~\cite{reichl2019kinetic}.

For a broader overview of these models and their roles in the study of Bose-Einstein condensation dynamics, we refer the reader to the reviews in~\cite{GriffinNikuniZaremba:BCG:2009, PomeauBinh, tran2021thermal}.

The complete quantum kinetic model for a Bose gas hence includes three distinct collision operators and is described by the equation
\begin{equation}
	\label{KineticFinal}
	\partial_t f(t,k) = \mathscr{C}_{12}[f](t,k) + \mathscr{C}_{22}[f](t,k) + \mathscr{C}_{31}[f](t,k), \quad f(0,k) = f_0(k),
\end{equation}
where the operators $\mathscr{C}_{12}$, $\mathscr{C}_{22}$, and $\mathscr{C}_{31}$, as mentioned above, represent different physical interaction mechanisms and are defined explicitly as follows:

\paragraph{\textit{The $\mathscr{C}_{12}$ Collision Operator:}}
\begin{equation}
	\label{C12Discrete}
	\begin{aligned}
		\mathscr{C}_{12}[f](t,k) =\; & 4\pi g^2 n \iiint_{\mathbb{R}^3} \mathrm{d}k_1\, \mathrm{d}k_2\, \mathrm{d}k_3\;
		\left[\delta(k - k_1) - \delta(k - k_2) - \delta(k - k_3)\right] \\
		&\times \delta\big(\omega(k_1) - \omega(k_2) - \omega(k_3)\big) \, \left(K_{1,2,3}^{1,2}\right)^2 \, \delta(k_1 - k_2 - k_3) \\
		&\times \left[ f(k_2) f(k_3) (f(k_1) + 1) - f(k_1)(f(k_2) + 1)(f(k_3) + 1) \right].
	\end{aligned}
\end{equation}

\paragraph{\textit{The $\mathscr{C}_{22}$ Collision Operator:}}
\begin{equation}
	\label{C22Discrete}
	\begin{aligned}
		\mathscr{C}_{22}[f](t,k) =\; & \pi g^2 \iiiint_{\mathbb{R}^3} \mathrm{d}k_1\, \mathrm{d}k_2\, \mathrm{d}k_3\, \mathrm{d}k_4 \;
		\left[\delta(k - k_1) + \delta(k - k_2) - \delta(k - k_3) - \delta(k - k_4)\right] \\
		&\times \left(K_{1,2,3,4}^{2,2}\right)^2 \, \delta(k_1 + k_2 - k_3 - k_4)\, \delta\big(\omega(k_1) + \omega(k_2) - \omega(k_3) - \omega(k_4)\big) \\
		&\times \left[ f(k_3) f(k_4) (f(k_1) + 1)(f(k_2) + 1) - f(k_1) f(k_2) (f(k_3) + 1)(f(k_4) + 1) \right].
	\end{aligned}
\end{equation}

\paragraph{\textit{The $\mathscr{C}_{31}$ Collision Operator:}}
\begin{equation}
	\label{C31Discrete}
	\begin{aligned}
		\mathscr{C}_{31}[f](t,k) =\; & \pi g^2 \iiiint_{\mathbb{R}^3} \mathrm{d}k_1\, \mathrm{d}k_2\, \mathrm{d}k_3\, \mathrm{d}k_4 \;
		\left[\delta(k - k_1) - \delta(k - k_2) - \delta(k - k_3) - \delta(k - k_4)\right] \\
		&\times \left(K_{1,2,3,4}^{3,1}\right)^2 \, \delta(k_1 - k_2 - k_3 - k_4)\, \delta\big(\omega(k_1) - \omega(k_2) - \omega(k_3) - \omega(k_4)\big) \\
		&\times \left[ f(k_2) f(k_3) f(k_4)(f(k_1) + 1) - f(k_1)(f(k_2) + 1)(f(k_3) + 1)(f(k_4) + 1) \right].
	\end{aligned}
\end{equation}

Here, \( n \) denotes the condensate density and \( g \) is the interaction constant. The functions \( \left(K_{1,2,3}^{1,2}\right)^2 \), \( \left(K_{1,2,3,4}^{2,2}\right)^2 \), and \( \left(K_{1,2,3,4}^{3,1}\right)^2 \) are known as the collision kernels and are explicit.

In a full coupled framework, the evolution of the condensate is governed by the Gross-Pitaevskii equation, which dynamically determines the condensate density \( n = n(t) \). However, in this study, to simplify the analysis, we assume that the condensate is sufficiently dense so that \( n \) can be treated as a fixed constant. For the more general time-dependent setting, we refer the reader to \cite{PomeauBinh}; see also Figure~\ref{fig2}.

To streamline notation in what follows, we define
\begin{equation} \label{Simpl1}
	\mathfrak{c}_{12} := 4\pi g^2 n, \qquad \mathfrak{c}_{22} := \pi g^2, \qquad \mathfrak{c}_{31} := \pi g^2.
\end{equation}

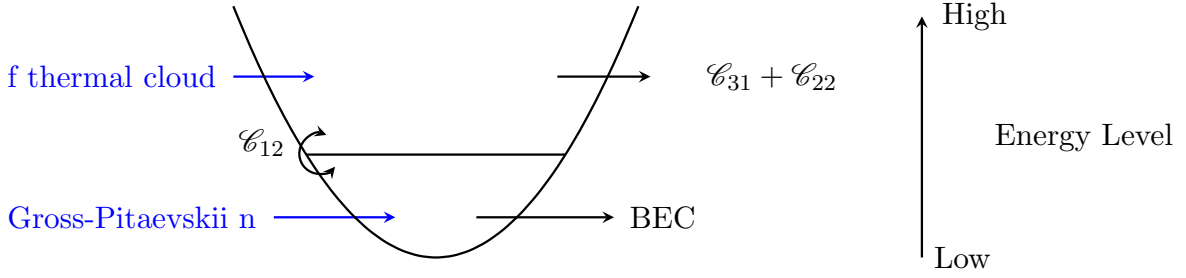
\begin{figure}
	\centering
	\resizebox{\textwidth}{!}{%
		\begin{tikzpicture}[>=stealth]
			\draw[thick, ->, blue] (-2.5,2.25) -- ++(1,0) node[xshift=-2.5cm] {f thermal cloud};
			\draw[thick, ->] (1.5,2.25) -- ++(1.15,0) node[xshift=1.5cm] { $\mathscr C_{31}+\mathscr C_{22}$};
			\draw[thick, ->, blue] (-2,.5) -- ++(1.5,0) node[xshift=-3.25cm] {Gross-Pitaevskii n };
			\draw[thick, ->] (0.5,.5) -- ++(1.7,0) node[xshift=.6cm] { BEC};
			\draw[thick, ->] (6,0) node [xshift=.5cm]{Low} -- ++(0, 1.5)  node[xshift=2.0cm] {Energy Level} -- ++(0,1.5) node[xshift=.65cm] {High};
			\draw[ thick,domain=-2.5:2.5,smooth,variable=\x,black] plot ({\x},{.5*\x*\x});
			\draw [thick, black] (-1.6, 1.28) -- (1.6, 1.28);
			
			\draw[thick, <->] (-1.35,1.525) arc (70:320:.25);
			\node at (-2.15, 1.425) {$\mathscr C_{12}$};
		\end{tikzpicture}
	}%
	\caption{The Bose--Einstein Condensate (BEC) and the thermal cloud.}\label{fig2}
\end{figure}

To simplify the collision terms, a common approximation is to retain only the dominant higher-order contributions while neglecting lower-order terms. As thus, we approximate the nonlinear expressions as follows:
\begin{itemize}
	\item The $\mathscr{C}_{12}$ term:
	\[
	f(k_2) f(k_3) - f(k_1)\left(f(k_2) + f(k_3) + 1\right)
	\quad \longrightarrow \quad
	f(k_2) f(k_3) - f(k_1)\left(f(k_2) + f(k_3)\right),
	\]
	\item The $\mathscr{C}_{22}$ term:
	\[
	f(k_3) f(k_4) (f(k_1) + 1)(f(k_2) + 1) - f(k_1) f(k_2) (f(k_3) + 1)(f(k_4) + 1)
	\]
	\[
	\longrightarrow \quad f(k_3) f(k_4)\left(f(k_1) + f(k_2)\right) - f(k_1) f(k_2)\left(f(k_3) + f(k_4)\right),
	\]
	\item The $\mathscr{C}_{31}$ term:
	\[
	f(k_2) f(k_3) f(k_4)(f(k_1) + 1) - f(k_1)(f(k_2) + 1)(f(k_3) + 1)(f(k_4) + 1)
	\]
	\[
	\longrightarrow \quad f(k_2) f(k_3) f(k_4) - \left[f(k_2)f(k_3) + f(k_2)f(k_4) + f(k_3)f(k_4)\right].
	\]
\end{itemize}

Under these approximations, the collision terms become:
\begin{align}
	& f(k_2) f(k_3) - f(k_1)\left(f(k_2) + f(k_3)\right), \label{Simpl2} \\
	& f(k_3) f(k_4)\left(f(k_1) + f(k_2)\right) - f(k_1) f(k_2)\left(f(k_3) + f(k_4)\right), \label{Simpl3} \\
	& f(k_2) f(k_3) f(k_4) - \left[ f(k_2)f(k_3) + f(k_2)f(k_4) + f(k_3)f(k_4) \right]. \label{Simpl4}
\end{align}

Furthermore, we replace the complicated collision kernels $\left(K_{1,2,3}^{1,2}\right)^2$, $\left(K_{1,2,3,4}^{2,2}\right)^2$, and $\left(K_{1,2,3,4}^{3,1}\right)^2$ with  effective coefficients $\mathcal{W}_{12}$, $\mathcal{W}_{22}$, and $\mathcal{W}_{31}$, respectively, and defined in \eqref{W21}-\eqref{W31}. These reductions, together with the simplifications in \eqref{Simpl2}--\eqref{Simpl4} and the constant definitions in \eqref{Simpl1}, yield the reduced kinetic equation \eqref{4wave}.

\medskip

In our model, we adopt the following structural assumption:

\bigskip

\noindent\textbf{Assumption Y.} \textit{We assume the constants satisfy} 
\[
\mathfrak{c}_{12} + \mathfrak{c}_{22} + \mathfrak{c}_{31} > 0.
\]

\subsection{Wave kinetic equations}

Wave turbulence theory, which investigates the behavior of nonlinear wave systems far from thermal equilibrium, has its origins in the pioneering contributions of Peierls~\cite{Peierls:1993:BRK,Peierls:1960:QTS}, Brout and Prigogine~\cite{brout1956statistical}, Zaslavskii and Sagdeev~\cite{zaslavskii1967limits}, Hasselmann~\cite{hasselmann1962non,hasselmann1974spectral}, and the works of Benney, Saffman, Newell~\cite{benney1966nonlinear,benney1969random}, and Zakharov~\cite{zakharov2012kolmogorov}. This framework has since become a cornerstone in both theoretical studies and practical applications involving weakly nonlinear wave systems.

Within the context of wave kinetic theory, interactions among wave modes are typically categorized into two main classes based on the number of participating waves: 3-wave and 4-wave processes. In this classification, the operator \( C_{12} \) represents a 3-wave interaction, while \( C_{22} \) and \( C_{31} \) correspond to 4-wave interactions.

The time evolution of a dilute Bose gas at room temperature is classically governed by the Boltzmann-Nordheim equation~\cite{Nordheim}, given by
\begin{equation}
	\label{Nordheim}
	\partial_t F(t,k) = \mathscr{C}_{22}[F](t,k), \qquad \text{with } \left(K_{1,2,3,4}^{2,2}\right)^2 = 1,
\end{equation}
which reduces to the simplified 4-wave kinetic equation
\begin{equation}
	\label{4wavepre}
	\partial_t F(t,k) = C_{22}[F](t,k), \qquad \text{with } \left(K_{1,2,3,4}^{2,2}\right)^2 = 1.
\end{equation}
It is well established in the physics literature~\cite{josserand2006self, PomeauBinh, Spohn:2010:KOT} that solutions to~\eqref{Nordheim} and~\eqref{4wavepre} can develop finite-time singularities, which are interpreted as the onset of Bose-Einstein condensation.

In the special case where \( \omega(k) = |k|^2 \) and the solution is isotropic, i.e., \( F(t,k) = F(t,|k|) = F(t,\omega) \), the formation of condensates for the 4-wave kinetic equation has been rigorously analyzed in~\cite{EscobedoVelazquez:2015:FTB, EscobedoVelazquez:2015:OTT}. These results have recently been extended to a broader class of dispersion relations and initial data in~\cite{staffilani2024energy, staffilani2024condensation}.

One key conclusion of these studies is that the total energy strictly vanishes on any bounded domain in the long-time limit:
\begin{equation}\label{CascadePre}
\lim_{t\to\infty} \int_{[0,R)} f(t,\omega)\, \omega \Gamma(\omega)\, \mathrm{d}\omega 
=0,  \ \ \ \ \quad \forall R,\   \infty>R>0.
\end{equation}

The present work generalizes the work~\cite{staffilani2024energy, staffilani2024condensation} to investigate the maximal time of energy conservation, \( T^* \), defined in~\eqref{T0}, for the much more general and complicated equation \eqref{4wave}. We demonstrate that under suitable conditions, \( T^* \) must be either finite or zero. We refer to Remark~\ref{RemarkMainTheo:1} for the discussion on \( T^* \) and the energy cascade phenomenon.

Our approach is built upon   a domain decomposition method (DDM) based on partitioning the half-line \( \mathbb{R}_+ \) into small intervals \cite{halpern2009nonlinear, Lions:1989:OSA, toselli2004domain}. This allows for a detailed divide-and-conquer analysis of the energy distribution within each subinterval, from which we derive precise estimates on the outward flow of energy. Moreover, this framework facilitates a quantitative comparison of energy transport toward infinity under the collective effects of all three collision operators: \( C_{12} \), \( C_{22} \), and \( C_{31} \).

The conclusions of our main theorem remain valid even when two out of the three collision operators vanish, as specified by Assumption~Y. In particular, the finite-time energy cascade results continue to hold in each of the following degenerate cases where only one collision operator is active:
\begin{itemize}
	\item \( \mathfrak{c}_{12} \ne 0 \), while \( \mathfrak{c}_{22} = \mathfrak{c}_{31} = 0 \),
	\item \( \mathfrak{c}_{22} \ne 0 \), while \( \mathfrak{c}_{12} = \mathfrak{c}_{31} = 0 \),
	\item \( \mathfrak{c}_{31} \ne 0 \), while \( \mathfrak{c}_{12} = \mathfrak{c}_{22} = 0 \).
\end{itemize}
That is, even when the dynamics are governed by a single collision mechanism among \( C_{12} \), \( C_{22} \), or \( C_{31} \), the energy cascade to infinity still occurs immediately or in finite time under the assumptions of our theorem.  {\it To the best of our knowledge, this is the first result of its kind for wave kinetic equations.} We remark that some numerical simulations of the system concerning $C_{12}$ and $C_{22}$ have been obtained in \cite{das2025energy}.

We now provide a brief overview of the current state of  results concerning the analysis of both 4-wave and 3-wave kinetic equations:

\begin{itemize}

		\item \textit{4-wave kinetic equations of $2 \leftrightarrow 2$ type:} Convergence rates for discrete approximations and local well-posedness results have been obtained for the MMT model (a one-dimensional 4-wave equation), in~\cite{dolce2024convergence, germain2023local}. Stability near equilibrium, scattering theory as well as the behavior and instability of the Kolmogorov-Zakharov (KZ) spectrum (a class of stationary solutions), have been studied in~\cite{menegaki20222,ampatzoglou2024scattering, escobedo2024instability, collot2024stability}. Further results on  local well-posed and ill-posed results and propagation of moments for 4-wave kinetic equations, particularly in polynomially weighted \( L^\infty \) spaces, are available in~\cite{ampatzoglou2025ill,ampatzoglou2025optimal,GermainIonescuTran, ampatzoglou2025inhomogeneous}.
		
			\item \textit{4-wave kinetic models on the torus:} Recent work on spatially periodic domains includes studies of entropy maximizers and spectral stability~\cite{escobedo2024entropy, germain2024stability}.
	
	\item \textit{Quantum kinetic equations with \( \mathscr{C}_{12} \) and \( \mathscr{C}_{22} \):} The global existence of classical solutions for models involving both of these collision operators has been established in~\cite{soffer2018dynamics}.  The Nordheim equation has been the subject of extensive investigation, with significant contributions documented in \cite{escobedo2007fundamental,Lu2000_ModifiedBoltzmannBE,Lu2004_IsotropicDistributionalBE,Lu2016_LongTimeBEC,Lu2018_LongTimeStrongBE}.
	
		\item \textit{3-wave kinetic equations:} These have been studied extensively across various physical contexts. Applications include stratified ocean flows~\cite{GambaSmithBinh}, Bose-Einstein condensates~\cite{cortes2020system, EPV, escobedo2023linearized1, escobedo2023linearized, escobedo2025local, nguyen2017quantum}, phonon interactions in crystal lattices~\cite{AlonsoGambaBinh, CraciunBinh, EscobedoBinh, GambaSmithBinh, tran2020reaction}, capillary waves~\cite{das2024numerical, nguyen2017quantum, soffer2020energy, walton2022deep, walton2023numerical, walton2024numerical}, and beam-wave interactions~\cite{rumpf2021wave}. The formation of condensates in the setting of non-radial solutions has been addressed in~\cite{staffilani2025formation}.

\end{itemize}


\section{Main results}\label{Sec:Setting}

We denote by $\mathscr{R}_+([0,\infty))$ the space of all non-negative Radon measures $f$ on $[0,\infty)$ such that
\begin{equation}\label{Radon}
	\|f\|_{\mathscr{R}_+} := \int_{[0,\infty)} \,\mathrm{d}\omega f(\omega)
	= \int_{[0,\infty)} f(\mathrm{d}\omega) < \infty.
\end{equation}

\begin{definition}\label{def}
	Under Assumption X and Assumption Y, we say that \( f(t,k) = f(t,\omega) \) is a \emph{global mild radial solution} to \eqref{4wave} with radial initial data \( f_0(k) = f_0(|k|) \ge 0 \) if the following hold:
	\begin{itemize}
		\item \( f(t,k) \ge 0 \) for all \( t \ge 0 \),
		\item the map \( t \mapsto f(t,\omega) \, \Gamma(\omega) \) lies in \( C^1([0,\infty), \mathscr{R}_+([0,\infty))) \),
		\item and for every test function \( \Xi \in C_c^2([0,\infty)) \) with compact support and satisfying \( \Xi'(0) = 0 \), the integral identity
		\begin{equation}\label{4wavemild}
			\begin{aligned}
				\int_{[0,\infty)} \mathrm{d}\omega\, f(t,\omega)\, \Gamma(\omega)\, \Xi(\omega)
				=\ & \int_{[0,\infty)} \mathrm{d}\omega\, f_0(\omega)\, \Gamma(\omega)\, \Xi(\omega) \\
				& + \int_0^t \mathrm{d}s \int_{[0,\infty)} \mathrm{d}\omega\, \mathbb{Q}[f](s,\omega)\, \Xi(\omega)\, \Gamma(\omega)
			\end{aligned}
		\end{equation}
		holds for all \( t \ge 0 \).
	\end{itemize}
\end{definition}

We have the main theorem, whose proof is given in Section \ref{Sec:Proof}.

\begin{theorem}
	\label{Theorem1} 
Assume Assumption~X and Assumption~Y hold.

Let the initial data \( f_0(k) = f_0(|k|) \geq 0 \) be radially symmetric and satisfy the mass and energy equations:
\begin{equation} \label{MassEnergy}
	\int_{\mathbb{R}^3} \, \mathrm{d}k f_0(k)= \mathfrak{M}, \quad 
	\int_{\mathbb{R}^3} \, \mathrm{d}k f_0(k)\, \omega(k)= \mathfrak{E},
\end{equation}
for some fixed constants \( \mathfrak{M}, \mathfrak{E} > 0 \).

Then there exists at least one global mild radial solution \( f(t,k) \) to equation~\eqref{4wave}, in the sense of Definition \ref{def}, such that
\begin{equation} \label{Theorem1:2}
	\int_{\mathbb{R}^3}\, \mathrm{d}k f(t,k) \leq \mathfrak{M}, \quad \text{for all } t \geq 0.
\end{equation}

Suppose further that the initial data satisfies a lower bound on high-frequency tails: there exist constants \( c_{\mathrm{in}} > 0 \), 
\begin{equation} \label{Theorem1:5}
	0 < c_{\mathrm{in}} <  3\varpi_2 + 2 - 2\theta + \kappa_2  + \gamma,
\end{equation}
 and \( r_0 > 0 \) such that
\begin{equation} \label{Theorem1:4}
	\int_{[R,\infty)}\, \mathrm{d}\omega\omega \mathfrak{F}(0,\omega) \geq C_{\mathrm{in}}\, R^{-c_{\mathrm{in}}}, \quad \text{for all } R > r_0,
\end{equation}
where the initial condition is defined as \( \mathfrak{F}(0,\omega) := f_0(k)\, \Gamma(k)\) (see \eqref{FDefinition}).

Then the following statements hold:

\begin{itemize}
	\item[(i)] \textbf{Immediate Energy Cascade.} If the exponent \( c_{\mathrm{in}} \) satisfies
	\begin{equation} \label{Theorem1:3}
		0 < c_{\mathrm{in}} < 	\frac15\min\left\{  \frac{4\varpi_{3}+3\theta+\alpha}{3},\ \frac{3\varpi_{1}+3\theta+\alpha}{2},\ \frac{4\varpi_2+ \alpha+\gamma}{6},\   3\varpi_2 + 2 - 2\theta + \kappa_2 - c_{\mathrm{in}} + \gamma  \right\},
	\end{equation}
	then the maximal time of energy conservation defined in~\eqref{T0} is zero:
	\[
	T^* = 0.
	\]
	
	\item[(ii)] \textbf{Finite-Time Energy Cascade.} If we assume only \eqref{Theorem1:5} 
		 then the maximal time \( T^* \) is finite:
	\[
	T^* < \infty.
	\]
	
\end{itemize}
\end{theorem}

\begin{remark}\label{RemarkMainTheo:1}

A standard symmetry argument shows that
\begin{equation}
	\label{RemarkMainTheo:E1}
	\partial_t \int_{\mathbb{R}^3} \mathrm{d}k\, f(t,k)\omega(k)
	= \int_{\mathbb{R}^3} \mathrm{d}k\, \mathbb{Q}[f](t,k)\omega(k)
	= 0.
\end{equation}
This means that, formally, the energy is conserved.
However, when we choose the solution defined on the space $\mathscr{R}_+([0,\infty))$ in~\eqref{Radon}, 
the integrals in $\int_{\mathbb{R}^3} \mathrm{d}k\, \mathbb{Q}[f](t,k)\omega(k)$ are not well defined. 
Consequently, the energy is not conserved and as thus $T^* < \infty$. 
In fact, our proof shows that the energy is lost to infinity.

As explained in \cite{soffer2019energy}, the energy cascade is analogous to the phenomenon of \emph{gelation} in coagulation--fragmentation equations \cite{banasiak2019analytic}, where the gelation time represents the first instant at which mass is lost. The loss of mass occurs because particles merge rapidly to form a giant particle of infinite size $\omega = \infty$  within the system. Similarly, in the energy cascade, a giant wave forms in the system, corresponding to an infinite wave number $\omega = \infty$. Therefore, the energy cascade time $T^*$ is the analogue of the gelation time $T_{\mathrm{gel}}$, commonly studied in coagulation--fragmentation equations \cite[Definition 9.1.1 and  Equation 
(9.1.1)]{banasiak2019analytic}.

 Obtaining the corresponding energy cascade rate for our equation \eqref{4wave} requires additional analysis, which will be the subject of forthcoming work.
\end{remark}

\begin{remark}\label{RemarkMainTheo}
	Part~(i) of the above theorem implies that, under the constraint \eqref{Theorem1:3}, the density function of the thermal cloud experiences an immediate loss of energy. 
	
	Part~(ii) indicates that, when starting from a regular initial condition with   weaker assumptions (see \eqref{Theorem1:5}), the energy still  cascades to infinity within a finite time interval.

In a companion paper~\cite{staffilani2025condensate}, we investigate the question of condensate growth for the kinetic equation~\eqref{4wave}. A key assumption in~\cite{staffilani2025condensate} is that the initial condition is supported near the origin, whereas in the present work, we consider the opposite scenario namely that the support of the initial condition can be supported outside of the origin.

\end{remark}

\section{Useful notations and a priori estimates}

\subsection{Some useful notations}
In this section, we introduce several notational conventions that will be used consistently throughout the remainder of the paper.

First, we define the set of times at which the total energy is conserved:
\begin{equation} \label{Theta}
	\Theta := \left\{ t \in [0,\infty) \,\middle|\,  \int_{[0,\infty)} \, \mathrm{d}\omega  f(t,\omega)\, \omega \Gamma(\omega)
 = \mathfrak{E} \right\}.
\end{equation}

Next, for any three real numbers \( x, y, z \in \mathbb{R} \), we define the median value among them by
\begin{equation} \label{Mid}
	\mathrm{mid}\{x, y, z\} := \text{the unique element in } \{x, y, z\} \setminus \{\min\{x, y, z\}, \max\{x, y, z\}\}.
\end{equation}

Given three energy variables \( \omega, \omega_1, \omega_2 \), we also define the following shorthand notations for their maximum, minimum, and median values:
\begin{equation} \label{Sec:DDM:6}
	\begin{aligned}
		\omega_{\mathrm{Sup}}(\omega, \omega_1, \omega_2) &:= \max\{\omega, \omega_1, \omega_2\}, \\
		\omega_{\mathrm{Inf}}(\omega, \omega_1, \omega_2) &:= \min\{\omega, \omega_1, \omega_2\}, \\
		\omega_{\mathrm{Med}}(\omega, \omega_1, \omega_2) &:= \mathrm{mid}\{\omega, \omega_1, \omega_2\}.
	\end{aligned}
\end{equation}

\subsection{Some useful formulas} In this subsection, we list several useful weak formulations for the collision operators.  
The proofs of these identities are standard (see~\cite{staffilani2024condensation, staffilani2024energy}); however, we include them here for the sake of completeness.

\begin{lemma}
	\label{lemma:C12}  Assume that Assumption X holds. 
For any suitable test function \( \Xi(\omega) \), we have the following identity:
\begin{equation} \label{Lemma:C12:1}
	\begin{aligned}
		\int_{\mathbb{R}_+} \mathrm{d}\omega\, C_{12}[f](\omega)\, \Xi(\omega)\, \Gamma(\omega)
		=\ & c_{12} \iiint_{\mathbb{R}_+^3} \mathrm{d}\omega_1\, \mathrm{d}\omega_2\, \mathrm{d}\omega\,
		\delta(\omega - \omega_1 - \omega_2)\, \mathfrak{P}\, \mathfrak{P}_1\, \mathfrak{P}_2 \\
		& \times [f_1 f_2 - f(f_1 + f_2)]\, [\Xi(\omega) - \Xi(\omega_1) - \Xi(\omega_2)],
	\end{aligned}
\end{equation}
where \( c_{12} \) is a constant independent of \( f \) and \( \Xi \), and we have used the shorthand notation introduced in~\eqref{Shorthand}.

\end{lemma} 

\begin{proof}
By employing a standard symmetrization technique (see, e.g., \cite{soffer2020energy}), and exchanging the integration variables \( k \leftrightarrow k_1 \) and \( k \leftrightarrow k_2 \), the integral involving the collision operator \( C_{12}[f] \) can be expressed as
\[
\int_{\mathbb{R}^3}\!\mathrm{d}k\, C_{12}[f](k) \,\Xi(\omega) 
= \mathfrak{c}_{12} \iiint_{\mathbb{R}^3 \times \mathbb{R}^3 \times \mathbb{R}^3}\!\mathrm{d}k\, \mathrm{d}k_1\, \mathrm{d}k_2\,  \delta(k - k_1 - k_2) \, \delta(\omega - \omega_1 - \omega_2) \, \bar{\mathfrak{P}}\, \bar{\mathfrak{P}}_1\, \bar{\mathfrak{P}}_2 \,
\]
\[
\times \left[ f(k_1)f(k_2) - f(k_1)f(k) - f(k_2)f(k) \right] 
\left[ \Xi(\omega) - \Xi(\omega_1) - \Xi(\omega_2) \right].
\]

Note that $\bar{\mathfrak{P}}, \bar{\mathfrak{P}}_1, \bar{\mathfrak{P}}_2$ are defined in Assumption X.

Following  \cite{staffilani2024condensation, staffilani2024energy}, we rewrite this integral by separating radial and angular components. Switching to spherical coordinates and integrating over the angular variables yields:
\[
\int_{\mathbb{R}^3}\!\mathrm{d}k\, C_{12}[f](k) \, \Xi(k)
= \mathfrak{c}_{12} \iiint_{\mathbb{R}_+^3}\!\mathrm{d}|k|\, \mathrm{d}|k_1|\, \mathrm{d}|k_2|\,
|k|^2 |k_1|^2 |k_2|^2 \, \delta(\omega - \omega_1 - \omega_2) \left[f_1 f_2 - f_1 f - f_2 f \right]
\]
\[
\times \bar{\mathfrak{P}} \, \bar{\mathfrak{P}}_1 \, \bar{\mathfrak{P}}_2 \left[\Xi(\omega) - \Xi(\omega_1) - \Xi(\omega_2) \right] 
\iiint_{(\mathbb{S}^2)^3}\!\mathrm{d}\mathcal{U} \, \mathrm{d}\mathcal{U}_1 \, \mathrm{d}\mathcal{U}_2 
\left[ \frac{1}{(2\pi)^3} \int_{\mathbb{R}^3}\!\mathrm{d}z\, e^{i z \cdot (k - k_1 - k_2)} \right].
\]

Evaluating the Fourier integral and the angular integrals explicitly, we reduce the expression to  
\[
\int_{\mathbb{R}^3}\!\mathrm{d}k\, C_{12}[f](k) \Xi(k)
= \mathfrak{c}_{12} \iiint_{\mathbb{R}_+^3}\!\mathrm{d}|k|\, \mathrm{d}|k_1|\, \mathrm{d}|k_2|\,
|k|^2 |k_1|^2 |k_2|^2 \, \delta(\omega - \omega_1 - \omega_2) \left[f_1 f_2 - f_1 f - f_2 f \right]
\]
\[
\times \bar{\mathfrak{P}} \, \bar{\mathfrak{P}}_1 \, \bar{\mathfrak{P}}_2 \left[\Xi(\omega) - \Xi(\omega_1) - \Xi(\omega_2) \right]
\cdot 32 \pi \int_0^\infty \frac{\sin(|k_1| y) \sin(|k_2| y) \sin(|k| y)}{y} \, \mathrm{d}y.
\]

Observing that \( |k| < |k_1| + |k_2| \), we apply the integral identity
\[
\int_0^\infty \frac{\sin(|k_1| y) \, \sin(|k_2| y) \, \sin(|k| y)}{y} \, \mathrm{d}y = \frac{\pi}{4},
\]
and obtain
\begin{equation}\label{Le:C12:E2}
	\begin{aligned}
		\int_{\mathbb{R}_+}\!\mathrm{d}\omega\, C_{12}[f](\omega) \, \Xi(\omega) \, \Gamma(\omega)
		=\ & c_{12} \iiint_{\mathbb{R}_+^3}\!\mathrm{d}|k|\, \mathrm{d}|k_1|\, \mathrm{d}|k_2|\,
		|k|^2 |k_1|^2 |k_2|^2 \, \delta(\omega - \omega_1 - \omega_2) \, \bar{\mathfrak{P}} \, \bar{\mathfrak{P}}_1 \, \bar{\mathfrak{P}}_2 \\
		& \times \big[ f_1 f_2 - f_1 f - f_2 f \big] \big[ \Xi(\omega) - \Xi(\omega_1) - \Xi(\omega_2) \big],
	\end{aligned}
\end{equation}
where \( c_{12} > 0 \) is a universal constant.

By performing the change of variables \( |k| \mapsto \omega \), \( |k_1| \mapsto \omega_1 \), and \( |k_2| \mapsto \omega_2 \), the integral in \eqref{Le:C12:E2} becomes
\begin{equation*}
	\begin{aligned}
		\int_{\mathbb{R}_+}\!\mathrm{d}\omega\, C_{12}[f](\omega) \, \Xi(\omega) \, \Gamma(\omega)
		= &\ c_{12} \iiint_{\mathbb{R}_+^3}\!\mathrm{d}\omega\, \mathrm{d}\omega_1\, \mathrm{d}\omega_2\, \delta(\omega - \omega_1 - \omega_2) \, {\mathfrak{P}} \, {\mathfrak{P}}_1 \, {\mathfrak{P}}_2 \\
		& \times \big[ f_1 f_2 - f (f_1 + f_2) \big] \big[ \Xi(\omega) - \Xi(\omega_1) - \Xi(\omega_2) \big],
	\end{aligned}
\end{equation*}
which completes the proof.

\end{proof}

\begin{lemma}
	\label{lemma:C22} Assume that Assumption X holds. 
	For any suitable test function \(\Xi(\omega)\), the following identity holds:
	\begin{equation}\label{Lemma:C22:1}
		\begin{aligned}
			\int_{\mathbb{R}_+} \mathrm{d}\omega\, C_{22}[f](\omega)\, \Xi(\omega)\, \Gamma(\omega) 
			=\, & c_{22} \iiiint_{\mathbb{R}_+^4} \mathrm{d}\omega_1\, \mathrm{d}\omega_2\, \mathrm{d}\omega_3\, \mathrm{d}\omega\, \delta(\omega + \omega_1 - \omega_2 - \omega_3) \\
			& \times {\mathfrak{R}} \, {\mathfrak{R}}_1 \, {\mathfrak{R}}_2 \, {\mathfrak{R}}_3 \, \min\{|k|, |k_1|, |k_2|, |k_3|\} \bar{\mathfrak{R}}_o(\omega,\omega_1,\omega_2,\omega+\omega_1-\omega_2)\\
			& \times f(\omega) f(\omega_1) f(\omega_2) \left[ -\Xi(\omega) - \Xi(\omega_1) + \Xi(\omega_2) + \Xi(\omega_3) \right],
		\end{aligned}
	\end{equation}
	where \(c_{22} > 0\) is a constant independent of \(f\) and \(\Xi\), and the shorthand notation follows \eqref{Shorthand}.

\end{lemma} 
\begin{proof}
Building on the approach in \cite{staffilani2024condensation, staffilani2024energy}, we express the collision operator \( C_{22}[f] \) as
\[
\begin{aligned}
	C_{22}[f](k) 
	=\, & \mathfrak{c}_{22} \iiint_{\mathbb{R}^9} \mathrm{d}k_1\, \mathrm{d}k_2\, \mathrm{d}k_3 \, \delta(\omega + \omega_1 - \omega_2 - \omega_3) \, \delta(k + k_1 - k_2 - k_3) \\
	& \times \left[- f f_1 (f_2 + f_3) + f_2 f_3 (f_1 + f) \right] \bar{\mathfrak{R}} \bar{\mathfrak{R}}_1 \bar{\mathfrak{R}}_2 \bar{\mathfrak{R}}_3\bar{\mathfrak{R}}_o(\omega,\omega_1,\omega_2,\omega+\omega_1-\omega_2),
\end{aligned}
\]
where \(\omega = \omega(k)\) and \(\omega_i = \omega(k_i)\). Note that $\bar{\mathfrak{R}}, \bar{\mathfrak{R}}_1, \bar{\mathfrak{R}}_2, \bar{\mathfrak{R}}_3$ are defined in Assumption X.

Passing to spherical coordinates and separating radial and angular variables, we write
\[
\begin{aligned}
	C_{22}[f](k) 
	=\, & \mathfrak{c}_{22} \iiint_{\mathbb{R}_+^3} \mathrm{d}|k_1|\, \mathrm{d}|k_2|\, \mathrm{d}|k_3| \, |k_1|^2 |k_2|^2 |k_3|^2 \, \delta(\omega + \omega_1 - \omega_2 - \omega_3) \\
	& \times \iiint_{(\mathbb{S}^2)^3} \mathrm{d}\mathcal{U}_1\, \mathrm{d}\mathcal{U}_2\, \mathrm{d}\mathcal{U}_3 \, \left[- f f_1 (f_2 + f_3) + f_2 f_3 (f_1 + f) \right] \\
	& \times \left[ \frac{1}{(2\pi)^3} \int_{\mathbb{R}^3} \mathrm{d}z\, e^{i z \cdot (k + k_1 - k_2 - k_3)} \right] \bar{\mathfrak{R}} \bar{\mathfrak{R}}_1 \bar{\mathfrak{R}}_2 \bar{\mathfrak{R}}_3\bar{\mathfrak{R}}_o.
\end{aligned}
\]

Evaluating the angular integrals and the Fourier transform yields
\[
\begin{aligned}
	C_{22}[f](k) 
	=\, & \mathfrak{c}_{22} \iiint_{\mathbb{R}_+^3} \mathrm{d}|k_1|\, \mathrm{d}|k_2|\, \mathrm{d}|k_3| \, \frac{32 \pi}{|k|} \int_0^\infty \mathrm{d}y \, \frac{\sin(|k| y) \sin(|k_1| y) \sin(|k_2| y) \sin(|k_3| y)}{y^2} \\
	& \times \delta(\omega + \omega_1 - \omega_2 - \omega_3) \, \left[- f f_1 (f_2 + f_3) + f_2 f_3 (f_1 + f) \right] |k_1| |k_2| |k_3| \\
	& \times \bar{\mathfrak{R}} \bar{\mathfrak{R}}_1 \bar{\mathfrak{R}}_2 \bar{\mathfrak{R}}_3\bar{\mathfrak{R}}_o.
\end{aligned}
\]

Using the energy conservation relation \(\omega + \omega_1 = \omega_2 + \omega_3\) and \eqref{X1:1}, the oscillatory integral evaluates explicitly as (see, e.g., \cite{staffilani2024energy})
\[
\int_0^\infty \frac{\sin(|k| y) \sin(|k_1| y) \sin(|k_2| y) \sin(|k_3| y)}{y^2} \, \mathrm{d}y = \frac{\pi}{4} \min\{ |k|, |k_1|, |k_2|, |k_3| \}.
\]

Substituting this back, we obtain
\[
\begin{aligned}
	\int_{\mathbb{R}_+} \mathrm{d}\omega\, C_{22}[f](\omega) \, \Xi(\omega) \, \Gamma(\omega) 
	=\, & c_{22} \iiiint_{\mathbb{R}_+^4} \mathrm{d}|k|\, \mathrm{d}|k_1|\, \mathrm{d}|k_2|\, \mathrm{d}|k_3| \, |k| |k_1| |k_2| |k_3| \min\{ |k|, |k_1|, |k_2|, |k_3| \} \\
	& \times \delta(\omega + \omega_1 - \omega_2 - \omega_3) \, \bar{\mathfrak{R}} \bar{\mathfrak{R}}_1 \bar{\mathfrak{R}}_2 \bar{\mathfrak{R}}_3\bar{\mathfrak{R}}_o \\
	& \times \left[- f f_1 (f_2 + f_3) + f_2 f_3 (f_1 + f) \right] \Xi(\omega),
\end{aligned}
\]
for some universal constant \( c_{22} > 0 \).

Finally, applying the change of variables \( |k| \mapsto \omega \), \( |k_i| \mapsto \omega_i \) for \(i=1,2,3\), and performing a standard symmetrization argument by swapping variables \(k \leftrightarrow k_1\), \(k \leftrightarrow k_2\), and \(k \leftrightarrow k_3\) (cf. \cite{soffer2018dynamics}), we arrive at the weak formulation stated in \eqref{Lemma:C22:1}.

\end{proof}

\begin{lemma}
	\label{lemma:C31}  Assume that Assumption X holds. 
For any suitable test function \(\Xi(\omega)\), we have the following equation:
\begin{equation}\label{Lemma:C31:1}
	\begin{aligned}
		& \int_{\mathbb{R}_+} \mathrm{d}\omega\, C_{31}[f](\omega)\, \Xi(\omega)\, \Gamma(\omega) \\
		=\, & c_{31} \iiiint_{\mathbb{R}_+^4} \mathrm{d}\omega_1\, \mathrm{d}\omega_2\, \mathrm{d}\omega_3\, \mathrm{d}\omega \,
		\delta(\omega - \omega_1 - \omega_2 - \omega_3) \, 
		{\mathfrak{Q}} \, {\mathfrak{Q}}_1 \, {\mathfrak{Q}}_2 \, {\mathfrak{Q}}_3 \\
		& \quad \times [f_1 f_2 f_3 - f(f_1 f_2 + f_2 f_3 + f_1 f_3)] \, 
		[\Xi(\omega) - \Xi(\omega_1) - \Xi(\omega_2) - \Xi(\omega_3)],
	\end{aligned}
\end{equation}
where \( c_{31} \) is a constant independent of \(f\) and \(\Xi\), and the shorthand notation follows from \eqref{Shorthand}.

\end{lemma} 
\begin{proof}
Analogously to Lemma~\ref{lemma:C12}, by applying the variable exchanges \( k \leftrightarrow k_1 \), \( k \leftrightarrow k_2 \), and \( k \leftrightarrow k_3 \), we obtain the representation
\[
\begin{aligned}
	\int_{\mathbb{R}^3} \mathrm{d}k\, C_{31}[f](k)\, \Xi(k) 
	= \mathfrak{c}_{31} \iiiint_{\mathbb{R}^3 \times \mathbb{R}^3 \times \mathbb{R}^3 \times \mathbb{R}^3} & \mathrm{d}k\, \mathrm{d}k_1\, \mathrm{d}k_2\, \mathrm{d}k_3\, \delta(k - k_1 - k_2 - k_3) \\
	& \times \bar{\mathfrak{Q}} \bar{\mathfrak{Q}}_1 \bar{\mathfrak{Q}}_2 \bar{\mathfrak{Q}}_3 \left| |k_1| + |k_2| + |k_3| - |k| \right|^{-1} \\
	& \times \left[ f_1 f_2 f_3 - f(f_1 f_2 + f_2 f_3 + f_3 f_1) \right] \\
	& \times \delta(\omega - \omega_1 - \omega_2 - \omega_3) \\
	& \times \left[ \Xi(\omega) - \Xi(\omega_1) - \Xi(\omega_2) - \Xi(\omega_3) \right].
\end{aligned}
\]

We then express this integral in spherical coordinates and separate the angular parts:
\[
\begin{aligned}
	\int_{\mathbb{R}^3} \mathrm{d}k\, C_{31}[f](k)\, \Xi(k) 
	= \mathfrak{c}_{31} \iiiint_{\mathbb{R}_+^4} & \mathrm{d}|k|\, \mathrm{d}|k_1|\, \mathrm{d}|k_2|\, \mathrm{d}|k_3|\, |k|^2 |k_1|^2 |k_2|^2 |k_3|^2 \, \bar{\mathfrak{Q}} \bar{\mathfrak{Q}}_1 \bar{\mathfrak{Q}}_2 \bar{\mathfrak{Q}}_3 \\
	& \times \delta(k - k_1 - k_2 - k_3) \, \delta(\omega - \omega_1 - \omega_2 - \omega_3) \, \left| |k_1| + |k_2| + |k_3| - |k| \right|^{-1} \\
	& \times \left[ f_1 f_2 f_3 - f(f_1 f_2 + f_2 f_3 + f_3 f_1) \right] \\
	& \times \left[ \Xi(\omega) - \Xi(\omega_1) - \Xi(\omega_2) - \Xi(\omega_3) \right] \\
	& \times \iiint_{(\mathbb{S}^2)^3} \mathrm{d}\mathcal{U}_1 \, \mathrm{d}\mathcal{U}_2 \, \mathrm{d}\mathcal{U}_3 \left[ \frac{1}{(2\pi)^3} \int_{\mathbb{R}^3} \mathrm{d}z\, e^{iz \cdot (k - k_1 - k_2 - k_3)} \right].
\end{aligned}
\]

Evaluating the Fourier integral and integrating over angles yields a constant \(\mathfrak{c}_{31}'>0\) such that
\[
\begin{aligned}
	\int_{\mathbb{R}^3} \mathrm{d}k\, C_{31}[f](k)\, \Xi(k)
	= \mathfrak{c}_{31}' \iiiint_{\mathbb{R}_+^4} & \mathrm{d}|k|\, \mathrm{d}|k_1|\, \mathrm{d}|k_2|\, \mathrm{d}|k_3|\, |k| |k_1| |k_2| |k_3| \, \bar{\mathfrak{Q}} \bar{\mathfrak{Q}}_1 \bar{\mathfrak{Q}}_2 \bar{\mathfrak{Q}}_3 \\
	& \times \left| |k_1| + |k_2| + |k_3| - |k| \right|^{-1} \, \delta(\omega - \omega_1 - \omega_2 - \omega_3) \\
	& \times \left[ f_1 f_2 f_3 - f(f_1 f_2 + f_2 f_3 + f_3 f_1) \right] \\
	& \times \left[ \Xi(\omega) - \Xi(\omega_1) - \Xi(\omega_2) - \Xi(\omega_3) \right] \\
	& \times \int_0^\infty \mathrm{d}y\, \frac{\sin(|k_1| y) \sin(|k_2| y) \sin(|k_3| y) \sin(|k| y)}{y^2}.
\end{aligned}
\]

By the resonance condition \(\omega = \omega_1 + \omega_2 + \omega_3\), it follows that
\[
|k| < |k_1| + |k_2| + |k_3|.
\]

The oscillatory integral can be computed explicitly as
\[
\begin{aligned}
	\int_0^\infty \mathrm{d}y\, \frac{\sin(|k_1| y) \sin(|k_2| y) \sin(|k_3| y) \sin(|k| y)}{y^2} 
	= & \frac{\pi}{16} \sum_{x_1,x_2,x_3,x_4=1}^2  (-1)^{x_1 + x_2 + x_3 + x_4 + 1} \\
	& \times \left| (-1)^{x_1} |k_1| + (-1)^{x_2} |k_2| + (-1)^{x_3} |k_3| + (-1)^{x_4} |k| \right| \\
	= & \frac{\pi}{8} \left[ |k_1| + |k_2| + |k_3| - |k| \right].
\end{aligned}
\]

Substituting back, we arrive at
\[
\begin{aligned}
	\int_{\mathbb{R}_+} \mathrm{d}\omega\, C_{31}[f](\omega)\, \Xi(\omega)\, \Gamma(\omega)
	= c_{31} \iiiint_{\mathbb{R}_+^4} & \mathrm{d}|k|\, \mathrm{d}|k_1|\, \mathrm{d}|k_2|\, \mathrm{d}|k_3|\, |k| |k_1| |k_2| |k_3| \, \bar{\mathfrak{Q}} \bar{\mathfrak{Q}}_1 \bar{\mathfrak{Q}}_2 \bar{\mathfrak{Q}}_3 \\
	& \times \delta(\omega - \omega_1 - \omega_2 - \omega_3) \\
	& \times \left[ f_1 f_2 f_3 - f(f_1 f_2 + f_2 f_3 + f_3 f_1) \right] \\
	& \times \left[ \Xi(\omega) - \Xi(\omega_1) - \Xi(\omega_2) - \Xi(\omega_3) \right],
\end{aligned}
\]
for some universal positive constant \(c_{31} > 0\). Applying the change of variables \(|k| \mapsto \omega\), \(|k_1| \mapsto \omega_1\), \(|k_2| \mapsto \omega_2\), and \(|k_3| \mapsto \omega_3\) recovers the weak formulation \eqref{Lemma:C31:1}.

\end{proof}
\subsection{Some a priori estimates}

\begin{lemma}
	\label{lemma:Apriori}  Assume that Assumption X holds.  
Let $f$ be a radial solution, in the sense of \eqref{4wavemild}, to the wave kinetic equation \eqref{4wave}. Then, for all $T^*>T > 0$, and for all $R>0$, the following estimate holds, for $\alpha$ defined in Assumption X:
\begin{equation}\label{Lemma:Apriori:1}
	\begin{aligned}
		\mathfrak{M} + \mathfrak{E} \ \ge\ & c_{12} \int_0^T \mathrm{d}t \iint_{[R,\infty)^2} \mathrm{d}\omega_1\, \mathrm{d}\omega_2\, 
		\mathfrak{P}_1\, \mathfrak{P}_2\, f_1 f_2\, \mathfrak{P}(\omega_1 + \omega_2)\, 
		\frac{\alpha(1-\alpha)\omega_1 \omega_2}{(\omega_1 + \omega_2)^{2-\alpha}} \\
		& + c_{22} \int_0^T \mathrm{d}t \iiint_{[R,\infty)^3} \mathrm{d}\omega_1\, \mathrm{d}\omega_2\, \mathrm{d}\omega\, 
		f_1 f_2 f\, \frac{\alpha(1-\alpha)(\omega_{\mathrm{Med}} - \omega_{\mathrm{Inf}})^2}{(2\omega_{\mathrm{Med}} - \omega_{\mathrm{Inf}})^{2-\alpha}}\, |k_{\mathrm{Inf}}| \\
		& \quad \times 	\mathfrak{R}(\omega_{\mathrm{Sup}})\, 	\mathfrak{R}(\omega_{\mathrm{Inf}})\, 	\mathfrak{R}(\omega_{\mathrm{Med}})\, 
			\mathfrak{R}(\omega_{\mathrm{Sup}} - \omega_{\mathrm{Inf}} + \omega_{\mathrm{Med}}) \\
		& \quad \times 	\mathfrak{R}_o(\omega_{\mathrm{Sup}},\, 	\omega_{\mathrm{Inf}}, 	\omega_{\mathrm{Med}}, 
		\omega_{\mathrm{Sup}} - \omega_{\mathrm{Inf}} + \omega_{\mathrm{Med}}) \\
		& + c_{31} \int_0^T \mathrm{d}t \iiint_{[R,\infty)^3} 
		\mathrm{d}\omega_1\, \mathrm{d}\omega_2\, \mathrm{d}\omega_3\, 
		\mathfrak{Q}(\omega_1 + \omega_2 + \omega_3)\, \mathfrak{Q}(\omega_1)\, \mathfrak{Q}(\omega_2)\, \mathfrak{Q}(\omega_3) \\
		& \quad \times f(\omega_1) f(\omega_2) f(\omega_3)\,
		\alpha(1-\alpha)\frac{\omega_1 \omega_2 + \omega_2 \omega_3 + \omega_3 \omega_1}{3(\omega_1 + \omega_2 + \omega_3)^{2-\alpha}}.
	\end{aligned}
\end{equation}

Here, $k_{\mathrm{Inf}}$ is the wavenumber associated with $\omega_{\mathrm{Inf}}$. All other notations are consistent with those introduced in \eqref{Sec:DDM:6}.

\end{lemma}

\begin{proof} 
We choose the test function to be  
\[
\Xi(\omega) = ((\omega - R)_+ + 1)^\alpha = (\max\{\omega - R, 0\} + 1)^\alpha,
\]  with $\alpha\in(0,1)$. 
By Lemmas \ref{lemma:C12}, \ref{lemma:C22}, and \ref{lemma:C31}, we compute

\begin{equation}\label{Lemma:Apriori:E1}
	\begin{aligned}
		& \int_{\mathbb{R}_+} \mathrm{d}\omega\, \partial_t f(\omega)\, \Xi(\omega)\Gamma \\
		=\ & \int_{\mathbb{R}_+} \mathrm{d}\omega\, C_{12}[f](\omega)\, \Xi(\omega) 	\Gamma 
		+ \int_{\mathbb{R}_+} \mathrm{d}\omega\, C_{22}[f](\omega)\, \Xi(\omega) \Gamma + \int_{\mathbb{R}_+} \mathrm{d}\omega\, C_{31}[f](\omega)\, \Xi(\omega) 	\Gamma \\
		=\ & c_{12} \iiint_{\mathbb{R}_+^3} \mathrm{d}\omega_1 \mathrm{d}\omega_2 \mathrm{d}\omega\, 
		\delta(\omega - \omega_1 - \omega_2)\, \mathfrak{P}\,\mathfrak{P}_1\, \mathfrak{P}_2\, \\
		& \qquad \times [f_1 f_2 - f(f_1 + f_2)]\, [\Xi - \Xi_1 - \Xi_2] \\
		& + c_{22} \iiiint_{\mathbb{R}_+^4} \mathrm{d}\omega_1 \mathrm{d}\omega_2 \mathrm{d}\omega_3 \mathrm{d}\omega\, 
		\delta(\omega + \omega_1 - \omega_2 - \omega_3) \\
		& \qquad \times 	\mathfrak{R}\,\mathfrak{R}_1\, \mathfrak{R}_2\,\mathfrak{R}_3\,
		\min\{|k|, |k_1|, |k_2|, |k_3|\}\, f_1 f_2 f\, [-\Xi - \Xi_1 + \Xi_2 + \Xi_3] \\
		& + c_{31} \iiiint_{\mathbb{R}_+^4} \mathrm{d}\omega_1 \mathrm{d}\omega_2 \mathrm{d}\omega_3 \mathrm{d}\omega\, 
		\delta(\omega - \omega_1 - \omega_2 - \omega_3) 		\mathfrak{Q}\,\mathfrak{Q}_1\, \mathfrak{Q}_2\,\mathfrak{Q}_3\, \\
		& \qquad \times 
	 \,  [f_1 f_2 f_3 - f(f_1 f_2 + f_2 f_3 + f_1 f_3)]\, 
		[\Xi - \Xi_1 - \Xi_2 - \Xi_3]\,.
	\end{aligned}
\end{equation}

	We  divide the remainder of the proof into several steps.

	{\it Step 1:} 
First, we rewrite the first term on the right-hand side of \eqref{Lemma:Apriori:E1} as

\begin{equation*}
	\begin{aligned}
		& \int_{\mathbb{R}_+} \mathrm{d}\omega\, C_{12}[f]\, \Xi(\omega)\, \Gamma(\omega) \\
		=\, & c_{12} \iint_{\mathbb{R}_+^2} \mathrm{d}\omega_1\, \mathrm{d}\omega_2\, 
		\mathfrak{P}(\omega_1 + \omega_2)\, \mathfrak{P}(\omega_1)\, \mathfrak{P}(\omega_2) \\
		& \quad \times \left[ f(\omega_1) f(\omega_2) - f(\omega_1 + \omega_2)\big(f(\omega_1) + f(\omega_2)\big) \right]  \left[ \Xi(\omega_1 + \omega_2) - \Xi(\omega_1) - \Xi(\omega_2) \right],
	\end{aligned}
\end{equation*}

which, after rearranging terms and performing the change of variables 
$(\omega_1 + \omega_2, \omega_1) \to (\omega_1, \omega_1 - \omega_2)$, can be re-expressed as

\begin{equation}\label{Lemma:Apriori:E8a}
	\begin{aligned}
		& \int_{\mathbb{R}_+} \mathrm{d}\omega\, C_{12}[f]\, \Xi(\omega)\, \Gamma(\omega) \\
		=\, & 2c_{12} \iint_{\omega_1 > \omega_2} \mathrm{d}\omega_1\, \mathrm{d}\omega_2\, 
		\mathfrak{P}(\omega_1)\, \mathfrak{P}(\omega_2)\, f(\omega_1) f(\omega_2) \\
		& \quad \times \left[ \mathfrak{P}(\omega_1 + \omega_2) \left( \Xi(\omega_1 + \omega_2) - \Xi(\omega_1) - \Xi(\omega_2) \right)\ -\ \mathfrak{P}(\omega_1 - \omega_2) \left( \Xi(\omega_1) - \Xi(\omega_1 - \omega_2) - \Xi(\omega_2) \right) \right] \\
		& + c_{12} \iint_{\omega_1 = \omega_2} \mathrm{d}\omega_1\, \mathrm{d}\omega_2\, 
		\mathfrak{P}(\omega_1)\, \mathfrak{P}(\omega_2)\, f(\omega_1) f(\omega_2)\, 
		\mathfrak{P}(2\omega_1) \left[ \Xi(2\omega_1) - 2\Xi(\omega_1) \right],
	\end{aligned}
\end{equation}

where we have used the fact that $\mathfrak{P}(0) = 0$.

Since $\Xi(0) = 0$, we write
\[
\mathfrak{P}(\omega_1 + \omega_2) \left( \Xi(\omega_1 + \omega_2) - \Xi(\omega_1) - \Xi(\omega_2) \right) 
= \mathfrak{P}(\omega_1 + \omega_2) \int_0^{\omega_2} \int_0^{\omega_1} \mathrm{d}\sigma\, \mathrm{d}\sigma_0\, \Xi''(\sigma + \sigma_0),
\]
and
\[
\mathfrak{P}(\omega_1 - \omega_2) \left( \Xi(\omega_1) - \Xi(\omega_1 - \omega_2) - \Xi(\omega_2) \right) 
= \mathfrak{P}(\omega_1 - \omega_2) \int_0^{\omega_2} \int_0^{\omega_1 - \omega_2} \mathrm{d}\sigma\, \mathrm{d}\sigma_0\, \Xi''(\sigma + \sigma_0),
\]
where $\Xi''$ is defined in the weak sense. We then deduce
\begin{equation}\label{Lemma:Apriori:E8b}
	\begin{aligned}
		& \mathfrak{P}(\omega_1 + \omega_2) \left( \Xi(\omega_1 + \omega_2) - \Xi(\omega_1) - \Xi(\omega_2) \right) \  -\ \mathfrak{P}(\omega_1 - \omega_2) \left( \Xi(\omega_1) - \Xi(\omega_1 - \omega_2) - \Xi(\omega_2) \right) \\
		=\ & \left[\mathfrak{P}(\omega_1 + \omega_2) - \mathfrak{P}(\omega_1 - \omega_2)\right] 
		\int_0^{\omega_2} \int_0^{\omega_1 - \omega_2} \mathrm{d}\sigma\, \mathrm{d}\sigma_0\, \Xi''(\sigma + \sigma_0) \\
		& +\ \mathfrak{P}(\omega_1 + \omega_2) \int_0^{\omega_2} \int_{\omega_1 - \omega_2}^{\omega_1} \mathrm{d}\sigma\, \mathrm{d}\sigma_0\, \Xi''(\sigma + \sigma_0),
	\end{aligned}
\end{equation}
which implies
	\begin{equation}\label{Lemma:Apriori:E9}
		\begin{aligned}
			& \int_{\mathbb{R}_+} \mathrm{d}\omega\, C_{12}[f]\, \Xi(\omega)\, \Gamma(\omega) \\
			=\ & 2c_{12} \iint_{\omega_1 > \omega_2} \mathrm{d}\omega_1\, \mathrm{d}\omega_2\, \mathfrak{P}(\omega_1)\, \mathfrak{P}(\omega_2)\, f(\omega_1) f(\omega_2) \\
			& \quad \times \left\{ [\mathfrak{P}(\omega_1 + \omega_2) - \mathfrak{P}(\omega_1 - \omega_2)] \int_0^{\omega_2} \int_0^{\omega_1 - \omega_2} \mathrm{d}s\, \mathrm{d}s_0\, \Xi''(\sigma + \sigma_0) \right. \\
			& \qquad \left. + \ \mathfrak{P}(\omega_1 + \omega_2) \int_0^{\omega_2} \int_{\omega_1 - \omega_2}^{\omega_1} \mathrm{d}s\, \mathrm{d}s_0\, \Xi''(\sigma + \sigma_0) \right\} \\
			& + c_{12} \iint_{\omega_1 = \omega_2} \mathrm{d}\omega_1\, \mathrm{d}\omega_2\, \mathfrak{P}(\omega_1)\, \mathfrak{P}(\omega_2)\, f(\omega_1) f(\omega_2)\, \mathfrak{P}(2\omega_1) \int_0^{\omega_1} \int_0^{\omega_1} \mathrm{d}\sigma\, \mathrm{d}\sigma_0\, \Xi''(\sigma + \sigma_0).
		\end{aligned}
	\end{equation}
Noting that 
\[
\Xi''(\sigma) = -\frac{\alpha(1-\alpha)}{\big((\sigma - R)_+ + 1\big)^{2-\alpha}} \quad \text{for } \sigma \in (R, \infty),
\]
and
\[
\Xi''(\sigma) = 0\quad \text{for } \sigma < R,
\]
we deduce from \eqref{Lemma:Apriori:E9} that

\begin{equation}\label{Lemma:Apriori:E10}
	\begin{aligned}
	&	  \int_{\mathbb{R}_+} \mathrm{d}\omega\, C_{12}[f]\, \Xi(\omega)\, \Gamma(\omega) \\ 
		\le\ & -2c_{12} \iint_{\omega_1 > \omega_2} \mathrm{d}\omega_1\, \mathrm{d}\omega_2\, \mathfrak{P}_1\, \mathfrak{P}_2\, f_1 f_2  \left\{ \left[\mathfrak{P}(\omega_1 + \omega_2) - \mathfrak{P}(\omega_1 - \omega_2)\right] 
		\cdot \frac{\alpha(1-\alpha)\omega_2(\omega_1 - \omega_2)}{(\omega_1)^{2-\alpha}}\, \chi_{(R,\infty)}(\omega_1)\right. \\
		& \quad\quad \left. +\ \mathfrak{P}(\omega_1 + \omega_2) 
		\cdot \frac{\omega_2 \omega_1}{(\omega_1 + \omega_2)^{2}}\, \chi_{(R,\infty)}(\omega_1 + \omega_2) \right\} \\
		& - c_{12} \iint_{\omega_1 = \omega_2} \mathrm{d}\omega_1\, \mathrm{d}\omega_2\, \mathfrak{P}_1\, \mathfrak{P}_2\, f_1 f_2\, 
		\mathfrak{P}(2\omega_1) \cdot \frac{\alpha(1-\alpha)\omega_1^2}{(2\omega_1)^{2-\alpha}} \chi_{(R,\infty)}(2\omega_1)\\ 
		\le\ & - c_{12} \iint_{\omega_1 \ge \omega_2\ge R} \mathrm{d}\omega_1\, \mathrm{d}\omega_2\, \mathfrak{P}_1\, \mathfrak{P}_2\, f_1 f_2  \left\{ \left[\mathfrak{P}(\omega_1 + \omega_2) - \mathfrak{P}(\omega_1 - \omega_2)\right] 
		\cdot \frac{\alpha(1-\alpha)\omega_2(\omega_1 - \omega_2)}{\omega_1^{2-\alpha}}\,\right. \\
		& \quad\quad \left. +\ \mathfrak{P}(\omega_1 + \omega_2) 
		\cdot \frac{\alpha(1-\alpha)\omega_2 \omega_1}{(\omega_1 + \omega_2)^{2-\alpha}}\, \right\},
	\end{aligned}
\end{equation}
where \(\mathfrak{P}_i := \mathfrak{P}(\omega_i)\) and \(f_i := f(\omega_i)\).

	{\it Step 2:} 
The second term on the right-hand side of \eqref{Lemma:Apriori:E1} can be rewritten as
\begin{equation*}
	\begin{aligned}
		& \int_{\mathbb{R}_+} \mathrm{d}\omega\, C_{22}[f]\, \Xi(\omega)\, \Gamma(\omega) \\
		=\, & c_{22} \iiint_{\mathbb{R}_+^3} \mathrm{d}\omega_1\, \mathrm{d}\omega_2\, \mathrm{d}\omega\, 
			\mathfrak{R}(\omega)	\mathfrak{R}(\omega_1)	\mathfrak{R}(\omega_2)	\mathfrak{R}(\omega + \omega_1 - \omega_2) \bar{\mathfrak{R}}_o(\omega,\omega_1,\omega_2,\omega+\omega_1-\omega_2)
			\\
		& \quad \times \min\big\{ |k|(\omega), |k|(\omega_1), |k|(\omega_2), |k|(\omega + \omega_1 - \omega_2) \big\}\,
		\mathbf{1}_{\omega + \omega_1 - \omega_2 \ge 0}\, f_1 f_2 f \\
		& \quad \times \big[-\Xi(\omega) - \Xi(\omega_1) + \Xi(\omega_2) + \Xi(\omega + \omega_1 - \omega_2)\big], \end{aligned}
\end{equation*}
which implies
\begin{equation}\label{Lemma:Apriori:E2}
	\begin{aligned}
	& \int_{\mathbb{R}_+} \mathrm{d}\omega\, C_{22}[f]\, \Xi(\omega)\, \Gamma(\omega) \\
		=\, & c_{22} \iiint_{\mathbb{R}_+^3} \mathrm{d}\omega_1\, \mathrm{d}\omega_2\, \mathrm{d}\omega\, f_1 f_2 f \\
		& \quad \times \Big\{ 
		\left[-\Xi(\omega_{\text{Sup}}) - \Xi(\omega_{\text{Inf}}) + \Xi(\omega_{\text{Med}}) + \Xi(\omega_{\text{Sup}} + \omega_{\text{Inf}} - \omega_{\text{Med}})\right] \\
		& \qquad \times 	\mathfrak{R}(\omega_{\text{Sup}})	\mathfrak{R}(\omega_{\text{Inf}})	\mathfrak{R}(\omega_{\text{Med}})
			\mathfrak{R}(\omega_{\text{Sup}} + \omega_{\text{Inf}} - \omega_{\text{Med}})\, |k_{\text{Inf}}| \bar{\mathfrak{R}}_o
			\\
		& \quad + \left[-\Xi(\omega_{\text{Sup}}) - \Xi(\omega_{\text{Med}}) + \Xi(\omega_{\text{Inf}}) + \Xi(\omega_{\text{Sup}} + \omega_{\text{Med}} - \omega_{\text{Inf}})\right] \\
		& \qquad \times 	\mathfrak{R}(\omega_{\text{Sup}})	\mathfrak{R}(\omega_{\text{Med}})	\Gamma(\omega_{\text{Inf}})
			\mathfrak{R}(\omega_{\text{Sup}} + \omega_{\text{Med}} - \omega_{\text{Inf}})\, |k_{\text{Inf}}|\bar{\mathfrak{R}}_o\\
		& \quad + \left[-\Xi(\omega_{\text{Inf}}) - \Xi(\omega_{\text{Med}}) + \Xi(\omega_{\text{Sup}}) + \Xi(\omega_{\text{Inf}} + \omega_{\text{Med}} - \omega_{\text{Sup}})\right] \\
		& \qquad \times 	\mathfrak{R}(\omega_{\text{Sup}})	\mathfrak{R}(\omega_{\text{Med}})	\mathfrak{R}(\omega_{\text{Inf}})\,
			\mathfrak{R}(\omega_{\text{Inf}} + \omega_{\text{Med}} - \omega_{\text{Sup}})\, 
		\mathbf{1}_{\omega_{\text{Inf}} + \omega_{\text{Med}} - \omega_{\text{Sup}} \ge 0}\bar{\mathfrak{R}}_o \\
		& \qquad \times \min\big\{ |k|(\omega_{\text{Sup}}), |k|(\omega_{\text{Inf}}), |k|(\omega_{\text{Med}}), 
		|k|(\omega_{\text{Inf}} + \omega_{\text{Med}} - \omega_{\text{Sup}}) \big\} 
		\Big\},
	\end{aligned}
\end{equation}
in which \( |k_{\text{Inf}}| \) is associated with \( \omega_{\text{Inf}} \), \( |k_{\text{Sup}}| \) with \( \omega_{\text{Sup}} \), and \( |k_{\text{Med}}| \) with \( \omega_{\text{Med}} \).

We compute
\begin{equation}\label{Lemma:Apriori:E3}
	\begin{aligned}
		& \left[-\Xi(\omega_{\text{Inf}}) - \Xi(\omega_{\text{Med}}) + \Xi(\omega_{\text{Sup}}) + \Xi(\omega_{\text{Inf}} + \omega_{\text{Med}} - \omega_{\text{Sup}})\right] \\
		=\, & \int_0^{\omega_{\text{Sup}} - \omega_{\text{Inf}}} \mathrm{d}\sigma_1 \int_0^{\omega_{\text{Sup}} - \omega_{\text{Med}}} \mathrm{d}\sigma_2\, 
		\Xi''(\sigma_1 + \sigma_2 + \omega_{\text{Inf}}) \ 
		\le\,  0,
	\end{aligned}
\end{equation}
where we have used the fact that \( \Xi'' \le 0 \).

\vspace{1em}

Next, we observe that
\begin{equation}\label{Lemma:Apriori:E4}
	\begin{aligned}
		& \left[-\Xi(\omega_{\text{Sup}}) - \Xi(\omega_{\text{Inf}}) + \Xi(\omega_{\text{Med}}) + \Xi(\omega_{\text{Sup}} + \omega_{\text{Inf}} - \omega_{\text{Med}})\right] 	\mathfrak{R}(\omega_{\text{Sup}} + \omega_{\text{Inf}} - \omega_{\text{Med}}) \\
		& + \left[-\Xi(\omega_{\text{Sup}}) - \Xi(\omega_{\text{Med}}) + \Xi(\omega_{\text{Inf}}) + \Xi(\omega_{\text{Sup}} + \omega_{\text{Med}} - \omega_{\text{Inf}})\right] 	\mathfrak{R}(\omega_{\text{Sup}} + \omega_{\text{Med}} - \omega_{\text{Inf}}) \\
		=\, & -\int_0^{\omega_{\text{Med}} - \omega_{\text{Inf}}} \mathrm{d}\sigma \int_0^{\omega_{\text{Sup}} - \omega_{\text{Med}}} \mathrm{d}\sigma_0\, 
			\mathfrak{R}(\omega_{\text{Sup}} + \omega_{\text{Inf}} - \omega_{\text{Med}})\, 
		\Xi''(\omega_{\text{Inf}} + \sigma + \sigma_0) \\
		& + \int_0^{\omega_{\text{Med}} - \omega_{\text{Inf}}} \mathrm{d}\sigma \int_0^{\omega_{\text{Sup}} - \omega_{\text{Inf}}} \mathrm{d}\sigma_0\, 
			\mathfrak{R}(\omega_{\text{Sup}} - \omega_{\text{Inf}} + \omega_{\text{Med}})\, 
		\Xi''(\omega_{\text{Inf}} + \sigma + \sigma_0) \\
		=\, & \int_0^{\omega_{\text{Med}} - \omega_{\text{Inf}}} \mathrm{d}\sigma \int_0^{\omega_{\text{Sup}} - \omega_{\text{Med}}} \mathrm{d}\sigma_0\, 
		\Xi''(\omega_{\text{Inf}} + \sigma + \sigma_0) \\
		& \quad \times \left[ 	\mathfrak{R}(\omega_{\text{Sup}} - \omega_{\text{Inf}} + \omega_{\text{Med}}) - 	\mathfrak{R}(\omega_{\text{Sup}} + \omega_{\text{Inf}} - \omega_{\text{Med}}) \right] \\
		& + \int_0^{\omega_{\text{Med}} - \omega_{\text{Inf}}} \mathrm{d}\sigma \int_0^{\omega_{\text{Med}} - \omega_{\text{Inf}}} \mathrm{d}\sigma_0\, 
			\mathfrak{R}(\omega_{\text{Sup}} - \omega_{\text{Inf}} + \omega_{\text{Med}})\, 
		\Xi''(\omega_{\text{Inf}} + \sigma + \sigma_0) \\
		\le\, & \int_0^{\omega_{\text{Med}} - \omega_{\text{Inf}}} \mathrm{d}\sigma \int_0^{\omega_{\text{Med}} - \omega_{\text{Inf}}} \mathrm{d}\sigma_0\, 
			\mathfrak{R}(\omega_{\text{Sup}} - \omega_{\text{Inf}} + \omega_{\text{Med}})\, 
		\Xi''(\omega_{\text{Inf}} + \sigma + \sigma_0) \\
		\le\, & -\int_0^{\omega_{\text{Med}} - \omega_{\text{Inf}}} \mathrm{d}\sigma \int_0^{\omega_{\text{Med}} - \omega_{\text{Inf}}} \mathrm{d}\sigma_0\, 
			\mathfrak{R}(\omega_{\text{Sup}} - \omega_{\text{Inf}} + \omega_{\text{Med}})\, 
		\frac{\alpha(1-\alpha)}{(\omega_{\text{Inf}} + \sigma + \sigma_0 )^{2-\alpha}} \chi_{(R,\infty)}(\omega_{\text{Inf}} + \sigma + \sigma_0) \\
		\le\, & -	\mathfrak{R}(\omega_{\text{Sup}} - \omega_{\text{Inf}} + \omega_{\text{Med}})\, 
		\frac{\alpha(1-\alpha)(\omega_{\text{Med}} - \omega_{\text{Inf}})^2}{(2\omega_{\text{Med}} - \omega_{\text{Inf}})^{2-\alpha }}\, 
		\chi_{(R,\infty)}(2\omega_{\text{Med}} - \omega_{\text{Inf}}).
	\end{aligned}
\end{equation}

Combining \eqref{Lemma:Apriori:E2}, \eqref{Lemma:Apriori:E3}, and \eqref{Lemma:Apriori:E4}, we obtain
\begin{equation}\label{Lemma:Apriori:E5}
	\begin{aligned}
		\int_{\mathbb{R}_+} \mathrm{d}\omega\, C_{22}[f]\, \Xi(\omega)\, \Gamma(\omega)
		\le\ & -c_{22} \iiint_{\mathbb{R}_+^3} \mathrm{d}\omega_1\, \mathrm{d}\omega_2\, \mathrm{d}\omega\, 
		f_1 f_2 f\, \frac{\alpha(1-\alpha)(\omega_{\text{Med}} - \omega_{\text{Inf}})^2}{(2\omega_{\text{Med}} - \omega_{\text{Inf}})^{2-\alpha}} 
		|k_{\text{Inf}}|\bar{\mathfrak{R}}_o \\
		& \times 	\mathfrak{R}(\omega_{\text{Sup}})\, 	\mathfrak{R}(\omega_{\text{Inf}})\, 	\mathfrak{R}(\omega_{\text{Med}})\,
			\mathfrak{R}(\omega_{\text{Sup}} - \omega_{\text{Inf}} + \omega_{\text{Med}})\,
		\chi_{(R,\infty)}(2\omega_{\text{Med}} - \omega_{\text{Inf}}).
	\end{aligned}
\end{equation}

	{\it Step 3:} 
Next, the last term on the right-hand side of \eqref{Lemma:Apriori:E1} can be rewritten as
\begin{equation*}\label{Lemma:Apriori:E11}
	\begin{aligned}
		& \int_{\mathbb{R}_+} \mathrm{d}\omega\, C_{31}[f]\, \Xi(\omega)\, \Gamma(\omega) \\
		=\, & c_{31} \iiint_{\mathbb{R}_+^3} \mathrm{d}\omega_1\, \mathrm{d}\omega_2\, \mathrm{d}\omega_3\, 
		\mathfrak{Q}(\omega_1 + \omega_2 + \omega_3)\, \mathfrak{Q}(\omega_1)\, \mathfrak{Q}(\omega_2)\, \mathfrak{Q}(\omega_3) \\
		& \quad \times \left[ f(\omega_1) f(\omega_2) f(\omega_3) 
		- f(\omega_1 + \omega_2 + \omega_3) \left( f(\omega_1) f(\omega_2) + f(\omega_2) f(\omega_3) + f(\omega_3) f(\omega_1) \right) \right] \\
		& \quad \times \left[ \Xi(\omega_1 + \omega_2 + \omega_3) - \Xi(\omega_1) - \Xi(\omega_2) - \Xi(\omega_3) \right].
	\end{aligned}
\end{equation*}

This expression can be rearranged into the form
\begin{equation*}\label{Lemma:Apriori:E12}
	\begin{aligned}
		& \int_{\mathbb{R}_+} \mathrm{d}\omega\, C_{31}[f]\, \Xi(\omega)\, \Gamma(\omega) \\
		=\, & c_{31}  \iiint_{\mathbb{R}_+^3} \mathrm{d}\omega_1\, \mathrm{d}\omega_2\, \mathrm{d}\omega_3\, 
		\mathfrak{Q}(\omega_1 + \omega_2 + \omega_3)\, \mathfrak{Q}(\omega_1)\, \mathfrak{Q}(\omega_2)\, \mathfrak{Q}(\omega_3)\, 
		f(\omega_1) f(\omega_2) f(\omega_3) \\
		& \quad \times \left[ \Xi(\omega_1 + \omega_2 + \omega_3) - \Xi(\omega_1) - \Xi(\omega_2) - \Xi(\omega_3) \right] \\
		& \quad - 3c_{31} \iiint_{\mathbb{R}_+^3} \mathrm{d}\omega_1\, \mathrm{d}\omega_2\, \mathrm{d}\omega_3\, 
		\mathfrak{Q}(\omega_1 + \omega_2 + \omega_3)\, \mathfrak{Q}(\omega_1)\, \mathfrak{Q}(\omega_2)\, \mathfrak{Q}(\omega_3) \\
		& \qquad \times f(\omega_1 + \omega_2 + \omega_3)\, f(\omega_2)\, f(\omega_3)\, 
		\left[ \Xi(\omega_1 + \omega_2 + \omega_3) - \Xi(\omega_1) - \Xi(\omega_2) - \Xi(\omega_3) \right].
	\end{aligned}
\end{equation*}

Performing the change of variables \( \omega_1 + \omega_2 + \omega_3 \to \omega_1 \), we find
\begin{equation}\label{Lemma:Apriori:E12aa}
	\begin{aligned}
		& \int_{\mathbb{R}_+} \mathrm{d}\omega\, C_{31}[f]\, \Xi(\omega)\, \Gamma(\omega) \\
		=\, & c_{31} \iiint_{\mathbb{R}_+^3} \mathrm{d}\omega_1\, \mathrm{d}\omega_2\, \mathrm{d}\omega_3\, 
		\mathfrak{Q}(\omega_1 + \omega_2 + \omega_3)\, \mathfrak{Q}(\omega_1)\, \mathfrak{Q}(\omega_2)\, \mathfrak{Q}(\omega_3)\, 
		f(\omega_1) f(\omega_2) f(\omega_3) \\
		& \quad \times \left[ \Xi(\omega_1 + \omega_2 + \omega_3) - \Xi(\omega_1) - \Xi(\omega_2) - \Xi(\omega_3) \right] \\
		& \quad - 3c_{31} \iiint_{\omega_1 - \omega_2 - \omega_3 \ge 0} \mathrm{d}\omega_1\, \mathrm{d}\omega_2\, \mathrm{d}\omega_3\, 
		\mathfrak{Q}(\omega_1 - \omega_2 - \omega_3)\, \mathfrak{Q}(\omega_1)\, \mathfrak Q(\omega_2)\, \mathfrak{Q}(\omega_3) \\
		& \qquad \times f(\omega_1) f(\omega_2) f(\omega_3)\, 
		\left[ \Xi(\omega_1) - \Xi(\omega_1 - \omega_2 - \omega_3) - \Xi(\omega_2) - \Xi(\omega_3) \right].
	\end{aligned}
\end{equation}

This implies
\begin{equation}\label{Lemma:Apriori:E12a}
	\begin{aligned}
		& \int_{\mathbb{R}_+} \mathrm{d}\omega\, C_{31}[f]\, \Xi(\omega)\, \Gamma(\omega) \\
		=\, & c_{31} \iiint_{	\{ \omega_1 \ge \omega_2 + \omega_3 \} \cup 
			\{ \omega_2 \ge \omega_1 + \omega_3 \} \cup 
			\{ \omega_3 \ge \omega_1 + \omega_2 \} } \mathrm{d}\omega_1\, \mathrm{d}\omega_2\, \mathrm{d}\omega_3\, 
		\mathfrak{Q}(\omega_1 + \omega_2 + \omega_3)\, \mathfrak{Q}(\omega_1)\, \mathfrak{Q}(\omega_2)\, \mathfrak{Q}(\omega_3)\, 
	\\ & \quad + c_{31} \iiint_{\mathbb{R}_+^3 \setminus \left( 
		\{ \omega_1 \ge \omega_2 + \omega_3 \} \cup 
		\{ \omega_2 \ge \omega_1 + \omega_3 \} \cup 
		\{ \omega_3 \ge \omega_1 + \omega_2 \} 
		\right) } \mathrm{d}\omega_1\, \mathrm{d}\omega_2\, \mathrm{d}\omega_3\, 
	\mathfrak{Q}(\omega_1 + \omega_2 + \omega_3)\, \mathfrak{Q}(\omega_1)\, \mathfrak{Q}(\omega_2)\, \mathfrak{Q}(\omega_3)\, 
	\\
		& \quad \times	f(\omega_1) f(\omega_2) f(\omega_3)  \left[ \Xi(\omega_1 + \omega_2 + \omega_3) - \Xi(\omega_1) - \Xi(\omega_2) - \Xi(\omega_3) \right] \\
		& \quad - 3c_{31} \iiint_{\omega_1 - \omega_2 - \omega_3 \ge 0} \mathrm{d}\omega_1\, \mathrm{d}\omega_2\, \mathrm{d}\omega_3\, 
		\mathfrak{Q}(\omega_1 - \omega_2 - \omega_3)\, \mathfrak{Q}(\omega_1)\, \mathfrak Q(\omega_2)\, \mathfrak{Q}(\omega_3) \\
		& \qquad \times f(\omega_1) f(\omega_2) f(\omega_3)\, 
		\left[ \Xi(\omega_1) - \Xi(\omega_1 - \omega_2 - \omega_3) - \Xi(\omega_2) - \Xi(\omega_3) \right]\\
		=\, & 3c_{31} \iiint_{\omega_1 > \omega_2 + \omega_3} \mathrm{d}\omega_1\, \mathrm{d}\omega_2\, \mathrm{d}\omega_3\, 
		\mathfrak{Q}(\omega_1)\, \mathfrak{Q}(\omega_2)\, \mathfrak{Q}(\omega_3)\, f(\omega_1) f(\omega_2) f(\omega_3) \\
		& \quad \times \left[ 
		\mathfrak{Q}(\omega_1 + \omega_2 + \omega_3) 
		\left( \Xi(\omega_1 + \omega_2 + \omega_3) - \Xi(\omega_1) - \Xi(\omega_2) - \Xi(\omega_3) \right) \right. \\
		& \qquad \left. -\ \mathfrak{Q}(\omega_1 - \omega_2 - \omega_3)\, 
		\left( \Xi(\omega_1) - \Xi(\omega_1 - \omega_2 - \omega_3) - \Xi(\omega_2) - \Xi(\omega_3) \right) 
		\right] \\
		& \quad + 3c_{31} \iiint_{\omega_1 = \omega_2 + \omega_3} \mathrm{d}\omega_1\, \mathrm{d}\omega_2\, \mathrm{d}\omega_3\, 
		\mathfrak{Q}(\omega_1)\, \mathfrak{Q}(\omega_2)\, \mathfrak{Q}(\omega_3)\, f(\omega_1) f(\omega_2) f(\omega_3) \\
		& \qquad \times \mathfrak{Q}(2\omega_1)\, 
		\left[ \Xi(2\omega_1) - \Xi(\omega_1) - \Xi(\omega_2) - \Xi(\omega_3) \right] \\
		& \quad + c_{31} \iiint_{\mathbb{R}_+^3 \setminus \left( 
			\{ \omega_1 \ge \omega_2 + \omega_3 \} \cup 
			\{ \omega_2 \ge \omega_1 + \omega_3 \} \cup 
			\{ \omega_3 \ge \omega_1 + \omega_2 \} 
			\right)} \mathrm{d}\omega_1\, \mathrm{d}\omega_2\, \mathrm{d}\omega_3\, 
		\mathfrak{Q}(\omega_1 + \omega_2 + \omega_3) \\
		& \qquad \times \mathfrak{Q}(\omega_1)\, \mathfrak{Q}(\omega_2)\, \mathfrak{Q}(\omega_3)\, 
		f(\omega_1) f(\omega_2) f(\omega_3)\, 
		\left[ \Xi(\omega_1 + \omega_2 + \omega_3) - \Xi(\omega_1) - \Xi(\omega_2) - \Xi(\omega_3) \right],
	\end{aligned}
\end{equation}
since \( \mathfrak{Q}(0) = 0 \).

Next, since \( \Xi(0) = 0 \), we compute

\begin{equation*}
	\begin{aligned}
		& \mathfrak{Q}(\omega_1 - \omega_2 - \omega_3) 
		\left( \Xi(\omega_1) - \Xi(\omega_1 - \omega_2 - \omega_3) - \Xi(\omega_2) - \Xi(\omega_3) \right) \\
		=\, & \mathfrak{Q}(\omega_1 - \omega_2 - \omega_3) \left( 
		\int_0^{\omega_1 - \omega_2 - \omega_3} \int_0^{\omega_2 + \omega_3} \mathrm{d}\sigma\, \mathrm{d}\sigma_0\, \Xi''(\sigma + \sigma_0)
		+ \int_0^{\omega_2} \int_0^{\omega_3} \mathrm{d}\sigma\, \mathrm{d}\sigma_0\, \Xi''(\sigma + \sigma_0)
		\right),
	\end{aligned}
\end{equation*}

and

\begin{equation*}
	\begin{aligned}
		& \mathfrak{Q}(\omega_1 + \omega_2 + \omega_3) 
		\left( \Xi(\omega_1 + \omega_2 + \omega_3) - \Xi(\omega_1) - \Xi(\omega_2) - \Xi(\omega_3) \right) \\
		=\, & \mathfrak{Q}(\omega_1 + \omega_2 + \omega_3) \left( 
		\int_0^{\omega_1} \int_0^{\omega_2 + \omega_3} \mathrm{d}\sigma\, \mathrm{d}\sigma_0\, \Xi''(\sigma + \sigma_0)
		+ \int_0^{\omega_2} \int_0^{\omega_3} \mathrm{d}\sigma\, \mathrm{d}\sigma_0\, \Xi''(\sigma + \sigma_0)
		\right),
	\end{aligned}
\end{equation*}
which implies
\begin{equation}\label{Lemma:Apriori:E13a}
	\begin{aligned}
		& \mathfrak{Q}(\omega_1 + \omega_2 + \omega_3) 
		\left( \Xi(\omega_1 + \omega_2 + \omega_3) - \Xi(\omega_1) - \Xi(\omega_2) - \Xi(\omega_3) \right) \\
		& \quad - \mathfrak{Q}(\omega_1 - \omega_2 - \omega_3) 
		\left( \Xi(\omega_1) - \Xi(\omega_1 - \omega_2 - \omega_3) - \Xi(\omega_2) - \Xi(\omega_3) \right) \\
		=\, & \left[ \mathfrak{Q}(\omega_1 + \omega_2 + \omega_3) - \mathfrak{Q}(\omega_1 - \omega_2 - \omega_3) \right] \\
		& \quad\times	
		\left( \int_0^{\omega_1 - \omega_2 - \omega_3} \int_0^{\omega_2 + \omega_3} \mathrm{d}\sigma\, \mathrm{d}\sigma_0\, \Xi''(\sigma + \sigma_0) 
		+ \int_0^{\omega_2} \int_0^{\omega_3} \mathrm{d}\sigma\, \mathrm{d}\sigma_0\, \Xi''(\sigma + \sigma_0) \right) \\
		& \quad + \mathfrak{Q}(\omega_1 + \omega_2 + \omega_3) 
		\int_{\omega_1 - \omega_2 - \omega_3}^{\omega_1} \int_0^{\omega_2 + \omega_3} \mathrm{d}\sigma\, \mathrm{d}\sigma_0\, \Xi''(\sigma + \sigma_0).
	\end{aligned}
\end{equation}

	Therefore, we can compute
	\begin{equation}\label{Lemma:Apriori:E13}
		\begin{aligned}
			& \int_{\mathbb{R}_+} \mathrm{d}\omega\, C_{31}[f]\, \Xi(\omega)\, \Gamma(\omega) \\
			=\, & 3c_{31} \iiint_{\omega_1 > \omega_2 + \omega_3} \mathrm{d}\omega_1\, \mathrm{d}\omega_2\, \mathrm{d}\omega_3\, 
			\mathfrak{Q}(\omega_1)\, \mathfrak{Q}(\omega_2)\, \mathfrak{Q}(\omega_3)\, f(\omega_1) f(\omega_2) f(\omega_3) \\
			& \quad \times \Big\{ \left[ \mathfrak{Q}(\omega_1 + \omega_2 + \omega_3) - \mathfrak{Q}(\omega_1 - \omega_2 - \omega_3) \right] \\
			& \qquad \times \left( \int_0^{\omega_1 - \omega_2 - \omega_3} \int_0^{\omega_2 + \omega_3} 
			\mathrm{d}\sigma\, \mathrm{d}\sigma_0\, \Xi''(\sigma + \sigma_0) 
			+ \int_0^{\omega_2} \int_0^{\omega_3} 
			\mathrm{d}\sigma\, \mathrm{d}\sigma_0\, \Xi''(\sigma + \sigma_0) \right) \\
			& \qquad +\ \mathfrak{Q}(\omega_1 + \omega_2 + \omega_3) 
			\int_{\omega_1 - \omega_2 - \omega_3}^{\omega_1} \int_0^{\omega_2 + \omega_3} 
			\mathrm{d}\sigma\, \mathrm{d}\sigma_0\, \Xi''(\sigma + \sigma_0) \Big\} \\
			& + 3c_{31} \iiint_{\omega_1 = \omega_2 + \omega_3} \mathrm{d}\omega_1\, \mathrm{d}\omega_2\, \mathrm{d}\omega_3\, 
			\mathfrak{Q}(\omega_1)\, \mathfrak{Q}(\omega_2)\, \mathfrak{Q}(\omega_3)\, f(\omega_1) f(\omega_2) f(\omega_3)\, \mathfrak{Q}(2\omega_1) \\
			& \quad \times \left( \int_0^{\omega_1} \int_0^{\omega_2 + \omega_3} 
			\mathrm{d}\sigma\, \mathrm{d}\sigma_0\, \Xi''(\sigma + \sigma_0) 
			+ \int_0^{\omega_2} \int_0^{\omega_3} 
			\mathrm{d}\sigma\, \mathrm{d}\sigma_0\, \Xi''(\sigma + \sigma_0) \right) \\
			& + c_{31} \iiint_{\mathbb{R}_+^3 \setminus \left( \{\omega_1 \ge \omega_2 + \omega_3\} \cup \{\omega_2 \ge \omega_1 + \omega_3\} \cup \{\omega_3 \ge \omega_1 + \omega_2\} \right)} 
			\mathrm{d}\omega_1\, \mathrm{d}\omega_2\, \mathrm{d}\omega_3\, 
			\mathfrak{Q}(\omega_1 + \omega_2 + \omega_3) \\
			& \quad \times \mathfrak{Q}(\omega_1)\, \mathfrak{Q}(\omega_2)\, \mathfrak{Q}(\omega_3)\, f(\omega_1) f(\omega_2) f(\omega_3) \\
			& \quad \times \left( \int_0^{\omega_1} \int_0^{\omega_2 + \omega_3} 
			\mathrm{d}\sigma\, \mathrm{d}\sigma_0\, \Xi''(\sigma + \sigma_0) 
			+ \int_0^{\omega_2} \int_0^{\omega_3} 
			\mathrm{d}\sigma\, \mathrm{d}\sigma_0\, \Xi''(\sigma + \sigma_0) \right).
		\end{aligned}
	\end{equation}
	
We now provide an upper bound for
\begin{equation}\label{Lemma:Apriori:E14}
	\begin{aligned}
		& \int_{\mathbb{R}_+} \mathrm{d}\omega\, C_{31}[f]\, \Xi(\omega)\, \Gamma(\omega) \\
		\le\ & -3c_{31}  \iiint_{\omega_1 > \omega_2 + \omega_3} 
		\mathrm{d}\omega_1\, \mathrm{d}\omega_2\, \mathrm{d}\omega_3\, 
		\mathfrak{Q}_1\, \mathfrak{Q}_2\, \mathfrak{Q}_3\, f_1 f_2 f_3 \\
		& \quad \times \left\{ \left[ \mathfrak{Q}(\omega_1 + \omega_2 + \omega_3) 
		- \mathfrak{Q}(\omega_1 - \omega_2 - \omega_3) \right] 
		\frac{\alpha(1-\alpha)(\omega_2+\omega_3)(\omega_1 - \omega_2 - \omega_3)}{\omega_1^{2-\alpha}}\, 
		\chi_{(R,\infty)}(\omega_1) \right. \\
		& \qquad + \left[ \mathfrak{Q}(\omega_1 + \omega_2 + \omega_3) 
		- \mathfrak{Q}(\omega_1 - \omega_2 - \omega_3) \right] 
		\frac{\alpha(1-\alpha)\omega_2 \omega_3}{(\omega_2+\omega_3)^{2-\alpha}}\, 
		\chi_{(R,\infty)}(\omega_2 + \omega_3) \\
		& \qquad \left. +\ \mathfrak{Q}(\omega_1 + \omega_2 + \omega_3)\, 
		\frac{\alpha(1-\alpha)(\omega_2+\omega_3)\, \omega_1}{(\omega_1 + \omega_2 + \omega_3)^{2-\alpha}}\, 
		\chi_{(R,\infty)}(\omega_1 + \omega_2 + \omega_3) \right\} \\
		& \ - 3c_{31} \iiint_{\omega_1 = \omega_2 + \omega_3} 
		\mathrm{d}\omega_1\, \mathrm{d}\omega_2\, \mathrm{d}\omega_3\, 
		\mathfrak{Q}_1\, \mathfrak{Q}_2\, \mathfrak{Q}_3\, f_1 f_2 f_3\, \mathfrak{Q}(2\omega_1)\, 
		\frac{\alpha(1-\alpha)}{(2\omega_1)^\alpha} \chi_{(R,\infty)}(2\omega_1) \\
		& - c_{31} \iiint_{\mathbb{R}_+^3 \setminus \left( \{\omega_1 \ge \omega_2 + \omega_3\} \cup \{\omega_2 \ge \omega_1 + \omega_3\} \cup \{\omega_3 \ge \omega_1 + \omega_2\} \right)} 
		\mathrm{d}\omega_1\, \mathrm{d}\omega_2\, \mathrm{d}\omega_3\, 
		\mathfrak{Q}(\omega_1 + \omega_2 + \omega_3) \\
		& \quad \times \mathfrak{Q}(\omega_1)\, \mathfrak{Q}(\omega_2)\, \mathfrak{Q}(\omega_3)\, f(\omega_1) f(\omega_2) f(\omega_3)
		\frac{\alpha(1-\alpha)(\omega_1\omega_2+\omega_2\omega_3+\omega_3\omega_1)}{(\omega_1 + \omega_2 + \omega_3)^{2-\alpha}}\chi_{(R,\infty)}(\omega_3+\omega_2+\omega_1).
	\end{aligned}
\end{equation}

	Combining \eqref{Lemma:Apriori:E1}, \eqref{Lemma:Apriori:E10}, \eqref{Lemma:Apriori:E5} and \eqref{Lemma:Apriori:E14}, we find
	\begin{equation}\label{Lemma:Apriori:E15}
		\begin{aligned}
			& \int_{\mathbb{R}} \mathrm{d}k\, f(T,\omega)\, \Xi(\omega)
			- \int_{\mathbb{R}} \mathrm{d}k\, f(0,\omega)\, \Xi(\omega) \\
			\le\ & -c_{12} \int_0^T \mathrm{d}t \iint_{\omega_1 \ge \omega_2} \mathrm{d}\omega_1\, \mathrm{d}\omega_2\, 
			\mathfrak{P}_1\, \mathfrak{P}_2\, f_1 f_2  \left\{ 
			[\mathfrak{P}(\omega_1 + \omega_2) - \mathfrak{P}(\omega_1 - \omega_2)] 
			\frac{\alpha(1-\alpha)\omega_2 (\omega_1 - \omega_2)\, }{\omega_1^{2-\alpha}} \chi_{(R,\infty)}(\omega_1)\right.\\
			&\left.\quad
			+ \mathfrak{P}(\omega_1 + \omega_2) 
			\frac{\alpha(1-\alpha)\omega_2 \omega_1\, }{(\omega_1 + \omega_2)^{2-\alpha}} 
			\chi_{(R,\infty)}(\omega_2+\omega_1)\right\} \\
			& - c_{22} \int_0^T \mathrm{d}t \iiint_{\mathbb{R}_+^3} \mathrm{d}\omega_1\, \mathrm{d}\omega_2\, \mathrm{d}\omega\,
			f_1 f_2 f\, \frac{\alpha(1-\alpha)(\omega_{\text{Med}} - \omega_{\text{Inf}})^2\,}{(2\omega_{\text{Med}} - \omega_{\text{Inf}})^{2-\alpha }}\, |k_{\text{Inf}}|\chi_{(R,\infty)}(2\omega_{\text{Med}} - \omega_{\text{Inf}}) \\
			& \quad \times	\mathfrak{R}(\omega_{\text{Sup}})\,	\mathfrak{R}(\omega_{\text{Inf}})\,	\mathfrak{R}(\omega_{\text{Med}})\,	\mathfrak{R}(\omega_{\text{Sup}} - \omega_{\text{Inf}} + \omega_{\text{Med}})   \bar{\mathfrak{R}}_o\\
			& - c_{31} \int_0^T \mathrm{d}t \iiint_{\omega_1 \ge \omega_2 + \omega_3} \mathrm{d}\omega_1\, \mathrm{d}\omega_2\, \mathrm{d}\omega_3\,
			\mathfrak{Q}_1\, \mathfrak{Q}_2\, \mathfrak{Q}_3\, f_1 f_2 f_3 \\
			& \quad \times \left\{ 
			[\mathfrak{Q}(\omega_1 + \omega_2 + \omega_3) - \mathfrak{Q}(\omega_1 - \omega_2 - \omega_3)] 
			\frac{\alpha(1-\alpha)(\omega_2+\omega_3)(\omega_1 - \omega_2 - \omega_3)\, }{\omega_1^{2-\alpha}} \chi_{(R,\infty)}(\omega_1)\right. \\
			& \qquad + [\mathfrak{Q}(\omega_1 + \omega_2 + \omega_3) - \mathfrak{Q}(\omega_1 - \omega_2 - \omega_3)] 
			\frac{\alpha(1-\alpha)\omega_2 \omega_3\, }{(\omega_2+\omega_3)^{2-\alpha}}\chi_{(R,\infty)}(\omega_2+\omega_3) \\
			& \qquad \left. + \mathfrak{Q}(\omega_1 + \omega_2 + \omega_3)
			\frac{\alpha(1-\alpha)(\omega_2+\omega_3)\omega_1\, }{(\omega_1 + \omega_2 + \omega_3)^{2-\alpha}} 
			\chi_{(R,\infty)}(\omega_1+\omega_2+\omega_3)	\right\} \\
			& - c_{31} \iiint_{\mathbb{R}_+^3 \setminus \left( \{\omega_1\ge\omega_2 + \omega_3\} \cup \{\omega_2 \ge \omega_1 + \omega_3\} \cup \{\omega_3 \ge \omega_1 + \omega_2\} \right)} 
			\mathrm{d}\omega_1\, \mathrm{d}\omega_2\, \mathrm{d}\omega_3\, 
			\mathfrak{Q}(\omega_1 + \omega_2 + \omega_3)\alpha(1-\alpha) \\
			& \quad \times \mathfrak{Q}(\omega_1)\, \mathfrak{Q}(\omega_2)\, \mathfrak{Q}(\omega_3)\, f(\omega_1) f(\omega_2) f(\omega_3)
			\frac{\omega_1\omega_2+\omega_2\omega_3+\omega_3\omega_1}{(\omega_1 + \omega_2 + \omega_3)^{2-\alpha}}\chi_{(R,\infty)}(\omega_3+\omega_2+\omega_1).
		\end{aligned}
	\end{equation}
Since \((1-
\alpha)(1 + \omega) \ge \omega^\alpha\), we bound  
\[
\int_{\mathbb{R}} \mathrm{d}k\, f(0,\omega)\, \Xi(\omega)
\ \le\ (1 + \omega) \int_{\mathbb{R}} \mathrm{d}k\, f(0,\omega)\, (1 + \omega)
\ \le\ (1 + \omega) (\mathfrak{M} + \mathfrak{E}),
\]

	which, by \eqref{Lemma:Apriori:E15}, yields \begin{equation}\label{Lemma:Apriori:E15}
		\begin{aligned}
			\mathfrak{M} + \mathfrak{E} \ \ge\ & c_{12} \int_0^T \mathrm{d}t \iint_{\omega_1 \ge \omega_2} \mathrm{d}\omega_1\, \mathrm{d}\omega_2\, 
			\mathfrak{P}_1\, \mathfrak{P}_2\, f_1 f_2 \\
			& \times \left\{ 
			[\mathfrak{P}(\omega_1 + \omega_2) - \mathfrak{P}(\omega_1 - \omega_2)] 
			\frac{\alpha(1-\alpha)\omega_2 (\omega_1 - \omega_2)}{\omega_1^{2-\alpha}} \chi_{(R,\infty)}(\omega_1)
			\right.\\
			&\left.\quad+ \mathfrak{P}(\omega_1 + \omega_2) 
			\frac{\alpha(1-\alpha)\omega_2 \omega_1 }{(\omega_1 + \omega_2)^{2-\alpha}} \chi_{(R,\infty)}(\omega_1+\omega_2)
			\right\}\\
		& + c_{22} \int_0^T \mathrm{d}t \iiint_{\mathbb{R}_+^3} \mathrm{d}\omega_1\, \mathrm{d}\omega_2\, \mathrm{d}\omega\,
		f_1 f_2 f\, \frac{\alpha(1-\alpha)(\omega_{\mathrm{Med}} - \omega_{\text{Inf}})^2}{(2\omega_{\mathrm{Med}} - \omega_{\text{Inf}})^{2 -\alpha}}\, |k_{\mathrm{Inf}}| \chi_{(R,\infty)}(2\omega_{\text{Med}} - \omega_{\text{Inf}})\\
		& \quad \times	\mathfrak{R}(\omega_{\mathrm{Sup}})\,	\mathfrak{R}(\omega_{\mathrm{Inf}})\,	\mathfrak{R}(\omega_{\mathrm{Med}})\,	\mathfrak{R}(\omega_{\mathrm{Sup}} - \omega_{\mathrm{Inf}} + \omega_{\mathrm{Med}}) \bar{\mathfrak{R}}_o \end{aligned}
	\end{equation} \begin{equation*}
		\begin{aligned}
					& +  c_{31} \int_0^T \mathrm{d}t \iiint_{\omega_1 \ge \omega_2 + \omega_3} \mathrm{d}\omega_1\, \mathrm{d}\omega_2\, \mathrm{d}\omega_3\,
			\mathfrak{Q}_1\, \mathfrak{Q}_2\, \mathfrak{Q}_3\, f_1 f_2 f_3 \\
			& \quad \times \left\{ 
			[\mathfrak{Q}(\omega_1 + \omega_2 + \omega_3) - \mathfrak{Q}(\omega_1 - \omega_2 - \omega_3)] 
			\frac{\alpha(1-\alpha)(\omega_2+\omega_3)(\omega_1 - \omega_2 - \omega_3)}{\omega_1^{2-\alpha}} \chi_{(R,\infty)}(\omega_1)\right. \\
			& \qquad + [\mathfrak{Q}(\omega_1 + \omega_2 + \omega_3) - \mathfrak{Q}(\omega_1 - \omega_2 - \omega_3)] 
			\frac{\alpha(1-\alpha)\omega_2 \omega_3}{(\omega_2+\omega_3)^{2-\alpha}} \chi_{(R,\infty)}(\omega_2+\omega_3)\\
			& \qquad \left. + \mathfrak{Q}(\omega_1 + \omega_2 + \omega_3)
			\frac{\alpha(1-\alpha)(\omega_2+\omega_3)\omega_1}{(\omega_1 + \omega_2 + \omega_3)^{2-\alpha}} \chi_{(R,\infty)}(\omega_1+\omega_2+\omega_3)
			\right\} \\
			& + c_{31} \iiint_{\mathbb{R}_+^3 \setminus \left( \{\omega_1 > \omega_2 + \omega_3\} \cup \{\omega_2 > \omega_1 + \omega_3\} \cup \{\omega_3 > \omega_1 + \omega_2\} \right)} 
			\mathrm{d}\omega_1\, \mathrm{d}\omega_2\, \mathrm{d}\omega_3\, 
			\mathfrak{Q}(\omega_1 + \omega_2 + \omega_3) \\
			& \quad \times \mathfrak{Q}(\omega_1)\, \mathfrak{Q}(\omega_2)\, \mathfrak{Q}(\omega_3)\, f(\omega_1) f(\omega_2) f(\omega_3)
			\alpha(1-\alpha)\frac{\omega_1\omega_2+\omega_2\omega_3+\omega_3\omega_1}{(\omega_1 + \omega_2 + \omega_3)^{2-\alpha}}\chi_{(R,\infty)}(\omega_3+\omega_2+\omega_1).
		\end{aligned}
	\end{equation*}
	Inequality \eqref{Lemma:Apriori:1} follows from \eqref{Lemma:Apriori:E15}.

\end{proof}

\begin{lemma}
	\label{Lemma:Apriori2} We assume Assumptions X and Y.
Let \( f \) be a radial solution, in the sense of Definition \ref{def}, to the wave kinetic equation~\eqref{4wave}. Then, for all $t$ such  that \( T^*> t > 0 \), the following a priori estimate holds:
\begin{equation} \label{Lemma:Apriori2:1}
	\partial_t \int_{[R,\infty)} \mathrm{d}\omega\, f(t,\omega)\, \Gamma(\omega)\, (\omega - R) \ge 0, \quad \forall\, R > 0.
\end{equation}

Furthermore, the total mass is uniformly bounded in time:
\begin{equation} \label{Lemma:Apriori2:2}
	\int_{\mathbb{R}^3} \mathrm{d}k\, f(t,k) \le \mathfrak{M}.
\end{equation}

\end{lemma}

\begin{proof} 
We choose the test function to be  
\[
\Xi(\omega) = (\omega - R)_+ = \max\{\omega - R, 0\}.
\]
Then, similarly to equation~\eqref{Lemma:Apriori:E1}, we compute
\begin{equation} \label{Lemma:Apriori2:E1}
	\begin{aligned}
		& \int_{\mathbb{R}_+} \mathrm{d}\omega\, \partial_t f(\omega)\, \Xi(\omega)\, \Gamma(\omega) \\
		=\ & \int_{\mathbb{R}_+} \mathrm{d}\omega\, C_{12}[f](\omega)\, \Xi(\omega)\, \Gamma(\omega)
		+ \int_{\mathbb{R}_+} \mathrm{d}\omega\, C_{22}[f](\omega)\, \Xi(\omega)\, \Gamma(\omega) \\
		& \quad + \int_{\mathbb{R}_+} \mathrm{d}\omega\, C_{31}[f](\omega)\, \Xi(\omega)\, \Gamma(\omega) \\
		=\ & c_{12} \iiint_{\mathbb{R}_+^3} \mathrm{d}\omega_1\, \mathrm{d}\omega_2\, \mathrm{d}\omega\,
		\delta(\omega - \omega_1 - \omega_2)\, \mathfrak{P}\, \mathfrak{P}_1\, \mathfrak{P}_2\,
		[f_1 f_2 - f(f_1 + f_2)]\, [\Xi - \Xi_1 - \Xi_2] \\
		& + c_{22} \iiiint_{\mathbb{R}_+^4} \mathrm{d}\omega_1\, \mathrm{d}\omega_2\, \mathrm{d}\omega_3\, \mathrm{d}\omega\,
		\delta(\omega + \omega_1 - \omega_2 - \omega_3)\,
		\mathfrak{R}\, \mathfrak{R}_1\, \mathfrak{R}_2\, \mathfrak{R}_3\, \\
		& \quad \times \min\{|k|, |k_1|, |k_2|, |k_3|\}\, f_1 f_2 f\, [-\Xi - \Xi_1 + \Xi_2 + \Xi_3] \\
		& + c_{31} \iiiint_{\mathbb{R}_+^4} \mathrm{d}\omega_1\, \mathrm{d}\omega_2\, \mathrm{d}\omega_3\, \mathrm{d}\omega\,
		\delta(\omega - \omega_1 - \omega_2 - \omega_3)\,
		\mathfrak{Q}\, \mathfrak{Q}_1\, \mathfrak{Q}_2\, \mathfrak{Q}_3\, \\
		& \quad \times [f_1 f_2 f_3 - f(f_1 f_2 + f_2 f_3 + f_1 f_3)]\, [\Xi - \Xi_1 - \Xi_2 - \Xi_3].
	\end{aligned}
\end{equation}

Similarly to the estimate in equation~\eqref{Lemma:Apriori:E9}, we obtain the lower bound
\begin{equation} \label{Lemma:Apriori2:E2}
	\int_{\mathbb{R}_+} \mathrm{d}\omega\, C_{12}[f](\omega)\, \Xi(\omega)\, \Gamma(\omega) \ge 0,
\end{equation}
since $\Xi(\omega)$ is convex.

Likewise, following the approach in equation~\eqref{Lemma:Apriori:E5}, we estimate
\begin{equation} \label{Lemma:Apriori2:E3}
	\int_{\mathbb{R}_+} \mathrm{d}\omega\, C_{22}[f](\omega)\, \Xi(\omega)\, \Gamma(\omega) \ge 0.
\end{equation}

Finally, as for equation~\eqref{Lemma:Apriori:E14}, we obtain
\begin{equation} \label{Lemma:Apriori2:E4}
	\int_{\mathbb{R}_+} \mathrm{d}\omega\, C_{31}[f](\omega)\, \Xi(\omega)\, \Gamma(\omega) \ge 0.
\end{equation}

Combining equations~\eqref{Lemma:Apriori2:E1}, \eqref{Lemma:Apriori2:E2}, \eqref{Lemma:Apriori2:E3}, and \eqref{Lemma:Apriori2:E4}, we obtain
\begin{equation} \label{Lemma:Apriori2:E5}
	\int_{\mathbb{R}_+} \mathrm{d}\omega\, \partial_t f(\omega)\, \Xi(\omega)\, \Gamma(\omega) \ge 0,
\end{equation}
which yields the desired result~\eqref{Lemma:Apriori2:1}.

The proof of \eqref{Lemma:Apriori2:2} is straightforward by employing the constant test function \( \Xi = 1 \).

Using equation~\eqref{Lemma:Apriori:E8a} with \( \Xi = 1 \), we obtain:
\begin{equation} \label{Lemma:Ma:E1}
	\int_{\mathbb{R}_+} \mathrm{d}\omega\, C_{12}[f](\omega)\, \Gamma(\omega) \le 0.
\end{equation}

Using equation~\eqref{Lemma:Apriori:E2} with \( \Xi = 1 \), we find:
\begin{equation} \label{Lemma:Ma:E2}
	\int_{\mathbb{R}_+} \mathrm{d}\omega\, C_{22}[f](\omega)\, \Gamma(\omega) = 0.
\end{equation}

Using equation~\eqref{Lemma:Apriori:E12a} with \( \Xi = 1 \), we have:
\begin{equation} \label{Lemma:Ma:E3}
	\int_{\mathbb{R}_+} \mathrm{d}\omega\, C_{31}[f](\omega)\, \Gamma(\omega) \le 0.
\end{equation}

Combining \eqref{Lemma:Ma:E1}--\eqref{Lemma:Ma:E3}, we deduce:
\begin{equation*}
	\begin{aligned}
		\partial_t \int_{\mathbb{R}^3} \mathrm{d}k\, f(t,k) 
		&= \int_{\mathbb{R}_+} \mathrm{d}\omega\, C_{12}[f](\omega)\, \Gamma(\omega) 
		+ \int_{\mathbb{R}_+} \mathrm{d}\omega\, C_{22}[f](\omega)\, \Gamma(\omega) \\
		&\quad + \int_{\mathbb{R}_+} \mathrm{d}\omega\, C_{31}[f](\omega)\, \Gamma(\omega) \le 0,
	\end{aligned}
\end{equation*}
which implies the mass is non-increasing:
\begin{equation*}
	\int_{\mathbb{R}^3} \mathrm{d}k\, f(t,k) \le \int_{\mathbb{R}^3} \mathrm{d}k\, f(0,k) = \mathfrak{M}.
\end{equation*}

As a consequence, the support of \( \mathfrak{F}(t,\omega) \) remains contained in the interval \( [C_*, \infty) \) for all \(T^*> t \ge 0 \).

\end{proof}

\section{Global existence}\label{Sec:Global}
The proposition below shows that equation~\eqref{4wavemild} admits a global mild solution. The proof follows a standard argument developed in~\cite{staffilani2024energy}.

\begin{proposition}
	\label{Propo:Glo} Assume that Assumption X holds. Let \( f_0(k) = f_0(|k|) \geq 0 \) be a radial, nonnegative initial datum satisfying
	\begin{equation} \label{Propo:Global:1}
		\int_{\mathbb{R}^3} f_0(k)\, \mathrm{d}k = \mathfrak{M},  \ \ \ \ \int_{\mathbb{R}^3} f_0(k)\,\omega(k) \mathrm{d}k = \mathfrak{E}.
	\end{equation}
	
	Then there exists a global in time mild solution \( f(t,k) = f(t,|k|) \) to the equation \eqref{4wave}, radial in \( k \), in the sense of  \eqref{4wavemild}. Moreover, 
	\begin{equation} \label{Propo:Global:2}
		\int_{\mathbb{R}^3} f(t,k)\, \mathrm{d}k \leq \mathfrak{M}, \quad \text{for all } t \geq 0.
	\end{equation}
	
\end{proposition}
\begin{proof}  Our proof adopts the cut-off strategy approach introduced in~\cite{staffilani2024energy}.
Let $\Xi$ be a suitable test function. Applying equation ~\eqref{Lemma:Apriori2:E1}, with $\mathfrak{F}(t, \omega)=\Gamma(\omega)f(t,\omega)$, we arrive at the identity
\begin{equation}\label{Propo:Glo:E1}
	\begin{aligned}
		\int_{\mathbb{R}_+} \mathrm{d}\omega\, \partial_t \mathfrak{F}(t, \omega)\, \Xi(\omega) 		=\, \mathscr{L}_1[\Xi] + \mathscr{L}_2[\Xi] + \mathscr{L}_3[\Xi],
	\end{aligned}
\end{equation}
where the terms on the right-hand side are given by
\begin{equation}\label{Propo:Glo:E1a}
	\begin{aligned}
		\mathscr{L}_1[\Xi] :=\ 
		& c_{12} \iint_{\mathbb{R}_+^2} \mathrm{d}\omega_1\, \mathrm{d}\omega_2\,\bar{\mathfrak{P}}(\omega_1)\, \bar{\mathfrak{P}}(\omega_2)\, \Gamma(\omega_1) f(\omega_1)  \Gamma(\omega_2)f(\omega_2)\,  \Gamma(\omega_1+\omega_2)\bar{\mathfrak{P}}(\omega_1 + \omega_2) \\
		& \quad \times \left[ \Xi(\omega_1 + \omega_2) - \Xi(\omega_1) - \Xi(\omega_2) \right] \\
		& - 2c_{12} \iint_{\omega_1 \ge \omega_2} \mathrm{d}\omega_1\, \mathrm{d}\omega_2\, \bar{\mathfrak{P}}(\omega_1)\, \bar{\mathfrak{P}}(\omega_1 - \omega_2)\, \bar{\mathfrak{P}}(\omega_2)\, \Gamma(\omega_1) f(\omega_1)  \Gamma(\omega_1)f(\omega_2) \Gamma(\omega_1-\omega_2) \\
		& \quad \times \left[ \Xi(\omega_1) - \Xi(\omega_1 - \omega_2) - \Xi(\omega_2) \right],
	\end{aligned}
\end{equation}

\begin{equation}\label{Propo:Glo:E1b}
	\begin{aligned}
		\mathscr{L}_2[\Xi] :=\ 
		& c_{22} \iiint_{\mathbb{R}_+^3} \mathrm{d}\omega_1\, \mathrm{d}\omega_2\, \mathrm{d}\omega\, \bar{\mathfrak{R}}(\omega)\, \bar{\mathfrak{R}}(\omega_1)\, \bar{\mathfrak{R}}(\omega_2)\, \bar{\mathfrak{R}}(\omega+\omega_1-\omega_2)\bar{\mathfrak{R}}_o(\omega,\omega_1,\omega_2,\omega+\omega_1-\omega_2) \\
		& \quad \times \min\left\{ |k|(\omega), |k|(\omega_1), |k|(\omega_2), |k|(\omega+\omega_1-\omega_2) \right\}\frac{\Gamma(\omega+\omega_1-\omega_2)}{ |k|(\omega+\omega_1-\omega_2)} \\
		& \quad \times \frac{\Gamma(\omega)}{|k|} f(\omega)  \frac{\Gamma(\omega_1)}{|k_1|}f(\omega_1)  \frac{\Gamma(\omega_2)}{|k_2|}f(\omega_2) \left[ -\Xi(\omega) - \Xi(\omega_1) + \Xi(\omega_2) + \Xi(\omega + \omega_1 - \omega_2) \right],
	\end{aligned}
\end{equation}

\begin{equation}\label{Propo:Glo:E1c}
	\begin{aligned}
		\mathscr{L}_3[\Xi] :=\ 
		& c_{31} \iiint_{\mathbb{R}_+^3} \mathrm{d}\omega_1\, \mathrm{d}\omega_2\, \mathrm{d}\omega_3\, \bar{\mathfrak{Q}}(\omega_1 + \omega_2 + \omega_3)\, \bar{\mathfrak{Q}}(\omega_1)\, \bar{\mathfrak{Q}}(\omega_2)\, \bar{\mathfrak{Q}}(\omega_3)\Gamma(\omega_1+\omega_2+\omega_3) \\
		& \quad \times \Gamma(\omega_1) f(\omega_1) \Gamma(\omega_2)f(\omega_2)\Gamma(\omega_3) f(\omega_3)\, \left[ \Xi(\omega_1 + \omega_2 + \omega_3) - \Xi(\omega_1) - \Xi(\omega_2) - \Xi(\omega_3) \right] \\
		& - 3c_{31} \iiint_{\omega_1 \ge \omega_2 + \omega_3} \mathrm{d}\omega_1\, \mathrm{d}\omega_2\, \mathrm{d}\omega_3\, \bar{\mathfrak{Q}}(\omega_1 - \omega_2 - \omega_3)\, \bar{\mathfrak{Q}}(\omega_1)\, \bar{\mathfrak{Q}}(\omega_2)\, \bar{\mathfrak{Q}}(\omega_3)\Gamma(\omega_1-\omega_2-\omega_3) \\
		& \quad \times\Gamma(\omega_1) f(\omega_1) \Gamma(\omega_2)f(\omega_2)\Gamma(\omega_3) f(\omega_3)\, \left[ \Xi(\omega_1) - \Xi(\omega_1 - \omega_2 - \omega_3) - \Xi(\omega_2) - \Xi(\omega_3) \right].
	\end{aligned}
\end{equation}

We rewrite the above terms as follows:
\begin{equation}\label{Propo:Glo:E2a}
	\begin{aligned}
		\mathscr{L}_1 =\ 
		& c_{12} \iint_{\mathbb{R}_+^2} \mathrm{d}\omega_1\, \mathrm{d}\omega_2\, \mathfrak{F}_1 \mathfrak{F}_2\, \bar{\mathfrak{P}}(\omega_1)\, \bar{\mathfrak{P}}(\omega_2)\, \mathfrak{P}(\omega_1 + \omega_2) \\
		& \quad \times \left[ \Xi(\omega_1 + \omega_2) - \Xi(\omega_1) - \Xi(\omega_2) \right] \\
		& - 2c_{12} \iint_{\omega_1 \ge \omega_2} \mathrm{d}\omega_1\, \mathrm{d}\omega_2\, \bar{\mathfrak{P}}(\omega_1)\, \bar{\mathfrak{P}}(\omega_2)\, \mathfrak{P}(\omega_1 - \omega_2)\, \mathfrak{F}_1 \mathfrak{F}_2 \\
		& \quad \times \left[ \Xi(\omega_1) - \Xi(\omega_1 - \omega_2) - \Xi(\omega_2) \right],
	\end{aligned}
\end{equation}

\begin{equation}\label{Propo:Glo:E2b}
	\begin{aligned}
		\mathscr{L}_2 =\ 
		& c_{22} \iiint_{\mathbb{R}_+^3} \mathrm{d}\omega_1\, \mathrm{d}\omega_2\, \mathrm{d}\omega\, \mathfrak{R}(\omega + \omega_1 - \omega_2)\, \bar{\mathfrak{R}}(\omega)\, \bar{\mathfrak{R}}(\omega_1)\, \bar{\mathfrak{R}}(\omega_2) \\
		& \quad \times \mathfrak{F}_1 \mathfrak{F}_2 \mathfrak{F}\,\frac{\min\left\{ |k|, |k_1|, |k_2|, |k|(\omega + \omega_1 - \omega_2) \right\}}{|k|\, |k_1|\, |k_2|}\, \mathbf{1}_{\omega + \omega_1 \ge \omega_2}  \\
		& \quad \times \left[ -\Xi(\omega) - \Xi(\omega_1) + \Xi(\omega_2) + \Xi(\omega + \omega_1 - \omega_2) \right],
	\end{aligned}
\end{equation}

\begin{equation}\label{Propo:Glo:E2c}
	\begin{aligned}
		\mathscr{L}_3 =\ 
		& c_{31} \iiint_{\mathbb{R}_+^3} \mathrm{d}\omega_1\, \mathrm{d}\omega_2\, \mathrm{d}\omega_3\, \mathfrak{Q}(\omega_1 + \omega_2 + \omega_3)\,   \bar{\mathfrak{Q}}(\omega_1)\, \bar{\mathfrak{Q}}(\omega_2)\, \bar{\mathfrak{Q}}(\omega_3)\, \mathfrak{F}_1 \mathfrak{F}_2 \mathfrak{F}_3 \\
		& \quad \times \left[ \Xi(\omega_1 + \omega_2 + \omega_3) - \Xi(\omega_1) - \Xi(\omega_2) - \Xi(\omega_3) \right] \\
		& - 3c_{31} \iiint_{\omega_1 \ge \omega_2 + \omega_3} \mathrm{d}\omega_1\, \mathrm{d}\omega_2\, \mathrm{d}\omega_3\, \mathfrak{Q}(\omega_1 - \omega_2 - \omega_3)\,  \bar{\mathfrak{Q}}(\omega_1)\, \bar{\mathfrak{Q}}(\omega_2)\, \bar{\mathfrak{Q}}(\omega_3)\, \mathfrak{F}_1 \mathfrak{F}_2 \mathfrak{F}_3 \\
		& \quad \times \left[ \Xi(\omega_1) - \Xi(\omega_1 - \omega_2 - \omega_3) - \Xi(\omega_2) - \Xi(\omega_3) \right].
	\end{aligned}
\end{equation}

We now introduce the following notations:
\begin{equation*}
	\begin{aligned}
		\mathfrak{V}^a_{12} &= \bar{\mathfrak{P}}(\omega_1)\, \bar{\mathfrak{P}}(\omega_2)\, \mathfrak{P}(\omega_1 + \omega_2), \qquad
		\mathfrak{V}^b_{12} = \bar{\mathfrak{P}}(\omega_1)\, \bar{\mathfrak{P}}(\omega_2)\, \mathfrak{P}(\omega_1 - \omega_2),
	\end{aligned}
\end{equation*}

\begin{equation*}
	\begin{aligned}
		\mathfrak{V}_{22} =\ 
		& \mathfrak{R}(\omega + \omega_1 - \omega_2)\, \bar{\mathfrak{R}}(\omega)\, \bar{\mathfrak{R}}(\omega_1)\, \bar{\mathfrak{R}}(\omega_2) \\
		& \times \frac{\min\left\{ |k|, |k_1|, |k_2|, |k|(\omega + \omega_1 - \omega_2) \right\}}{|k|\, |k_1|\, |k_2|}\, \mathbf{1}_{\omega + \omega_1 \ge \omega_2},
	\end{aligned}
\end{equation*}

\begin{equation*}
	\begin{aligned}
		\mathfrak{V}^a_{31} &= \mathfrak{Q}(\omega_1 + \omega_2 + \omega_3)\, \bar{\mathfrak{Q}}(\omega_1)\, \bar{\mathfrak{Q}}(\omega_2)\, \bar{\mathfrak{Q}}(\omega_3), \\
		\mathfrak{V}^b_{31} &= \mathfrak{Q}(\omega_1 - \omega_2 - \omega_3)\, \bar{\mathfrak{Q}}(\omega_1)\, \bar{\mathfrak{Q}}(\omega_2)\, \bar{\mathfrak{Q}}(\omega_3).
	\end{aligned}
\end{equation*}

We then rewrite equation~\eqref{Propo:Glo:E1} as
\begin{equation}\label{Propo:Glo:E6}
	\int_{\mathbb{R}_+} \mathrm{d}\omega\, \partial_t \mathfrak{F}(t, \omega)\, \Xi(\omega) ,
\end{equation}
where
\begin{equation}\label{Propo:Glo:E6o}
\left\langle \mathfrak{Z}[\mathfrak{F}], \Xi \right\rangle := \mathscr{L}_1[\Xi] + \mathscr{L}_2[\Xi] + \mathscr{L}_3[\Xi],
\end{equation}
and
\begin{equation}\label{Propo:Glo:E6a}
	\begin{aligned}
		\mathscr{L}_1[\Xi] =\ 
		& c_{12} \iint_{\mathbb{R}_+^2} \mathrm{d}\omega_1\, \mathrm{d}\omega_2\, \mathfrak{F}_1 \mathfrak{F}_2\, \mathfrak{V}^a_{12} \left[ \Xi(\omega_1 + \omega_2) - \Xi(\omega_1) - \Xi(\omega_2) \right] \\
		& - 2c_{12} \iint_{\omega_1 \ge \omega_2} \mathrm{d}\omega_1\, \mathrm{d}\omega_2\, \mathfrak{F}_1 \mathfrak{F}_2\, \mathfrak{V}^b_{12} \left[ \Xi(\omega_1) - \Xi(\omega_1 - \omega_2) - \Xi(\omega_2) \right],
	\end{aligned}
\end{equation}

\begin{equation}\label{Propo:Glo:E6b}
	\begin{aligned}
		\mathscr{L}_2[\Xi] =\ 
		& c_{22} \iiint_{\mathbb{R}_+^3} \mathrm{d}\omega_1\, \mathrm{d}\omega_2\, \mathrm{d}\omega\, \delta(\omega + \omega_1 - \omega_2 - \omega_3)\, \mathfrak{V}_{22} \\
		& \quad \times \mathfrak{F} \mathfrak{F}_1 \mathfrak{F}_2 \left[ -\Xi(\omega) - \Xi(\omega_1) + \Xi(\omega_2) + \Xi(\omega + \omega_1 - \omega_2) \right],
	\end{aligned}
\end{equation}

\begin{equation}\label{Propo:Glo:E6c}
	\begin{aligned}
		\mathscr{L}_3[\Xi] =\ 
		& c_{31} \iiint_{\mathbb{R}_+^3} \mathrm{d}\omega_1\, \mathrm{d}\omega_2\, \mathrm{d}\omega_3\, \mathfrak{V}^a_{31}\, \mathfrak{F}_1 \mathfrak{F}_2 \mathfrak{F}_3 \\
		& \quad \times \left[ \Xi(\omega_1 + \omega_2 + \omega_3) - \Xi(\omega_1) - \Xi(\omega_2) - \Xi(\omega_3) \right] \\
		& - 3c_{31} \iiint_{\omega_1 \ge \omega_2 + \omega_3} \mathrm{d}\omega_1\, \mathrm{d}\omega_2\, \mathrm{d}\omega_3\, \mathfrak{V}^b_{31}\, \mathfrak{F}_1 \mathfrak{F}_2 \mathfrak{F}_3 \\
		& \quad \times \left[ \Xi(\omega_1) - \Xi(\omega_1 - \omega_2 - \omega_3) - \Xi(\omega_2) - \Xi(\omega_3) \right].
	\end{aligned}
\end{equation}

with initial condition
\[
\mathfrak{F}(0, \omega) = \Gamma f_0(\omega).
\]


	We now divide the proof into several smaller steps.

\textit{Step 1.}

We first observe that, on the resonance manifold associated with \( \mathfrak{V}^a_{12} \), and for \( \omega_1, \omega_2 \ge 1 \),
\begin{equation}\label{Propo:Glo:E6c:1}
	\begin{aligned}
		\mathfrak{V}^a_{12}
		&\le \Gamma(\omega_1 + \omega_2)\big(C_{\mathfrak{P}}'\big)^3(\omega_1 + \omega_2)^{\varpi_1+1}\omega_1^{\varpi_1+1}\omega_2^{\varpi_1+1} \
		&\lesssim \frac{|k_1+k_2|^2}{|k_1+k_2|^{\frac{1}{\theta}-1}}(\omega_1 + \omega_2)^{\varpi_1+1}\omega_1^{\varpi_1+1}\omega_2^{\varpi_1+1} \\
		&\lesssim \frac{(\omega_1+\omega_2)^{2\theta}}{(\omega_1+\omega_2)^{1-\theta}}(\omega_1 + \omega_2)^{\varpi_1+1}\omega_1^{\varpi_1+1}\omega_2^{\varpi_1+1} \\
		&\lesssim (\omega_1+\omega_2)^{3\theta+\varpi_1}\omega_1^{\varpi_1+1}\omega_2^{\varpi_1+1} \\
		&\lesssim \omega_1^{3\theta+2\varpi_1+1}\omega_2^{\varpi_1+1} + \omega_2^{3\theta+2\varpi_1+1}\omega_1^{\varpi_1+1}\ 
		\lesssim \ \omega_1\omega_2.
	\end{aligned}
\end{equation}

Similarly, on the corresponding resonance manifolds, we also have
\begin{equation}\label{Propo:Glo:E6c:2}
	\begin{aligned}
		\mathfrak{V}^b_{12} &\lesssim \omega_1^{3\theta+2\varpi_1+1}\omega_2^{\varpi_1+1} + \omega_2^{3\theta+2\varpi_1+1}\omega_1^{\varpi_1+1}\ 
		  \lesssim \ \omega_1\omega_2, \\
		\mathfrak{V}^a_{31} &\lesssim \omega_1^{3\theta+2\varpi_3+1} \omega_2^{\varpi_3+1} \omega_3^{\varpi_3+1} + \omega_2^{3\theta+2\varpi_3+1} \omega_1^{\varpi_3+1} \omega_3^{\varpi_3+1}\\
		&\quad + \omega_3^{3\theta+2\varpi_3+1} \omega_1^{\varpi_3+1} \omega_2^{\varpi_1+1}\ 
		  \lesssim\ \omega_1\omega_2\omega_3, \\
		\mathfrak{V}^b_{31} &\lesssim \omega_1^{3\theta+2\varpi_3+1} \omega_2^{\varpi_3+1} \omega_3^{\varpi_3+1} + \omega_2^{3\theta+2\varpi_3+1} \omega_1^{\varpi_3+1} \omega_3^{\varpi_3+1}\\
		&\quad + \omega_3^{3\theta+2\varpi_3+1} \omega_1^{\varpi_3+1} \omega_2^{\varpi_3+1}\ 
		  \lesssim\ \omega_1\omega_2\omega_3, \\
		\mathfrak{R}(\omega + \omega_1 - \omega_2)\, 
		\bar{\mathfrak{R}}(\omega)\, 
		\bar{\mathfrak{R}}(\omega_1)\, 
		\bar{\mathfrak{R}}(\omega_2) 
		&\lesssim \omega_1^{2\theta+2\varpi_2+1} \omega_2^{\varpi_2+1} \omega^{\varpi_2+1} + \omega_2^{2\theta+2\varpi_2+1} \omega_1^{\varpi_2+1} \omega^{\varpi_2+1}\\
		&\quad + \omega^{2\theta+2\varpi_2+1} \omega_1^{\varpi_2+1} \omega_2^{\varpi_2+1}\ 
		\lesssim\ \omega_1\omega_2\omega,
	\end{aligned}
\end{equation}
since \( \varpi_1 \le 0 \), \( \varpi_2 \le 0 \),  \( \varpi_3 \le 0 \), \( 3\theta + 2\varpi_1 \le 0 \), \( 2\theta + 2\varpi_2 \le 0 \), and \( 3\theta + 2\varpi_3 \le 0 \) by \eqref{X4}.

Let \( \Xi \in C^2(\mathbb{R}_+) \) be a test function with compact support.

We now turn our attention to the quantity
\begin{equation}
	\begin{aligned}
		\mathfrak{V}_*(\omega,\omega_1,\omega_2;k,k_1,k_2)
		&= \frac{\min\{\,|k|,|k_1|,|k_2|,\,|k|(\omega+\omega_1-\omega_2)\,\}}{|k|\,|k_1|\,|k_2|}
		\,\mathbf{1}_{\{\omega+\omega_1\ge \omega_2\}}
		\\
		&\quad \times \bar{\mathfrak{R}}_o\!\left(\omega,\omega_1,\omega_2,\omega+\omega_1-\omega_2\right)\,
		\bigl[-\Xi(\omega)-\Xi(\omega_1)+\Xi(\omega_2)+\Xi(\omega+\omega_1-\omega_2)\bigr].
	\end{aligned}
\end{equation}

with the goal of showing that \( \mathfrak{V} \) remains bounded on the support of the term in brackets. To this end, it suffices to verify that \( \mathfrak{V} \) is uniformly bounded in neighborhoods of the potential singularities at \( \omega = 0 \), \( \omega_1 = 0 \), and \( \omega_2 = 0 \). Near the singularities, we have the bound
\begin{equation}
	\begin{aligned}
		\mathfrak{V}_* \le \mathfrak{V}
		&:= \frac{\min\{\,|k|,|k_1|,|k_2|,\,|k|(\omega+\omega_1-\omega_2)\,\}}{|k|\,|k_1|\,|k_2|}
		\,\mathbf{1}_{\{\omega+\omega_1 \ge \omega_2\}}
		\\
		&\qquad \times \bigl[-\Xi(\omega)-\Xi(\omega_1)+\Xi(\omega_2)+\Xi(\omega+\omega_1-\omega_2)\bigr].
	\end{aligned}
\end{equation}

We begin by considering the case where exactly one of the quantities \( |k|, |k_1|, |k_2| \) approaches zero. To be more precise, we suppose \( |k| \to 0 \) while \( |k_1| \) and \( |k_2| \) remain bounded below by a fixed constant \( C_1 > 0 \) since he exact same conclusion happens if  \( |k_1| \to 0 \) or  \( |k_2| \to 0 \). In this regime, we estimate:
\[
\mathfrak{V} 
\le \frac{1 }{|k_1||k_2|} \,  
\lesssim \frac{1}{C_1^2} \,  \mathfrak{R}(\omega + \omega_1 - \omega_2)\, ,
\]
which remains bounded on the support of the expression
\[
- \Xi(\omega) - \Xi(\omega_1) + \Xi(\omega_2) + \Xi(\omega + \omega_1 - \omega_2).
\]

We now consider the case in which exactly two of the quantities \( |k|, |k_1|, |k_2| \) tend to zero, while the remaining one is bounded below by a fixed constant \( C_1 > 0 \). Due to the constraint \( \omega + \omega_1 \ge \omega_2 \) on the integration domain, the case where both \( |k| \) and \( |k_1| \) are small while \( |k_2| \ge C_1 \) is not admissible. Consequently, the only relevant case is when \( |k| \to 0 \) and \( |k_2| \to 0 \), while \( |k_1| \ge C_1 \).

To capture the behavior of the expression in brackets, we employ a second-order Taylor expansion. Specifically, we write:
\[
\begin{aligned}
	\mathfrak{V}
	=\, &\frac{\min\left\{ |k|, |k_1|, |k_2|, |k|(\omega + \omega_1 - \omega_2) \right\}}{|k|\, |k_1|\, |k_2|}\, \mathbf{1}_{\omega + \omega_1 \ge \omega_2} \Bigg[
	- \int_0^{\omega - \omega_2} \int_0^{\omega_2 + \upsilon} \Xi''(\upsilon') \, \mathrm{d}\upsilon'\, \mathrm{d}\upsilon 
	- \Xi'(0)(\omega - \omega_2) \\
	& \quad + \int_0^{\omega - \omega_2} \int_0^{\upsilon} \Xi''(\omega_1 + \upsilon') \, \mathrm{d}\upsilon'\, \mathrm{d}\upsilon 
	+ \Xi'(\omega_1)(\omega - \omega_2) 
	\Bigg].
\end{aligned}
\]

We now estimate \( \mathfrak{V} \) on the compact support of \( \Xi \). Using the smoothness of \( \Xi \), we obtain the bound
\[
\begin{aligned}
	|\mathfrak{V}| 
	\lesssim\ & \frac{\min\{ |k|, |k_1|, |k_2|,|k|(\omega + \omega_1 - \omega_2) \}}{|k|\,|k_1|\,|k_2|} \, \mathbf{1}_{\omega + \omega_1 \geq \omega_2} \,\left( \left| \Xi'(\omega_1)(\omega - \omega_2) \right| + \left| \Xi'(0)(\omega - \omega_2) \right| + C_2 (\omega^2 + \omega_2^2) \right),
\end{aligned}
\]
where the constant \( C_2 > 0 \) depends only on \( \Xi \), \( \Xi' \), and \( \Xi'' \). Simplifying the expression further, we obtain
\[
|\mathfrak{V}| 
\lesssim\ \frac{\min\{ |k|, |k_1|, |k_2|, |k|(\omega + \omega_1 - \omega_2) \}}{|k|\,|k_1|\,|k_2|} \, \mathbf{1}_{\omega + \omega_1 \geq \omega_2} \left( |\Xi'(\omega_1)|\, |\omega - \omega_2| + C_2 (\omega^2 + \omega_2^2) \right).
\]

Observe that
\[
\min\{ |k_1|, |k_2|, |k|(\omega + \omega_1 - \omega_2), |k| \} \leq \min\{ |k|, |k_2| \}.
\]
Using this observation along with {Assumption X}, we derive the estimate
\begin{equation}\label{Propo:Glo:E8}
	|\mathfrak{V}| \lesssim \frac{1}{\max\{ |k|, |k_2| \}\, |k_1|} \, \mathbf{1}_{\omega + \omega_1 \geq \omega_2} \left( |\omega - \omega_2|\, |\Xi'(\omega_1)| + \big( \max\{ |k|, |k_2| \} \big)^{\frac{2}{\theta'}} \right).
\end{equation}

We now estimate the second term on the right-hand side of~\eqref{Propo:Glo:E8}. Since \( |k_1| \geq C_1 > 0 \) by assumption, we obtain
\begin{equation}\label{Propo:Glo:E9}
	\frac{\left( \max\{ |k|, |k_2| \} \right)^{\frac{2}{\theta'}}}{\max\{ |k|, |k_2| \}\, |k_1|} 
	= \frac{ \max\{ |k|, |k_2| \}^{\frac{2}{\theta'} - 1} }{ |k_1| } 
	\le \frac{ \max\{ |k|, |k_2| \}^{\frac{2}{\theta'} - 1} }{C_1} 
	\le \frac{1}{C_1},
\end{equation}
provided that \( 1/\theta' \ge \frac12 \), so that the exponent satisfies \( \frac{2}{\theta'} - 1 \ge 0 \).

We now estimate the first term on the right-hand side of~\eqref{Propo:Glo:E8}. Using elementary bounds, we write:
\begin{equation}\label{Propo:Glo:E10}
	\begin{aligned}
		\frac{1}{\max\{ |k|, |k_2| \}\, |k_1|} \, |\omega - \omega_2| \, |\Xi'(\omega_1)| 
		&\le \frac{\max\{ \omega, \omega_2 \}}{\max\{ |k|, |k_2| \}} \cdot \frac{\omega_1}{|k_1|} \cdot \left| \frac{\Xi'(\omega_1)}{\omega_1} \right| \\
		&\le \frac{\omega \Big( \max\{ |k|, |k_2| \}\Big)}{\max\{ |k|, |k_2| \}} \cdot \frac{\omega_1}{|k_1|} \cdot \left| \frac{\Xi'(\omega_1)}{\omega_1} \right|.
	\end{aligned}
\end{equation}

We now examine the boundedness of the terms in this expression. Since $1/\theta>1$,  \( \omega / |k| \) remains bounded for \( 0 \leq |k| \leq 1 \), and because
\[
\lim_{\omega_1 \to 0} \frac{\Xi'(\omega_1)}{\omega_1} = \Xi''(0),
\]
it follows that the quantity \( \frac{\Xi'(\omega_1)}{\omega_1} \) is uniformly bounded on the support of \( \Xi \).   Therefore, the entire expression in~\eqref{Propo:Glo:E10} is uniformly bounded.
	Combining the estimates from \eqref{Propo:Glo:E8}, \eqref{Propo:Glo:E9}, and \eqref{Propo:Glo:E10}, we conclude that \( |\mathfrak{V}| \) remains uniformly bounded on the support of the integrand. This completes the bound for \( \mathfrak{V} \) in the regime where two of the frequencies are small and the third is bounded away from zero.

We now address the final case in which all three quantities \( |k|, |k_1|, |k_2| \) are simultaneously small. In this regime, we analyze the expression
\[
\begin{aligned}
	\mathfrak{V}
	=\, &  \frac{\min\{ |k|, |k_1|, |k_2| \}}{ |k|\, |k_1|\, |k_2| } \, \mathbf{1}_{\omega + \omega_1 \ge \omega_2} \int_0^{\omega - \omega_2} \int_0^{\omega_2 - \omega_1} \left| \Xi''(\omega_1 + \upsilon + \upsilon') \right| \, \mathrm{d}\upsilon' \mathrm{d}\upsilon.
\end{aligned}
\]

We estimate the magnitude of \( \mathfrak{V} \) as
\begin{equation}\label{Propo:Glo:E11}
	\left| \mathfrak{V} \right| \lesssim \frac{ \min\{ |k|, |k_1|, |k_2| \} }{ |k|\, |k_1|\, |k_2| } \, \mathbf{1}_{\omega + \omega_1 \ge \omega_2} \, |\omega - \omega_2|\, |\omega_2 - \omega_1| \, \|\Xi''\|_{L^\infty}.
\end{equation}

To estimate \( |\mathfrak{K}| \) as given in~\eqref{Propo:Glo:E11}, we consider three separate scenarios based on which among the three is the smallest.

\textit{Case 1: \( |k_2| = \min\{ |k|, |k_1|, |k_2| \} \).}  
In this setting, the inequalities \( \omega - \omega_2 \ge 0 \) and \( \omega_1 - \omega_2 \ge 0 \) hold. Using~\eqref{Propo:Glo:E11}, we estimate:
\[
\begin{aligned}
	|\mathfrak{V}|
	&\lesssim \frac{1}{|k|\, |k_1|} \cdot |\omega - \omega_2|\, |\omega_2 - \omega_1|\, \|\Xi''\|_{L^\infty} \\
	&\lesssim \frac{\omega - \omega_2}{|k|} \cdot \frac{\omega_1 - \omega_2}{|k_1|} \cdot \|\Xi''\|_{L^\infty} \\
	&\lesssim \frac{\omega}{|k|} \cdot \frac{\omega_1}{|k_1|} \cdot \|\Xi''\|_{L^\infty},
\end{aligned}
\]
which confirms that \( |\mathfrak{V}| \) remains bounded in this case.

\medskip

\textit{Case 2: \( |k_1| = \min\{ |k|, |k_1|, |k_2| \} \).}  
Again from~\eqref{Propo:Glo:E11}, we write:
\[
\begin{aligned}
	|\mathfrak{V}|
	&\lesssim \frac{1}{|k|\, |k_2|} \cdot |\omega - \omega_2|\, |\omega_2 - \omega_1|\, \|\Xi''\|_{L^\infty} \\
	&\lesssim \frac{|\omega - \omega_2|}{|k|} \cdot \frac{\omega_2 - \omega_1}{|k_2|} \cdot \|\Xi''\|_{L^\infty}.
\end{aligned}
\]
We distinguish two subcases:
\begin{itemize}
	\item If \( \omega \ge \omega_2 \), then \( \frac{|\omega - \omega_2|}{|k|} \le \frac{\omega}{|k|} \), which is bounded.
	\item If \( \omega < \omega_2 \), then \( \omega_2 - \omega \le \omega_1 \) by the function $\mathbf{1}_{\omega + \omega_1 \geq \omega_2}$, and hence \( \frac{\omega_2 - \omega}{|k|} \le \frac{\omega_1}{|k|} \), which remains controlled.
\end{itemize}
Thus, in both subcases, \( |\mathfrak{V}| \) is bounded.

\medskip

\textit{Case 3: \( |k| = \min\{ |k|, |k_1|, |k_2| \} \).}  
This time we estimate:
\[
\begin{aligned}
	|\mathfrak{V}|
	&\lesssim \frac{1}{|k_1|\, |k_2|} \cdot |\omega - \omega_2|\, |\omega_2 - \omega_1|\, \|\Xi''\|_{L^\infty} \\
	&\lesssim \frac{\omega_2 - \omega}{|k_2|} \cdot \frac{|\omega_2 - \omega_1|}{|k_1|} \cdot \|\Xi''\|_{L^\infty}.
\end{aligned}
\]
We again split into two subcases:
\begin{itemize}
	\item If \( \omega_1 \ge \omega_2 \), then \( \frac{|\omega_2 - \omega_1|}{|k_1|} \le \frac{\omega_1}{|k_1|} \) as above, which is bounded.
	\item If \( \omega_1 < \omega_2 \), then \( \omega_2 - \omega_1 \le \omega \), implying \( \frac{\omega_2 - \omega_1}{|k_1|} \le \frac{\omega}{|k_1|} \), which is also bounded.
\end{itemize}
Hence, in all cases, we conclude that \( |\mathfrak{V}| \) remains uniformly bounded.

	\textit{Step 2: Cut-off versions of \eqref{Propo:Glo:E6}.}
	
For each \( n > 0 \), we introduce truncated  versions of the collision kernels by defining:
\begin{equation}\label{Propo:Glo:E3n}
	\begin{aligned}
		\mathfrak{K}^a_{12n} &:=  \mathfrak{V}^a_{12}\, \chi_{\{\omega_1,\omega_2\le n\}}, \qquad
		\mathfrak{K}^b_{12n} :=   \mathfrak{V}^b_{12}\, \chi_{\{\omega_1,\omega_2\le n\}},
	\end{aligned}
\end{equation}

\begin{equation}\label{Propo:Glo:E4n}
	\begin{aligned}
		\mathfrak{K}_{22n} :=  \mathfrak{K}_{22}\, \chi_{\{\omega,\omega_1,\omega_2\le n\}},
	\end{aligned}
\end{equation}

\begin{equation}\label{Propo:Glo:E5n}
	\begin{aligned}
		\mathfrak{K}^a_{31n} &:=   \mathfrak{K}^a_{31}\, \chi_{\{\omega_1,\omega_2,\omega_3\le n\}}, \qquad
		\mathfrak{K}^b_{31n} :=  \mathfrak{K}^b_{31}\, \chi_{\{\omega_1,\omega_2,\omega_3\le n\}}.
	\end{aligned}
\end{equation}

We now consider the regularized (truncated) equation associated with the cut-off kernels. For each \( n > 0 \), we define the approximated equation:
\begin{equation}\label{Propo:Glo:E6n}
	\int_{\mathbb{R}_+} \mathrm{d}\omega\, \partial_t \mathfrak{F}^n(t,\omega)\, \Xi(\omega) 
	= \left\langle \mathfrak{Z}^n[\mathfrak{F}^n],\, \Xi \right\rangle 
	:= \mathscr{L}_1^n[\Xi] + \mathscr{L}_2^n[\Xi] + \mathscr{L}_3^n[\Xi],
\end{equation}
with components given by
\begin{equation}\label{Propo:Glo:E6an}
	\begin{aligned}
		\mathscr{L}_1^n[\Xi] =\ 
		& c_{12} \iint_{\mathbb{R}_+^2} \mathrm{d}\omega_1\, \mathrm{d}\omega_2\, \mathfrak{F}_1^n \mathfrak{F}_2^n\, \mathfrak{K}^a_{12n} \left[ \Xi(\omega_1 + \omega_2) - \Xi(\omega_1) - \Xi(\omega_2) \right] \\
		& - 2c_{12} \iint_{\omega_1 \ge \omega_2} \mathrm{d}\omega_1\, \mathrm{d}\omega_2\, \mathfrak{F}_1^n \mathfrak{F}_2^n\, \mathfrak{K}^b_{12n} \left[ \Xi(\omega_1) - \Xi(\omega_1 - \omega_2) - \Xi(\omega_2) \right],
	\end{aligned}
\end{equation}

\begin{equation}\label{Propo:Glo:E6bn}
	\begin{aligned}
		\mathscr{L}_2^n[\Xi] =\ 
		& c_{22} \iiint_{\mathbb{R}_+^3} \mathrm{d}\omega\, \mathrm{d}\omega_1\, \mathrm{d}\omega_2\, \delta(\omega + \omega_1 - \omega_2 - \omega_3)\, \mathfrak{K}_{22n} \\
		& \quad \times \mathfrak{F}^n \mathfrak{F}_1^n \mathfrak{F}_2^n \left[ -\Xi(\omega) - \Xi(\omega_1) + \Xi(\omega_2) + \Xi(\omega + \omega_1 - \omega_2) \right],
	\end{aligned}
\end{equation}

\begin{equation}\label{Propo:Glo:E6cn}
	\begin{aligned}
		\mathscr{L}_3^n[\Xi] =\ 
		& c_{31} \iiint_{\mathbb{R}_+^3} \mathrm{d}\omega_1\, \mathrm{d}\omega_2\, \mathrm{d}\omega_3\, \mathfrak{K}^a_{31n}\, \mathfrak{F}_1^n \mathfrak{F}_2^n \mathfrak{F}_3^n \left[ \Xi(\omega_1 + \omega_2 + \omega_3) - \Xi(\omega_1) - \Xi(\omega_2) - \Xi(\omega_3) \right] \\
		& - 3c_{31} \iiint_{\omega_1 \ge \omega_2 + \omega_3} \mathrm{d}\omega_1\, \mathrm{d}\omega_2\, \mathrm{d}\omega_3\, \mathfrak{K}^b_{31n}\, \mathfrak{F}_1^n \mathfrak{F}_2^n \mathfrak{F}_3^n \left[ \Xi(\omega_1) - \Xi(\omega_1 - \omega_2 - \omega_3) - \Xi(\omega_2) - \Xi(\omega_3) \right].
	\end{aligned}
\end{equation}

The initial condition for the regularized problem is given by
\[
\mathfrak{F}^n(0,\omega) = \Gamma f_0(\omega)\chi_{\{\omega\le n\}}.
\]

We begin by noting that the following bounds  hold  for the approximated solution \( \mathfrak{F}^n \) 
\begin{equation}\label{Propo:Glo:E12n}
	\int_{\mathbb{R}_+} \mathrm{d}\omega\, \mathfrak{F}^n(t) \,(1+\omega)
	\le \int_{\mathbb{R}_+} \mathrm{d}\omega\, \mathfrak{F}^n(0) \,(1+\omega)
	= \mathfrak{M} + \mathfrak{E}.
\end{equation}

Next, we estimate the nonlinear term \( \mathfrak{Z}^n[\mathfrak{F}^n] \) by testing it against functions \( \Xi \in L^\infty(\mathbb{R}_+) \) with \( \|\Xi\|_{L^\infty} = 1 \). We obtain
\begin{equation}\label{Propo:Glo:E12nx}
	\begin{aligned}
		\sup_{\|\Xi\|_{L^\infty} = 1} \left| \left\langle \mathfrak{Z}^n[\mathfrak{F}^n], \Xi \right\rangle \right|
		\lesssim\ 
		& \iiint_{\mathbb{R}_+^3} \mathrm{d}\omega\, \mathrm{d}\omega_1\, \mathrm{d}\omega_2\, \left| \mathfrak{F}^n(\omega)\, \mathfrak{F}^n(\omega_1)\, \mathfrak{F}^n(\omega_2) \right| \\
		& + \iint_{\mathbb{R}_+^2} \mathrm{d}\omega\, \mathrm{d}\omega_1\, \left| \mathfrak{F}^n(\omega)\, \mathfrak{F}^n(\omega_1) \right| \\
		\lesssim\ 
		& \left( \int_{\mathbb{R}_+} \mathrm{d}\omega\, |\mathfrak{F}^n(\omega)| \right)^3 
		+ \left( \int_{\mathbb{R}_+} \mathrm{d}\omega\, |\mathfrak{F}^n(\omega)| \right)^2,
	\end{aligned}
\end{equation}
where the constants on the right hand side depend on $n$.

Next, we note that the operator \( \mathfrak{Z}^n \) is Lipschitz continuous on the set
\[
S_{\mathfrak{Z}} := \left\{ F \mid F \geq 0, \quad \|F\|_{\mathscr{R}_+} \leq \mathfrak{M} \right\},
\]
where the Radon norm is defined in \eqref{Radon}.

More precisely, for any \( F, \bar{F} \in S_{\mathfrak{Z}} \), we estimate
\begin{equation*}
	\begin{aligned}
		& \sup_{\|\Xi\|_{L^\infty} = 1} \left| \left\langle \mathfrak{Z}^n[F] - \mathfrak{Z}^n[\bar{F}], \Xi \right\rangle \right| \\
		\lesssim\ & \left| c_{12} \iint_{\mathbb{R}_+^2} \mathrm{d}\omega_1 \mathrm{d}\omega_2 \left( F_1 F_2 - \bar{F}_1 \bar{F}_2 \right) \mathfrak{V}^a_{12n} \left[ \Xi(\omega_1 + \omega_2) - \Xi(\omega_1) - \Xi(\omega_2) \right] \right| \\
		& + \left| 2 c_{12} \iint_{\omega_1 \ge \omega_2} \mathrm{d}\omega_1 \mathrm{d}\omega_2\left( F_1 F_2 - \bar{F}_1 \bar{F}_2 \right) \mathfrak{V}^b_{12n} \left[ \Xi(\omega_1) - \Xi(\omega_1 - \omega_2) - \Xi(\omega_2) \right]  \right| \\
		& + \left| c_{22} \iiint_{\mathbb{R}_+^3} \delta(\omega + \omega_1 - \omega_2 - \omega_3)  \mathrm{d}\omega \mathrm{d}\omega_1 \mathrm{d}\omega_2\mathfrak{V}_{22n} \left( F F_1 F_2 - \bar{F} \bar{F}_1 \bar{F}_2 \right) \right. \\
		& \quad \times \left. \left[ -\Xi(\omega) - \Xi(\omega_1) + \Xi(\omega_2) + \Xi(\omega + \omega_1 - \omega_2) \right] \right| \\
		& + \left| c_{31} \iiint_{\mathbb{R}_+^3}\mathrm{d}\omega_1 \mathrm{d}\omega_2 \mathrm{d}\omega_3 \left( F F_1 F_2 - \bar{F} \bar{F}_1 \bar{F}_2 \right) \mathfrak{V}^a_{31n} \left[ \Xi(\omega_1 + \omega_2 + \omega_3) - \Xi(\omega_1) - \Xi(\omega_2) - \Xi(\omega_3) \right]  \right| \\
		& + 3 c_{31} \left| \iiint_{\omega_1 \ge \omega_2 + \omega_3} \mathrm{d}\omega_1 \mathrm{d}\omega_2 \mathrm{d}\omega_3 \left( F F_1 F_2 - \bar{F} \bar{F}_1 \bar{F}_2 \right) \mathfrak{V}^b_{31n} \right. \\
		& \quad \times \left. \left[ \Xi(\omega_1) - \Xi(\omega_1 - \omega_2 - \omega_3) - \Xi(\omega_2) - \Xi(\omega_3) \right] \right|,
	\end{aligned}
\end{equation*}
which implies the estimate, by \eqref{Propo:Glo:E6c:1}, \eqref{Propo:Glo:E6c:2} and  \eqref{Propo:Glo:E12n},
\begin{equation*}
	\begin{aligned}
		\sup_{\|\Xi\|_{L^\infty} = 1} \left| \left\langle \mathfrak{Z}^n[F] - \mathfrak{Z}^n[\bar{F}], \Xi \right\rangle \right|
		\lesssim\ 
		& \iiint_{\mathbb{R}_+^3} \mathrm{d}\omega_1 \mathrm{d}\omega_2 \mathrm{d}\omega  \left| F_1 F_2 F - \bar{F}_1 \bar{F}_2 \bar{F} \right|\\
		& + \iint_{\mathbb{R}_+^2} \mathrm{d}\omega_1 \mathrm{d}\omega_2\left| F_1 F_2 - \bar{F}_1 \bar{F}_2 \right|  \\
		\lesssim\ 
		& \left( (\mathfrak{M} +\mathfrak{E})  + (\mathfrak{M} +\mathfrak{E})^2 \right) \| F - \bar{F} \|_{\mathscr{R}_+}.
	\end{aligned}
\end{equation*}

The local existence of a solution \( \mathfrak{F}^n \in C^1([0, T]; \mathscr{R}_+([0, \infty))) \) to \eqref{Propo:Glo:E6n} follows from a standard fixed-point argument, valid for sufficiently small times \( T > 0 \), depending on $\mathfrak{E}+\mathfrak{M}$. 

Utilizing the a priori estimate  \eqref{Propo:Glo:E12n}, we extend the local existence result by iterating the fixed-point argument successively over intervals \( [0, T] \), \( [T, 2T] \), \( [2T, 3T] \), and so forth. Consequently, this procedure guarantees the global existence of a solution 
\[
\mathfrak{F}^n \in C^1([0, \infty); \mathscr{R}_+([0, \infty))).
\]

\textit{Step 3: Global Existence.}

Consider the sequence of solutions constructed in Step 2, \( \{\mathfrak{F}^n\} \subset C^1([0, \infty); \mathscr{R}_+([0,\infty))) \). For any test function \( \Xi \in C^2_c(\mathbb{R}_+) \), we estimate, using \eqref{Propo:Glo:E6c:1}, \eqref{Propo:Glo:E6c:2} and  \eqref{Propo:Glo:E12n},
\[
\left| \int_0^\infty\, \mathrm{d}\omega  \mathfrak{F}^n(t_2, \omega)\, \Xi(\omega)- \int_0^\infty\, \mathrm{d}\omega  \mathfrak{F}^n(t_1, \omega)\, \Xi(\omega)\right| \leq C |t_2 - t_1|,
\]
for all \( t_1, t_2 \geq 0 \), where the constant \( C > 0 \) is uniform in \( n \).

By the standard Arzela-Ascoli argument (see \cite{giri2011continuous,stewart1989global}), there exists a subsequence \( \{\mathfrak{F}^{i_n}\} \) and a limit function 
\[
\mathfrak{F}^\infty \in C^1([0, \infty); \mathscr{R}_+([0,\infty)))
\]
such that \( \mathfrak{F}^{i_n}(t) \rightharpoonup \mathfrak{F}^\infty(t) \) weakly in \( \mathscr{R}_+([0,\infty)) \) for each \( t \geq 0 \), and the convergence is uniform in \( t \) on compact intervals.

The limit function \( \mathfrak{F}^\infty \) is a mild solution of \eqref{Propo:Glo:E6} with \( n = \infty \). Moreover, the estimate  \eqref{Propo:Global:2} follows directly from the uniform bounds established for \( \mathfrak{F}^n \).

\end{proof}

	\section{The DDM framework}
\label{Sec:DDM}

In this section, we set the foundation for our domain decomposition method (DDM), which is based on partitioning the half-line \( \mathbb{R}_+ \) into small intervals~\cite{halpern2009nonlinear, Lions:1989:OSA, toselli2004domain}. This method enables a detailed divide-and-conquer analysis of the energy distribution within each subinterval, from which we derive precise estimates of the outward flow of energy.  

In contrast to the companion paper~\cite{staffilani2025condensate}, where the DDM approach is designed to study the concentration of mass near the origin, the method considered in this work is aimed at investigating the cascade of energy toward infinity.

Let \( \ell \) be a natural number such that \( \ell > 10^{10} \). Let the parameter \( \sigma \) satisfy the condition
\begin{equation}\label{sigma}\begin{aligned}
&	\frac15\min\left\{  \frac{4\varpi_{3}+3\theta+\alpha}{3},\ \frac{3\varpi_{1}+3\theta+\alpha}{2},\ \frac{4\varpi_2+\alpha+\gamma}{6},\   3\varpi_2 + 2 - 2\theta + \kappa_2 - c_{\mathrm{in}} + \gamma  \right\} > \sigma > 0,\end{aligned}
\end{equation}
where we have used \eqref{X4}.

In the sequel, we will also use another parameter, \( \epsilon >0 \), which is chosen such that
\begin{equation}\label{epsilon}
	\begin{aligned}
		10(\epsilon + \sigma)  &< 3\theta + 3\varpi_1 + 1, \\
		 10(\epsilon + \sigma)   &< 3\varpi_2 + 2 - 2\theta + \kappa_2 - c_{\mathrm{in}} +\gamma, \\
		10(\epsilon + \sigma)   &< 3\theta + 4\varpi_3 + 1.
	\end{aligned}
\end{equation}

We define  
\begin{equation} \label{Sec:DD:0}
	\Omega_{\ell} = 2^{\ell}, \quad 
	\Upsilon_{\ell} = \left\lfloor \min\left\{ \frac{\ell}{4} \left(4\varpi_2+\alpha +\gamma\right),\ \frac{\ell \epsilon}{16} \right\} \right\rfloor, \quad 
	\Delta_{\ell} = \frac{\Omega_{\ell}}{2^{\Upsilon_{\ell}}},
\end{equation}
and construct our multiscale domain decomposition for the interval \([\Omega_{\ell}, \infty)\) as follows:
\begin{equation}\label{Sec:DD:1}
	\begin{aligned}
		\text{(I) Number of subdomains:} \quad & \mathscr{O}_{\Omega_{\ell}} = 2^{\Upsilon_{\ell}}, \\[5pt]
		\text{(II) Non-overlapping subdomains:} \quad & \mathcal{D}_{2^{\Upsilon_{\ell}} -1 - i}^{\Omega_{\ell}} = \left[ \Omega_{\ell}+i \Delta_\ell, \Omega_{\ell}+(i+1) \Delta_\ell \right), \quad i = 0, \ldots, 2^{\Upsilon_{\ell}} - 2, \\[2pt]
		& \mathcal{D}_{0}^{\Omega_{\ell}} = \left[ 2\Omega_{\ell}-\Delta_\ell, \infty \right), \\[5pt]
		\text{(III) Overlapping subdomains:} \quad & \mathscr{D}^{\Omega_{\ell}}_{2^{\Upsilon_{\ell}}  - i} = \left[ \Omega_{\ell}+(i-1) \Delta_\ell, \Omega_{\ell}+(i+2) \Delta_\ell \right), \quad i = 1, \ldots, 2^{\Upsilon_\ell} - 3, \\[2pt]
		& \mathscr{D}_{0}^{\Omega_{\ell}} = \left[ 2\Omega_{\ell} - 2\Delta_\ell, \infty \right), \quad \mathscr{D}_{1}^{\Omega_{\ell}} = \left[ 2\Omega_{\ell} - 3\Delta_\ell, \infty \right), \\[2pt]
		& \mathscr{D}_{2^{\Upsilon_{\ell}} - 2}^{\Omega_{\ell}} = \left[ \Omega_{\ell},\Omega_{\ell}+ 3\Delta_\ell \right), \\[2pt]
		& \mathscr{D}_{2^{\Upsilon_{\ell}} -1}^{\Omega_{\ell}} = \left[\Omega_{\ell}, \Omega_{\ell}+2\Delta_\ell \right).
	\end{aligned}
\end{equation}

We now define, for a fixed constant \( c_o > 0 \) and \(T^*>T\ge0\)
\begin{equation}\label{Sec:DD:2}
	\mathscr{M}_\ell^T := \left\{ t \in [0, T] : 
	\int_{[\Omega_{\ell}, \infty)} \mathrm{d}\omega\, \mathfrak{F}(t, \omega) 
	\ge c_o\, \Omega_{\ell}^{-\sigma} \right\},
\end{equation}
\begin{equation}\label{Sec:DD:3}
	\mathscr{M}_{\ell,i}^T := \left\{ t \in [0, T] :
	\int_{\mathscr{D}_{i}^{\Omega_{\ell}}} \mathrm{d}\omega\, \mathfrak{F}(t, \omega)
	\ge c_o\, \Omega_{\ell+1}^{-\sigma} \right\},
\end{equation}
for \( i = 0, \ldots, 2^{\Upsilon_{\ell}} - 1 \).

We then define:
\begin{equation}\label{Sec:DD:4}
	\begin{aligned}
		\mathscr{N}_{\ell}^T &:= \mathscr{M}_{\ell}^T \setminus 
		\bigcup_{i=0}^{2^{\Upsilon_{\ell}} - 1} \mathscr{M}_{\ell,i}^T, \\
		\mathscr{P}_{\ell}^T &:= \bigcup_{i=2^{\Upsilon_{\ell} - 1} - 1}^{2^{\Upsilon_{\ell}} - 1}
		\mathscr{M}_{\ell,i}^T, \\
		\mathscr{Q}_{\ell}^T &:= \bigcup_{i=0}^{2^{\Upsilon_{\ell} - 1} - 2}
		\mathscr{M}_{\ell,i}^T.
	\end{aligned}
\end{equation}

We now outline the approach used to establish the multiscale estimates presented in Sections~\ref{Sec:First}, \ref{Sec:Second}, and \ref{Sec:Third}, all in within the context of our DDM framework.

\begin{itemize}
	\item In Section~\ref{Sec:First}, we derive an estimate for the Lebesgue measure of the set $\mathscr{P}_{\ell}^T$,   as stated in equation~\eqref{Lemma:Mutis1:1}.
	
	\item In Section~\ref{Sec:Second}, we estimate the size of the set $\mathscr{N}_{\ell}^T$, as given in~\eqref{Lemma:Mutiscale2:1}. This result relies on the bound established in~\eqref{Propo:Collision:1}. Combining the estimates for $\mathscr{P}_{\ell}^T$ and $\mathscr{N}_{\ell}^T$, we obtain a corresponding estimate for the size of the set $\mathscr{M}_{\ell}^T$.
	
	\item In Section~\ref{Sec:Third}, based on the bounds for the size of \( \mathscr{M}_{\ell}^T \), Proposition~\ref{Propo:Cascade1} demonstrates the occurrence of an immediate energy cascade. Moreover, it shows that, even under weaker assumptions on the initial data, an energy cascade still arises within finite time.
\end{itemize}

	\section{Estimating $\mathscr{P}_{\ell}^T$}\label{Sec:First}

	\begin{proposition}
		\label{Lemma:Mutis1} We assume Assumptions X and Y. Let \( T^* \) be as defined in~\eqref{T0}. We choose $0\le T<T^*$ if $T^*>0$ and $ T=0$ if $T^*=0$. By the definitions in~Assumption X,~\eqref{sigma}, and~\eqref{epsilon}, there exists a constant \( \mathfrak{C}_{\mathscr{P}} \in \mathbb{N} \) such that for all \( \ell > \mathfrak{C}_{\mathscr{P}} \) with \( \ell \in \mathbb{N} \), the following bound holds:
		\begin{equation} \label{Lemma:Mutis1:1}
			\left| \mathscr{P}_{\ell}^T \right| \le C_{\mathscr{P}}\, \Omega_{\ell}^{-c_{\mathscr{P}}},
		\end{equation}
			for some universal constant \( C_{\mathscr{P}} > 0 \), independent of $ \epsilon,\sigma,\theta, \varpi_{1},\varpi_{2},\varpi_{3},  T, T^*$ and \( c_{\mathscr{P}} \) is given by:
		
		\begin{itemize}
			\item If \( c_{12} > 0 \), \( c_{22} = c_{31} = 0 \), then
			\begin{equation} \label{Lemma:Mutis1:1_case}
				c_{\mathscr{P}} := -(\epsilon + 2\sigma -  3\theta - 3\varpi_1 - 1).
			\end{equation}
			
			\item If \( c_{22} > 0 \),   then
			\begin{equation} \label{Lemma:Mutis1:2_case}
				c_{\mathscr{P}} := -({\epsilon + \sigma- 3\varpi_2 - 2 + 2\theta - \kappa_2 + c_{\mathrm{in}} -\gamma}).
			\end{equation}
			
			\item If \( c_{31} > 0 \), \( c_{12} = c_{22} = 0 \), then
			\begin{equation} \label{Lemma:Mutis1:3_case}
				c_{\mathscr{P}} := -(\epsilon + 3\sigma -  3\theta - 4\varpi_3 - 1).
			\end{equation}

			\item If \( c_{12} > 0 \), \( c_{31} > 0 \), \( c_{22} = 0 \), then
			\begin{equation} \label{Lemma:Mutis1:7_case}
				c_{\mathscr{P}} := -\min\left\{\epsilon + 2\sigma -  3\theta - 3\varpi_1 - 1,\epsilon + 3\sigma -  3\theta - 4\varpi_3 - 1
				\right\}.
			\end{equation}

		\end{itemize}

	\end{proposition}
	\begin{proof}

We now prove by contradiction that, for some constant \( \mathcal{C}_1 > 0 \) independent of \( \ell \) and \( \mathfrak{F} \), there cannot exist an infinite set  
\[
\mathscr{U}_o = \{\ell_1, \ell_2, \dots\}
\]
such that:
\begin{itemize}
	\item for each element \( \ell_j \) in \( \mathscr{U}_o \), there exists a subset \( \mathscr{P}_{\ell_j}'^{T} \subset \mathscr{P}_{\ell_j}^{T} \);
	\item for each time \( t \) in \( \mathscr{P}_{\ell_j}'^{T} \), there exists  
	\( \rho(t) \in \{2^{\Upsilon_{\ell_j} - 1} - 1,\, \dots,\, 2^{\Upsilon_{\ell_j}} - 1\} \) satisfying
	\begin{equation} \label{Lemma:Mutis:E1a}
		\begin{aligned}
			& \int_0^{T} \mathrm{d}t\, 4^{-1}\, \theta\, C_\omega^{-1} C_{\mathfrak{P}}^{3}\, c_{12}\, 
			\lvert \Omega_{\ell_j} \rvert^{3\theta + 3\varpi_1 + 1}
			\left[ \int_{\mathbb{R}_+} \mathrm{d}\omega\,
			\mathfrak{F}\, \omega\,
			\chi_{\big\{\omega \in \mathscr{D}_{\rho(t)}^{\Omega_{\ell_j}}\big\}} \right]^2
			\chi_{\mathscr{P}_{\ell_j}'^{\mathscr{T}_{\ell_j}}}(t) \\[0.4em]
			& \quad + \int_0^{T} \mathrm{d}t\, \tfrac{1}{2}\, c_{31}\, C_{\mathfrak{Q}}^{3}\, \theta\, C_\omega^{-1} C_{\mathfrak{P}}\, 
			\lvert \Omega_{\ell_j} \rvert^{3\theta + 4\varpi_3 + 1}
			\left[ \int_{\mathbb{R}_+} \mathrm{d}\omega\,
			\mathfrak{F}\, \omega\,
			\chi_{\big\{\omega \in \mathscr{D}_{\rho(t)}^{\Omega_{\ell_j}}\big\}} \right]^3
			\chi_{\mathscr{P}_{\ell_j}'^{\mathscr{T}_{\ell_j}}}(t)
			> \mathcal{C}_1\, \Omega_{\ell_j}^{\epsilon},
		\end{aligned}
	\end{equation}
	when \( c_{22} = 0 \), and
	\begin{equation} \label{Lemma:Mutis:E1}
		\begin{aligned}
			& \int_0^{T} \mathrm{d}t\,  \chi_{\mathscr{P}_{\ell_j}'^{\mathscr{T}_{\ell_j}}}(s)
			\Bigg[
			2^{\kappa_2 - 5 - c_{\mathrm{in}}}\, c_{22}\, C_{\omega}^{2\theta}\, C_{\mathfrak{R}}^{2}\,
			C_{\mathfrak{R}'}  
			\int_{\mathscr{D}_{\rho(s)}^{\Omega_{\ell_j}}} \mathrm{d}\omega\, \mathfrak{F}(s,\omega)\, \omega \\[0.4em]
			& \quad \times \Omega_{\ell_j}^{3\varpi_2 + 2 - 2\theta + \kappa_2 - c_{\mathrm{in}} + \gamma}
			\Bigg]
			> \mathcal{C}_1\, \Omega_{\ell_j}^{\epsilon},
		\end{aligned}
	\end{equation}
	when \( c_{22} \ne 0 \).
\end{itemize}

Suppose that the above assertion holds. Then there exists a constant \( C' \in \mathbb{N} \) such that for all \( \ell > C' \), and for all 
\( \rho \in \left\{ 2^{\Upsilon_{\ell} - 1} - 1,\, \dots,\, 2^{\Upsilon_{\ell}} - 1 \right\} \),
	\begin{equation} \label{Lemma:Mutis:E2a}
	\begin{aligned}
		& \int_0^{T} \mathrm{d}t\, 4^{-1}\, \theta\, C_\omega^{-1} C_{\mathfrak{P}}^{3}\, c_{12}\, 
		\lvert \Omega_{\ell_j} \rvert^{3\theta + 3\varpi_1 + 1}
		\left[ \int_{\mathbb{R}_+} \mathrm{d}\omega\,
		\mathfrak{F}\, \omega\,
		\chi_{\big\{\omega \in \mathscr{D}_{\rho(t)}^{\Omega_{\ell_j}}\big\}} \right]^2
		\chi_{\mathscr{P}_{\ell_j}'^{\mathscr{T}_{\ell_j}}}(t) \\[0.4em]
		& \quad + \int_0^{T} \mathrm{d}t\, \tfrac{1}{2}\, c_{31}\, C_{\mathfrak{Q}}^{3}\, \theta\, C_\omega^{-1} C_{\mathfrak{P}}\, 
		\lvert \Omega_{\ell_j} \rvert^{3\theta + 4\varpi_3 + 1}
		\left[ \int_{\mathbb{R}_+} \mathrm{d}\omega\,
		\mathfrak{F}\, \omega\,
		\chi_{\big\{\omega \in \mathscr{D}_{\rho(t)}^{\Omega_{\ell_j}}\big\}} \right]^3
		\chi_{\mathscr{P}_{\ell_j}'^{\mathscr{T}_{\ell_j}}}(t)
		\le \mathcal{C}_1\, \Omega_{\ell_j}^{\epsilon},
	\end{aligned}
\end{equation}
when \( c_{22} = 0 \), and
\begin{equation} \label{Lemma:Mutis:E2}
	\begin{aligned}
		& \int_0^{T} \mathrm{d}t\,  \chi_{\mathscr{P}_{\ell_j}'^{\mathscr{T}_{\ell_j}}}(s)
		\Bigg[
		2^{\kappa_2 - 5 - c_{\mathrm{in}}}\, c_{22}\, C_{\omega}^{2\theta}\, C_{\mathfrak{R}}^{2}\,
		C_{\mathfrak{R}'}  
		\int_{\mathscr{D}_{\rho(s)}^{\Omega_{\ell_j}}} \mathrm{d}\omega\, \mathfrak{F}(s,\omega)\, \omega \\[0.4em]
		& \quad \times \Omega_{\ell_j}^{3\varpi_2 + 2 - 2\theta + \kappa_2 - c_{\mathrm{in}} + \gamma}
		\Bigg]
		\le \mathcal{C}_1\, \Omega_{\ell_j}^{\epsilon},
	\end{aligned}
\end{equation}
when \( c_{22} \ne 0 \).

Combined with~\eqref{Sec:DD:3}, the bound \eqref{Lemma:Mutis:E2a} implies
\begin{equation} \label{Lemma:Mutis:E3a}
	c_{12}\, \left| \mathscr{P}_{\ell}^T \right|\, \Omega_{\ell}^{-2\sigma+3\theta + 3\varpi_1 + 1}
	+ c_{31}\, \left| \mathscr{P}_{\ell}^T \right|  \Omega_{\ell}^{-3\sigma+3\theta + 4\varpi_3 + 1}
		\le C''\, \Omega_{\ell}^{\epsilon},
\end{equation}
for some constant \( C'' > 0 \).

Combined with~\eqref{Sec:DD:3}, the bound \eqref{Lemma:Mutis:E2} implies
\begin{equation} \label{Lemma:Mutis:E3}
	c_{22}\, \left| \mathscr{P}_{\ell}^T \right|\, \Omega_{\ell}^{-\sigma+3\varpi_2 + 2 - 2\theta + \kappa_2 - c_{\mathrm{in}} + \gamma} 
	\le C''\, \Omega_{\ell}^{\epsilon},
\end{equation}
for some constant \( C'' > 0 \).

We consider the following cases:
\begin{itemize}
	\item If \( c_{12} > 0 \), we bound
	\begin{equation} \label{Lemma:Mutis:E4}
		\left| \mathscr{P}_{\ell}^T \right| \le C_{\mathscr{P}}\, \Omega_{\ell}^{\epsilon + 2\sigma -  3\theta - 3\varpi_1 - 1},
	\end{equation}
	for some universal constant \( C_{\mathscr{P}} \) independent of \( T, T^* \), and the other parameters.
	
	\item If \( c_{22} > 0 \), we bound
	\begin{equation} \label{Lemma:Mutis:E5}
		\left| \mathscr{P}_{\ell}^T \right| \le C_{\mathscr{P}}\, \Omega_{\ell}^{\epsilon + \sigma- 3\varpi_2 - 2 + 2\theta - \kappa_2 + c_{\mathrm{in}} -\gamma},
	\end{equation}
	for some universal constant \( C_{\mathscr{P}} \) independent of \( T, T^* \), and the other parameters.
	
	\item If \( c_{31} > 0 \), we bound
	\begin{equation} \label{Lemma:Mutis:E6}
		\left| \mathscr{P}_{\ell}^T \right| \le C_{\mathscr{P}}\, \Omega_{\ell}^{\epsilon + 3\sigma -  3\theta - 4\varpi_3 - 1},
	\end{equation}
	for some universal constant \( C_{\mathscr{P}} \) independent of \( T, T^* \), and the other parameters.
\end{itemize}

Combining the cases, we obtain \eqref{Lemma:Mutis1:1}.

Now, returning to \eqref{Lemma:Mutis:E1}, suppose for contradiction that \eqref{Lemma:Mutis:E1} holds. Then there must exist a time 
\( \mathscr{T}_{\ell_j} \in [0, T] \) such that
\begin{equation} \label{Lemma:Mutis:E8a}
	\begin{aligned}
		& \int_0^{\mathscr{T}_{\ell_j}} \mathrm{d}t\, 4^{-1}\, \theta\, C_\omega^{-1} C_{\mathfrak{P}}^{3}\, c_{12}\, 
		\lvert \Omega_{\ell_j} \rvert^{3\theta + 3\varpi_1 + 1}
		\left[ \int_{\mathbb{R}_+} \mathrm{d}\omega\,
		\mathfrak{F}\, \omega\,
		\chi_{\big\{\omega \in \mathscr{D}_{\rho(t)}^{\Omega_{\ell_j}}\big\}} \right]^2
		\chi_{\mathscr{P}_{\ell_j}'^{\mathscr{T}_{\ell_j}}}(t) \\[0.4em]
		& \quad + \int_0^{T} \mathrm{d}t\, \tfrac{1}{2}\, c_{31}\, C_{\mathfrak{Q}}^{3}\, \theta\, C_\omega^{-1} C_{\mathfrak{P}}\, 
		\lvert \Omega_{\ell_j} \rvert^{3\theta + 4\varpi_3 + 1}
		\left[ \int_{\mathbb{R}_+} \mathrm{d}\omega\,
		\mathfrak{F}\, \omega\,
		\chi_{\big\{\omega \in \mathscr{D}_{\rho(t)}^{\Omega_{\ell_j}}\big\}} \right]^3
		\chi_{\mathscr{P}_{\ell_j}'^{\mathscr{T}_{\ell_j}}}(t)
		= \mathcal{C}_1\, \Omega_{\ell_j}^{\epsilon},
	\end{aligned}
\end{equation}
when \( c_{22} = 0 \), and
\begin{equation} \label{Lemma:Mutis:E8}
\begin{aligned}
& \int_0^{\mathscr{T}_{\ell_j}} \mathrm{d}s\,  \chi_{\mathscr{P}_{\ell_j}'^{\mathscr{T}_{\ell_j}}}(s)
\Bigg[
2^{\kappa_2 - 5 - c_{\mathrm{in}}}\, c_{22}\, C_{\omega}^{2\theta}\, C_{\mathfrak{R}}^{2}\,
C_{\mathfrak{R}'}  
\int_{\mathscr{D}_{\rho(s)}^{\Omega_{\ell_j}}} \mathrm{d}\omega\, \mathfrak{F}(s,\omega)\, \omega \\[0.4em]
& \quad \times \Omega_{\ell_j}^{3\varpi_2 + 2 - 2\theta + \kappa_2 - c_{\mathrm{in}} + \gamma}
\Bigg]
= \mathcal{C}_1\, \Omega_{\ell_j}^{\epsilon},
\end{aligned}
\end{equation}
when \( c_{22} \ne 0 \).
		
		\subsubsection*{Step 1. \textit{The test function}}

We now define the function
\begin{equation}\label{Lemma:TestFunc:E3}
	\vartheta(t,z) := (z - 7\Omega_{\ell_j}/4)_+ 
	= \max\{\, z - 7\Omega_{\ell_j}/4,\; 0 \,\}.
\end{equation}

Finally, using $\vartheta$ as a test function, and proceeding similarly as in~\eqref{Lemma:Apriori:E8a}, \eqref{Lemma:Apriori:E2}, and \eqref{Lemma:Apriori:E12a}, we find \begin{equation}\label{Lemma:Supersolu:E9}
	\partial_t \!\left( \int_{\mathbb{R}_+} \mathrm{d}\omega\, \mathfrak{F}\,\vartheta \right)
	= \mathfrak{Z}_1 + \mathfrak{Z}_2 + \mathfrak{Z}_3,
\end{equation}
		where
		\begin{equation} \label{Lemma:Growth2:E2}
			\begin{aligned}
				\mathfrak{Z}_1 :=\; & 2c_{12} \iint_{\omega_1 > \omega_2} \mathrm{d}\omega_1 \mathrm{d}\omega_2\, \mathfrak{P}(\omega_1) \mathfrak{P}(\omega_2)\, f(\omega_1) f(\omega_2) \\
				& \quad \times \Big[ \mathfrak{P}(\omega_1 + \omega_2) \big( \vartheta(\omega_1 + \omega_2) - \vartheta(\omega_1) - \vartheta(\omega_2) \big) \\
				& \qquad\quad - \mathfrak{P}(\omega_1 - \omega_2) \big( \vartheta(\omega_1) - \vartheta(\omega_1 - \omega_2) - \vartheta(\omega_2) \big) \Big] \\
				& + c_{12} \iint_{\omega_1 = \omega_2} \mathrm{d}\omega_1 \mathrm{d}\omega_2\, \mathfrak{P}(\omega_1)^2 f(\omega_1)^2\, \mathfrak{P}(2\omega_1)\, \big[ \vartheta(2\omega_1) - 2\vartheta(\omega_1) \big],
			\end{aligned}
		\end{equation}
		
		\begin{equation} \label{Lemma:Growth2:E3}
			\begin{aligned}
				\mathfrak{Z}_2 :=\; & c_{22} \iiint_{\mathbb{R}_+^3} \mathrm{d}\omega_1\,\mathrm{d}\omega_2\,\mathrm{d}\omega\, f_1 f_2 f\,   \\
				& \times \Big\{ 
				[-\vartheta(\omega_{\text{Sup}}) - \vartheta(\omega_{\text{Inf}}) + \vartheta(\omega_{\text{Med}}) + \vartheta(\omega_{\text{Sup}} + \omega_{\text{Inf}} - \omega_{\text{Med}})] \\
				& \qquad \times \mathfrak{R}_o\prod_{x \in \{\omega_{\text{Sup}}, \omega_{\text{Inf}}, \omega_{\text{Med}}, \omega_{\text{Sup}} + \omega_{\text{Inf}} - \omega_{\text{Med}}\}}\mathfrak{R}(x)\, |k_{\text{Inf}}| \\
				& + [-\vartheta(\omega_{\text{Sup}}) - \vartheta(\omega_{\text{Med}}) + \vartheta(\omega_{\text{Inf}}) + \vartheta(\omega_{\text{Sup}} + \omega_{\text{Med}} - \omega_{\text{Inf}})] \\
				& \qquad \times \mathfrak{R}_o\prod_{x \in \{\omega_{\text{Sup}}, \omega_{\text{Med}}, \omega_{\text{Inf}}, \omega_{\text{Sup}} + \omega_{\text{Med}} - \omega_{\text{Inf}}\}}\mathfrak{R}(x)\, |k_{\text{Inf}}| \\
				& + [-\vartheta(\omega_{\text{Inf}}) - \vartheta(\omega_{\text{Med}}) + \vartheta(\omega_{\text{Sup}}) + \vartheta(\omega_{\text{Inf}} + \omega_{\text{Med}} - \omega_{\text{Sup}})] \\
				& \qquad \times \mathfrak{R}_o \mathbf{1}_{\omega_{\text{Inf}} + \omega_{\text{Med}} - \omega_{\text{Sup}} \ge 0} \prod_{x \in \{\omega_{\text{Sup}}, \omega_{\text{Med}}, \omega_{\text{Inf}}, \omega_{\text{Inf}} + \omega_{\text{Med}} - \omega_{\text{Sup}}\}}\mathfrak{R}(x) \\
				& \qquad \times \min \left\{ |k|(\omega_{\text{Sup}}), |k|(\omega_{\text{Inf}}), |k|(\omega_{\text{Med}}), |k|(\omega_{\text{Inf}} + \omega_{\text{Med}} - \omega_{\text{Sup}}) \right\}
				\Big\}
			\end{aligned}
		\end{equation}
		
		and
		\begin{equation} \label{Lemma:Growth2:E4}
			\begin{aligned}
				\mathfrak{Z}_3 :=\; & 3c_{31} \iiint_{\omega_1 > \omega_2 + \omega_3} \mathrm{d}\omega_1\, \mathrm{d}\omega_2\, \mathrm{d}\omega_3\, \mathfrak{Q}(\omega_1) \mathfrak{Q}(\omega_2) \mathfrak{Q}(\omega_3)\, f(\omega_1) f(\omega_2) f(\omega_3) \\
				& \times \Big[ \mathfrak{Q}(\omega_1 + \omega_2 + \omega_3) \big(\vartheta(\omega_1 + \omega_2 + \omega_3) - \vartheta(\omega_1) - \vartheta(\omega_2) - \vartheta(\omega_3) \big) \\
				& \quad - \mathfrak{Q}(\omega_1 - \omega_2 - \omega_3) \big( \vartheta(\omega_1) - \vartheta(\omega_1 - \omega_2 - \omega_3) - \vartheta(\omega_2) - \vartheta(\omega_3) \big) \Big] \\
				& + c_{31} \iint_{\omega_1 = \omega_2 + \omega_3} \mathrm{d}\omega_1\, \mathrm{d}\omega_2\, \mathfrak{Q}(\omega_1) \mathfrak{Q}(\omega_2) \mathfrak{Q}(\omega_3)\, f(\omega_1) f(\omega_2) f(\omega_3) \\
				& \quad \times \mathfrak{Q}(3\omega_1) \big[ \vartheta(3\omega_1) - 3\vartheta(\omega_1) \big] \\
				& + c_{31} \iiint_{\mathbb{R}_+^3 \setminus \left( \{\omega_1 > \omega_2 + \omega_3\} \cup \{\omega_2 > \omega_1 + \omega_3\} \cup \{\omega_3 > \omega_1 + \omega_2\} \right)} \mathrm{d}\omega_1\, \mathrm{d}\omega_2\, \mathrm{d}\omega_3 \\
				& \quad \times \mathfrak{Q}(\omega_1 + \omega_2 + \omega_3)\, \mathfrak{Q}(\omega_1) \mathfrak{Q}(\omega_2) \mathfrak{Q}(\omega_3)\, f(\omega_1) f(\omega_2) f(\omega_3) \\
				& \quad \times \big[ \vartheta(\omega_1 + \omega_2 + \omega_3) - \vartheta(\omega_1) - \vartheta(\omega_2) - \vartheta(\omega_3) \big].
			\end{aligned}
		\end{equation}

		\subsubsection*{Step 2. \textit{Controlling} \(\mathfrak{Z}_1\)}

Similarly as in~\eqref{Lemma:Apriori:E8b}, we can write
\begin{equation}\label{Lemma:Supersolu:E6}
	\begin{aligned}
		& \mathfrak{P}(\omega_1 + \omega_2) \left( \vartheta(\omega_1 + \omega_2) - \vartheta(\omega_1) - \vartheta(\omega_2) \right) \\
		&\quad -\ \mathfrak{P}(\omega_1 - \omega_2) \left( \vartheta(\omega_1) - \vartheta(\omega_1 - \omega_2) - \vartheta(\omega_2) \right) \\
		=\ & \left[\mathfrak{P}(\omega_1 + \omega_2) - \mathfrak{P}(\omega_1 - \omega_2)\right] 
		\int_0^{\omega_2} \int_0^{\omega_1 } \mathrm{d}\zeta\, \mathrm{d}\zeta_0\, \partial_{\omega}^2\vartheta(\zeta + \zeta_0) \\
		& +\ \mathfrak{P}(\omega_1 - \omega_2) 
		\int_0^{\omega_2} \int_{\omega_1 - \omega_2}^{\omega_1} \mathrm{d}\zeta\, \mathrm{d}\zeta_0\, \partial_{\omega}^2\vartheta(\zeta + \zeta_0) \ \ge \ 0.
	\end{aligned}
\end{equation}
Here, the second derivative $\partial_{\omega}^2 \vartheta$ is understood in the weak sense.

Next, we  estimate

\begin{equation*}
	\begin{aligned}
		\mathfrak{Z}_1 \ \ge\ & 2c_{12} \iint_{ \mathbb{R}_+^2} 
		\mathrm{d}\omega_1\, \mathrm{d}\omega_2\, 
	\bar{\mathfrak{P}}(\omega_1)\, 	\bar{\mathfrak{P}}(\omega_2)\, \mathfrak F(\omega_1)\, \mathfrak F(\omega_2)  
		\chi_{\left\{ \omega_1\ge \omega_2;\omega_1,\omega_2 \in \mathscr{D}_{\rho(t)}^{\Omega_{\ell_j}} \right\}}\chi_{\mathscr{P}_{\ell_j}'^{\mathscr{T}_{\ell_j}}}(t) \\
		& \quad \times \Big[ \mathfrak{P}(\omega_1 + \omega_2) \big( \vartheta(\omega_1 + \omega_2) - \vartheta(\omega_1) - \vartheta(\omega_2) \big) \\
		& \qquad\quad - \mathfrak{P}(\omega_1 - \omega_2) \big( \vartheta(\omega_1) - \vartheta(\omega_1 - \omega_2) - \vartheta(\omega_2) \big) \Big] \\
		\ge\ & 2c_{12} \iint_{\mathbb{R}_+^2 } 
		\mathrm{d}\omega_1\, \mathrm{d}\omega_2\, 
		\tilde{\mathfrak{P}}(\omega_1)\, \tilde{\mathfrak{P}}(\omega_2)\, \mathfrak F_1 \mathfrak F_2\,\omega_1\omega_2\chi_{\left\{ \omega_1\ge \omega_2;\omega_1,\omega_2 \in \mathscr{D}_{\rho(t)}^{\Omega_{\ell_j}} \right\}} \chi_{\mathscr{P}_{\ell_j}'^{\mathscr{T}_{\ell_j}}}(t) \\
		& \quad \times \Big[ \mathfrak{P}(\omega_1 + \omega_2) \big( \vartheta(\omega_1 + \omega_2) - \vartheta(\omega_1) - \vartheta(\omega_2) \big) \\
	& \qquad\quad - \mathfrak{P}(\omega_1 - \omega_2) \big( \vartheta(\omega_1) - \vartheta(\omega_1 - \omega_2) - \vartheta(\omega_2) \big) \Big],
	\end{aligned}
\end{equation*}
where we recall that $\mathfrak F_1$, $\mathfrak F_2$ denote  $\mathfrak F(\omega_1)$, $\mathfrak F(\omega_2).$


Using the inequality \( \tilde{\mathfrak{P}}(\omega) \ge C_{\mathfrak{P}} \omega^{\varpi_1} \), we can estimate \( \mathfrak{Z}_1 \) from below as
\begin{equation} \label{Lemma:Mutiscale1:E1}
	\begin{aligned}
		\mathfrak{Z}_1 \ge\ & 2c_{12}   C_{\mathfrak{P}}^2 
		\iint_{\mathbb{R}_+^2} 
		\mathrm{d}\omega_1\, \mathrm{d}\omega_2\,
		\mathfrak{F}_1\,\cdot\omega_1^{\varpi_1+1}\, \mathfrak{F}_2\,\cdot\omega_2^{\varpi_1+1}\,  
		\chi_{\left\{ \omega_1\ge \omega_2;\omega_1,\omega_2 \in \mathscr{D}_{\rho(t)}^{\Omega_{\ell_j}} \right\}} \,
		\chi_{\mathscr{P}_{\ell_j}'^{\mathscr{T}_{\ell_j}}}(t) \\[0.5em]
		& \quad \times \Big[ \mathfrak{P}(\omega_1 + \omega_2) \big( \vartheta(\omega_1 + \omega_2) - \vartheta(\omega_1) - \vartheta(\omega_2) \big) \\[0.5em]
	& \qquad\quad - \mathfrak{P}(\omega_1 - \omega_2) \big( \vartheta(\omega_1) - \vartheta(\omega_1 - \omega_2) - \vartheta(\omega_2) \big) \Big].
	\end{aligned}
\end{equation}

We now write
\begin{equation*} 
	\begin{aligned}
		& 
		\mathfrak{P}(\omega_1 + \omega_2) \big( \vartheta(\omega_1 + \omega_2) - \vartheta(\omega_1) - \vartheta(\omega_2) \big)  - \mathfrak{P}(\omega_1 - \omega_2) \big( \vartheta(\omega_1) - \vartheta(\omega_1 - \omega_2) - \vartheta(\omega_2) \big) \\
		=\ &
		 \mathfrak{P}(\omega_1 + \omega_2)  \vartheta(\omega_1+\omega_2),
	\end{aligned}
\end{equation*}
which, in combination with~\eqref{Lemma:Mutiscale1:E1}, yields
\begin{equation} \label{Lemma:Mutiscale1:E1:1}
	\begin{aligned}
		\mathfrak{Z}_1 \ge\ & 
		2c_{12}   C_{\mathfrak{P}}^2 
		\iint_{\mathbb{R}_+^2} 
		\mathrm{d}\omega_1\, \mathrm{d}\omega_2\,
		\mathfrak{F}_1  \,\cdot\omega_1^{\varpi_1+1}\, \mathfrak{F}_2\,\cdot\omega_2^{\varpi_1+1}\,  
	\chi_{\left\{ \omega_1\ge \omega_2;\omega_1,\omega_2 \in \mathscr{D}_{\rho(t)}^{\Omega_{\ell_j}} \right\}} \,
		\chi_{\mathscr{P}_{\ell_j}'^{\mathscr{T}_{\ell_j}}}(t) \\[0.5em]
		& \times 
	  \mathfrak{P}(\omega_1 + \omega_2)  \vartheta(\omega_1+\omega_2).
	\end{aligned}
\end{equation}

Since  on the domain of integration
\[
\vartheta(\omega_1+\omega_2)  = \omega_1+\omega_2 -7\Omega_\ell/4 \ge \Omega_\ell/4 \ge \omega_2/8,
\]
we continue to estimate~\eqref{Lemma:Mutiscale1:E1:1} as
\begin{equation} \label{Lemma:Mutiscale1:E1:2}
	\begin{aligned}
		\mathfrak{Z}_1 \ge\ & 
		4^{-1}c_{12}   C_{\mathfrak{P}}^2 
		\iint_{\mathbb{R}_+^2} 
		\mathrm{d}\omega_1\, \mathrm{d}\omega_2\,
		\mathfrak{F}_1\,\cdot\omega_1^{\varpi_1+1}\,  \mathfrak{F}_2  \,\cdot\omega_2^{\varpi_1+2}\,
		\chi_{\left\{ \omega_1\ge \omega_2;\omega_1,\omega_2 \in \mathscr{D}_{\rho(t)}^{\Omega_{\ell_j}} \right\}}  \,
		\chi_{\mathscr{P}_{\ell_j}'^{\mathscr{T}_{\ell_j}}}(t) \\[0.5em]
		&\times  \mathfrak{P}(\omega_1 + \omega_2).
	\end{aligned}
\end{equation}

Thanks to Assumption~X, we can then deduce that  
\begin{equation} \label{Lemma:Multiscale1:E1:3}
	\begin{aligned}
		\mathfrak{Z}_1 \ge\ &
		4^{-1}\theta\, C_\omega^{-1} C_{\mathfrak{P}}^{3}\, c_{12}\, \lvert \Omega_\ell \rvert^{3\theta + 3\varpi_1 + 1}
		\left[ \int_{\mathbb{R}_+} \mathrm{d}\omega \,
		\mathfrak{F}\, \omega\,
		\chi_{\big\{\omega \in \mathscr{D}_{\rho(t)}^{\Omega_{\ell_j}}\big\}} \right]^2
		\chi_{\mathscr{P}_{\ell_j}'^{\mathscr{T}_{\ell_j}}}(t).
	\end{aligned}
\end{equation}

		\subsubsection*{Step 3. \textit{Controlling} \(\mathfrak{Z}_2\)}

			Using \eqref{Sec:DDM:6},	similar  as in~\eqref{Lemma:Apriori:E3}, we have the estimate
		\begin{equation} \label{Lemma:Mutiscale1:E3}
			\begin{aligned}
				& \left[ -\vartheta(\omega_{\text{Inf}}) - \vartheta(\omega_{\text{Med}}) + \vartheta(\omega_{\text{Sup}}) + \vartheta(\omega_{\text{Inf}} + \omega_{\text{Med}} - \omega_{\text{Sup}}) \right] \\
				=\, & \int_{0}^{\omega_{\text{Sup}} - \omega_{\text{Inf}}} \mathrm{d}\xi_1 \int_{0}^{\omega_{\text{Sup}} - \omega_{\text{Med}}} \mathrm{d}\xi_2\, 
				\partial_{\omega}^2 \vartheta(\xi_1 + \xi_2 + \omega_{\text{Inf}}) \ \ge\, 0,
			\end{aligned}
		\end{equation}
		since \(\partial_{\omega}^2 \vartheta \ge 0\).
		Similar  as in~\eqref{Lemma:Apriori:E4}, we also have the estimate
		\begin{equation}\label{Lemma:Mutiscale1:E4}
			\begin{aligned}
				& \left[ -\vartheta(\omega_{\text{Sup}}) - \vartheta(\omega_{\text{Inf}}) + \vartheta(\omega_{\text{Med}}) + \vartheta(\omega_{\text{Sup}} + \omega_{\text{Inf}} - \omega_{\text{Med}}) \right] \mathfrak{R}(\omega_{\text{Sup}} + \omega_{\text{Inf}} - \omega_{\text{Med}}) \\
				& + \left[ -\vartheta(\omega_{\text{Sup}}) - \vartheta(\omega_{\text{Med}}) + \vartheta(\omega_{\text{Inf}}) + \vartheta(\omega_{\text{Sup}} + \omega_{\text{Med}} - \omega_{\text{Inf}}) \right] \mathfrak{R}(\omega_{\text{Sup}} + \omega_{\text{Med}} - \omega_{\text{Inf}}) \\
				=\, & \int_{0}^{\omega_{\text{Med}} - \omega_{\text{Inf}}} \mathrm{d}\xi \int_{0}^{\omega_{\text{Sup}} - \omega_{\text{Med}}} \mathrm{d}\xi_0\, 
				\partial_{\omega}^2 \vartheta(\omega_{\text{Inf}} + \xi + \xi_0) \\
				& \quad \times \left[ \mathfrak{R}(\omega_{\text{Sup}} - \omega_{\text{Inf}} + \omega_{\text{Med}}) - \mathfrak{R}(\omega_{\text{Sup}} + \omega_{\text{Inf}} - \omega_{\text{Med}}) \right] \\
				& + \int_{0}^{\omega_{\text{Med}} - \omega_{\text{Inf}}} \mathrm{d}\xi \int_{0}^{\omega_{\text{Med}} - \omega_{\text{Inf}}} \mathrm{d}\xi_0\, 
				\mathfrak{R}(\omega_{\text{Sup}} - \omega_{\text{Inf}} + \omega_{\text{Med}})\, \partial_{\omega}^2 \vartheta(\omega_{\text{Inf}} + \xi + \xi_0) \\
				\ge\, & \int_{0}^{\omega_{\text{Med}} - \omega_{\text{Inf}}} \mathrm{d}\xi \int_{0}^{\omega_{\text{Med}} - \omega_{\text{Inf}}} \mathrm{d}\xi_0\, 
				\mathfrak{R}(\omega_{\text{Sup}} - \omega_{\text{Inf}} + \omega_{\text{Med}})\, \partial_{\omega}^2 \vartheta(\omega_{\text{Inf}} + \xi + \xi_0) \\
				\ge\ & 0.
			\end{aligned}
		\end{equation}

	Using \eqref{Lemma:Mutiscale1:E3} and \eqref{Lemma:Mutiscale1:E4}, we restrict the domain of integration of \( \mathfrak{Z}_2 \) to \( \omega_2 \ge \omega_1 \ge \omega \), and we bound
	\begin{equation} \label{Lemma:Mutiscale1:E5}
		\begin{aligned}
			\mathfrak{Z}_2 
			\ge\ &   c_{22} \iiint_{\mathbb{R}_+^3} \mathrm{d}\omega_1\, \mathrm{d}\omega_2\, \mathrm{d}\omega\,
			f_1 f_2 f\, \mathbf{1}_{ \omega_2 \ge \omega_1 \ge \omega} 
			\\
			& \times \Big\{ [-\vartheta(\omega) - \vartheta(\omega_1) + \vartheta(\omega_2) + \vartheta(\omega + \omega_1 - \omega_2)] \\
			& \qquad \times \mathfrak{R}_o \mathfrak{R}(\omega)\mathfrak{R}(\omega_2)\mathfrak{R}(\omega_1)\mathfrak{R}(\omega + \omega_1 - \omega_2)\min\{|k_3|, |k|\} \mathbf{1}_{\omega + \omega_1 - \omega_2>0}\\
			& + [-\vartheta(\omega) - \vartheta(\omega_2) + \vartheta(\omega_1) + \vartheta(\omega + \omega_2 - \omega_1)] \\
			& \qquad \times \mathfrak{R}_o \mathfrak{R}(\omega)\mathfrak{R}(\omega_2)\mathfrak{R}(\omega_1)\mathfrak{R}(\omega + \omega_2 - \omega_1)\, |k| \\
			& + [-\vartheta(\omega_2) - \vartheta(\omega_1) + \vartheta(\omega) + \vartheta(\omega_2 + \omega_1 - \omega)] \\
			& \qquad \times \mathfrak{R}_o \mathfrak{R}(\omega)\mathfrak{R}(\omega_2)\mathfrak{R}(\omega_1)\mathfrak{R}(\omega_2 + \omega_1 - \omega)\, |k| \Big\},
		\end{aligned}
	\end{equation}
	where \( |k_3| \) denotes the wavevector associated with \( \omega + \omega_2 - \omega_1 \).

We can now bound \eqref{Lemma:Mutiscale1:E5} from below as follows:
\begin{equation*}
	\begin{aligned}
		\mathfrak{Z}_2 
		\ge\ & c_{22} \iiint_{\mathbb{R}_+^3} \mathrm{d}\omega_1\, \mathrm{d}\omega_2\, \mathrm{d}\omega\,
		\frac{\mathfrak{F}(\omega)}{|k|} \cdot \frac{\mathfrak{F}(\omega_2)}{|k_2|} \cdot \frac{\mathfrak{F}(\omega_1)}{|k_1|}\,
		\mathbf{1}_{\omega_2 \ge \omega_1 \ge \omega} \\
	& \times \Big\{ [-\vartheta(\omega) - \vartheta(\omega_1) + \vartheta(\omega_2) + \vartheta(\omega + \omega_1 - \omega_2)] \\
	& \qquad \times \mathfrak{R}_o \mathfrak{R}(\omega)\mathfrak{R}(\omega_2)\mathfrak{R}(\omega_1)\mathfrak{R}(\omega + \omega_1 - \omega_2)\min\{|k_3|, |k|\} \mathbf{1}_{\omega + \omega_1 - \omega_2>0}\\
	& + [-\vartheta(\omega) - \vartheta(\omega_2) + \vartheta(\omega_1) + \vartheta(\omega + \omega_2 - \omega_1)] \\
	& \qquad \times\mathfrak{R}_o  \mathfrak{R}(\omega)\mathfrak{R}(\omega_2)\mathfrak{R}(\omega_1)\mathfrak{R}(\omega + \omega_2 - \omega_1)\, |k| \\
	& + [-\vartheta(\omega_2) - \vartheta(\omega_1) + \vartheta(\omega) + \vartheta(\omega_2 + \omega_1 - \omega)] \\
	& \qquad \times\mathfrak{R}_o  \mathfrak{R}(\omega)\mathfrak{R}(\omega_2)\mathfrak{R}(\omega_1)\mathfrak{R}(\omega_2 + \omega_1 - \omega)\, |k| \Big\},
	\end{aligned}
\end{equation*}

which can be rewritten as
\begin{equation} \label{Lemma:Mutiscale1:E6}
	\begin{aligned}
		\mathfrak{Z}_2 
		\ge\ & c_{22} \iiint_{\mathbb{R}_+^3} \mathrm{d}\omega_1\, \mathrm{d}\omega_2\, \mathrm{d}\omega\,
		\frac{\mathfrak{F}\cdot\omega\, \mathfrak{F}_1\cdot\omega_1\, \mathfrak{F}_2\cdot\omega_2}{|k_2|\, |k_1|}\, \mathbf{1}_{\omega_2 \ge \omega_1 \ge \omega}\,
		\tilde{\mathfrak{R}}(\omega)\, \tilde{\mathfrak{R}}(\omega_1)\, \tilde{\mathfrak{R}}(\omega_2) \\
		& \times \Big\{
		\mathfrak{R}_o [-\vartheta(\omega) - \vartheta(\omega_1) + \vartheta(\omega_2) + \vartheta(\omega + \omega_1 - \omega_2)]\,
		\mathfrak{R}(\omega + \omega_1 - \omega_2)\, \frac{\min\{|k_3|, |k|\}}{|k|} \mathbf{1}_{\omega + \omega_1 - \omega_2>0}\\
		& + \mathfrak{R}_o [-\vartheta(\omega) - \vartheta(\omega_2) + \vartheta(\omega_1) + \vartheta(\omega + \omega_2 - \omega_1)]\,
		\mathfrak{R}(\omega + \omega_2 - \omega_1) \\
		& +\mathfrak{R}_o  [-\vartheta(\omega_2) - \vartheta(\omega_1) + \vartheta(\omega) + \vartheta(\omega_2 + \omega_1 - \omega)]\,
		\mathfrak{R}(\omega_2 + \omega_1 - \omega)
		\Big\},
	\end{aligned}
\end{equation}
where we write \( \mathfrak{F}_1 := \mathfrak{F}(\omega_1) \), \( \mathfrak{F}_2 := \mathfrak{F}(\omega_2) \), and \( \mathfrak{F} := \mathfrak{F}(\omega) \).

Under our original assumptions, namely that
\[
\omega(k) \ge C_{\omega} |k|^{1/\theta} \quad \text{and} \quad \tilde{\mathfrak{R}}(\omega) \ge C_{{\mathfrak{R}}} \omega^{\varpi_2},
\]
we obtain the following lower bound for \eqref{Lemma:Mutiscale1:E6}
\begin{equation} \label{Lemma:Mutiscale1:E7}
	\begin{aligned}
		\mathfrak{Z}_2 
		\ge\ & c_{22} C_{\omega}^{2\theta}    C_{{\mathfrak{R}}}^3 
		\iiint_{\mathbb{R}_+^3}
		\mathrm{d}\omega_1\, \mathrm{d}\omega_2\, \mathrm{d}\omega\,
		\Omega_{\ell_j}^{3\varpi_2   - 2\theta+2}\,
		\mathfrak{F}\, \mathfrak{F}_1\cdot\omega_2\, \mathfrak{F}_2\\[0.5em]
		& \times \chi_{\left\{ \omega_2\ge \omega_1 \ge 	2\Omega_{\ell_j} \right\}} \,
	\chi_{\left\{  \omega \in \mathscr{D}_{\rho(t)}^{\Omega_{\ell_j}} \right\}} \,
	\chi_{\mathscr{P}_{\ell_j}'^{\mathscr{T}_{\ell_j}}}(t) \\[0.5em]
		& \times \Big\{ 
		[-\vartheta(\omega) - \vartheta(\omega_2) + \vartheta(\omega_1) + \vartheta(\omega + \omega_2 - \omega_1)]\, 
		\mathfrak{R}(\omega + \omega_2 - \omega_1)\omega_2^\gamma \\[0.25em]
		& \quad + 
		[-\vartheta(\omega_2) - \vartheta(\omega_1) + \vartheta(\omega) + \vartheta(\omega_2 + \omega_1 - \omega)]\, 
		\mathfrak{R}(\omega_2 + \omega_1 - \omega)(\omega_2 + \omega_1 - \omega)^\gamma
		\Big\},
	\end{aligned}
\end{equation}
where we have eliminated the quantity containing $	[-\vartheta(\omega) - \vartheta(\omega_1) + \vartheta(\omega_2) + \vartheta(\omega + \omega_1 - \omega_2)]\,
\mathfrak{R}(\omega + \omega_1 - \omega_2)\, \frac{\min\{|k_3|, |k|\}}{|k|} \mathbf{1}_{\omega + \omega_1 - \omega_2>0}$ and we have used the fact that $0  <  \ 3\varpi_2 + 2 - 2\theta$ in \eqref{X4}.

Since \(  \omega \in \mathscr{D}_{\rho(t)}^{\Omega_{\ell_j}} \), we deduce that 
\[
\vartheta(\omega)  = 0.
\]
Combining this with \eqref{Lemma:Mutiscale1:E7}, we obtain
\begin{equation} \label{Lemma:Mutiscale1:E8}
	\begin{aligned}
		\mathfrak{Z}_2 
		\ge\ & c_{22} C_{\omega}^{2\theta}     C_{{\mathfrak{R}}}^3 
		\iiint_{\mathbb{R}_+^3}
		\mathrm{d}\omega_1\, \mathrm{d}\omega_2\, \mathrm{d}\omega\,
		\Omega_{\ell_j}^{3\varpi_2   - 2\theta+2}
		\, \mathfrak{F} \, \mathfrak{F}_1\cdot\omega_2\, \mathfrak{F}_2\\[0.5em]
		& \times \chi_{\left\{ \omega_2 \ge \omega_1 \ge	2\Omega_{\ell_j} \right\}} \,
		\chi_{\left\{ \omega \in \mathscr{D}_{\rho(t)}^{\Omega_{\ell_j}} \right\}} \,
		\chi_{\mathscr{P}_{\ell_j}'^{\mathscr{T}_{\ell_j}}}(t) \\[0.5em]
		& \times \Big\{ 
		\left[ -\vartheta(\omega_2) + \vartheta(\omega_1)  + \vartheta(\omega + \omega_2 - \omega_1) \right] 
		\mathfrak{R}(\omega + \omega_2 - \omega_1) \omega_2^\gamma\\[0.25em]
		& \qquad + 
		\left[ -\vartheta(\omega_2) - \vartheta(\omega_1)  + \vartheta(\omega_1 + \omega_2 - \omega) \right] 
		\mathfrak{R}(\omega_1 + \omega_2 - \omega) (\omega_2 + \omega_1 - \omega)^\gamma
		\Big\},
	\end{aligned}
\end{equation}
which  can be bounded as
\begin{equation} \label{Lemma:Mutiscale1:E8:1}
	\begin{aligned}
		\mathfrak{Z}_2 
		\ge\ & c_{22} C_{\omega}^{2\theta}    C_{{\mathfrak{R}}}^2 
		\iiint_{\mathbb{R}_+^3}
		\mathrm{d}\omega_1\, \mathrm{d}\omega_2\, \mathrm{d}\omega\,
		\Omega_{\ell_j}^{3\varpi_2   - 2\theta+2}
		\, \mathfrak{F} \, \mathfrak{F}_1\cdot\omega_2\, \mathfrak{F}_2\\[0.5em]
		& \times \chi_{\left\{ \omega_2 \ge \omega_1 \ge 	2\Omega_{\ell_j} \right\}} \,
		\chi_{\left\{  \omega \in \mathscr{D}_{\rho(t)}^{\Omega_{\ell_j}} \right\}} \,
		\chi_{\mathscr{P}_{\ell_j}'^{\mathscr{T}_{\ell_j}}}(t) \\[0.5em]
		& \times  
		\left[ -\vartheta(\omega_2)- \vartheta(\omega_1) + \vartheta(\omega_1 + \omega_2 - \omega) \right] 
		\Big[\mathfrak{R}(\omega_1 + \omega_2 - \omega) (\omega_2 + \omega_1 - \omega)^\gamma - \mathfrak{R}(\omega + \omega_2 - \omega_1) \omega_2^\gamma\Big].
	\end{aligned}
\end{equation}

We next establish a bound, over the domain of integration of \eqref{Lemma:Mutiscale1:E8:1},  
\begin{equation*}
	\begin{aligned}
		& \left[ -\vartheta(\omega_2) - \vartheta(\omega_1) + \vartheta(\omega_1 + \omega_2 - \omega) \right] 
		\Big[ \mathfrak{R}(\omega_1 + \omega_2 - \omega) - \mathfrak{R}(\omega + \omega_2 - \omega_1) \Big] \omega_2^\gamma \\
		\ge\ & (7\Omega_{\ell_j}/4 - \omega)
		\Big[ \mathfrak{R}(\omega_1 + \omega_2 - \omega) - \mathfrak{R}(\omega + \omega_2 - \omega_1) \Big] \omega_2^\gamma  \\
		\ge\ & \frac{\Omega_{\ell_j}}{4}
		\int_{\omega + \omega_2 - \omega_1}^{\omega_1 + \omega_2 - \omega} \mathrm{d}\xi\, \mathfrak{R}'(\xi)\, \omega_2^\gamma  \ 
		\ge\    \frac{\omega}{8}
		\int_{\omega_2}^{\omega_1 + \omega_2 - \omega} \mathrm{d}\xi\, \mathfrak{R}'(\xi)\, \omega_2^\gamma \\
		\ge\ &	C_{\mathfrak{R}'}\, (\omega_1 - \omega)\frac{\omega}{8}\, (2\omega_2)^{\kappa_2} \omega_2^\gamma  \ 
		\ge\    C_{\mathfrak{R}'}\,  \frac{\omega_1\omega}{16}\, (2\omega_2)^{\kappa_2} \omega_2^\gamma  \\
		\ge\ &  C_{\mathfrak{R}'}\, 2^{-4+\kappa_2}\, \omega_1\omega\, \omega_2^{\gamma+\kappa_2},
	\end{aligned}
\end{equation*}

which, in combination with \eqref{Lemma:Mutiscale1:E8:1}, yields
\begin{equation} \label{Lemma:Mutiscale1:E8:2}
	\begin{aligned}
		\mathfrak{Z}_2 
		\ge\ & 2^{\kappa_2-4} c_{22} C_{\omega}^{2\theta}   C_{\mathfrak{R}}^2 
		C_{\mathfrak{R}'} 
		\iiint_{\mathbb{R}_+^3}
		\mathrm{d}\omega_1\, \mathrm{d}\omega_2\, \mathrm{d}\omega\, 
		\Omega_{\ell_j}^{3\varpi_2 + 2 - 2\theta + \kappa_2+\gamma}\,
		\mathfrak{F}\cdot\omega\, \mathfrak{F}_1\cdot\omega_1\, \mathfrak{F}_2 \\[0.5em]
		& \times \chi_{\left\{ \omega_2 \ge \omega_1 \ge \Omega_{\ell_j} \right\}} \,
		\chi_{\left\{ \omega \in \mathscr{D}_{\rho(t)}^{\Omega_{\ell_j}} \right\}} \,
		\chi_{\mathscr{P}_{\ell_j}'^{\mathscr{T}_{\ell_j}}}(t)\, \vartheta(\omega_2) \\[0.75em]
		\ge\ &  2^{\kappa_2-4} c_{22} C_{\omega}^{2\theta}   C_{\mathfrak{R}}^2 
		C_{\mathfrak{R}'}  
		\int_{\mathbb{R}_+}
		\mathrm{d}\omega\, 
		\Omega_{\ell_j}^{3\varpi_2 + 2 - 2\theta + \kappa_2+\gamma}\,
		\mathfrak{F}\omega\, 
		\chi_{\left\{ \omega \in \mathscr{D}_{\rho(t)}^{\Omega_{\ell_j}} \right\}} \,
		\chi_{\mathscr{P}_{\ell_j}'^{\mathscr{T}_{\ell_j}}}(t)\, 
		\left( \int_{\mathbb{R}_+} \mathrm{d}\omega\, \mathfrak{F}\, \vartheta(\omega) \right)^2,
	\end{aligned}
\end{equation}
where we have used the fact that $\gamma+\kappa_2\ge 0$ in \eqref{X4}.

By \eqref{Lemma:Apriori2:1}, we bound
\begin{equation*}
	\begin{aligned}
		\int_{\mathbb{R}_+} \mathrm{d}\omega\, \mathfrak{F}(t,\omega)\, \vartheta(\omega) 
		\ge \int_{\mathbb{R}_+} \mathrm{d}\omega\, \mathfrak{F}(0,\omega)\, \vartheta(\omega) \\
		\ge \frac{1}{2} \int_{[2\Omega_\ell, \infty)} \mathrm{d}\omega\, \mathfrak{F}(0,\omega)\, \omega\, C_{\mathrm{in}} \
		\ge \frac{1}{2}(2\Omega_\ell)^{-c_{\mathrm{in}}},
	\end{aligned}
\end{equation*}
which implies
\begin{equation} \label{Lemma:Mutiscale1:E8:3}
	\begin{aligned}
		\mathfrak{Z}_2 
		\ge\ & 2^{\kappa_2-5-c_{\mathrm{in}}} c_{22} C_{\omega}^{2\theta}   C_{\mathfrak{R}}^2 
		C_{\mathfrak{R}'}  
		\int_{\mathbb{R}_+}
		\mathrm{d}\omega\, 
		\Omega_{\ell_j}^{3\varpi_2 + 2 - 2\theta + \kappa_2 - c_{\mathrm{in}}+\gamma}\,
		\mathfrak{F}\omega \\[0.5em]
		&\times \chi_{\left\{ \omega \in \mathscr{D}_{\rho(t)}^{\Omega_{\ell_j}} \right\}} \,
		\chi_{\mathscr{P}_{\ell_j}'^{\mathscr{T}_{\ell_j}}}(t)\,
		\left( \int_{\mathbb{R}_+} \mathrm{d}\omega\, \mathfrak{F}\, \vartheta(\omega) \right).
	\end{aligned}
\end{equation}

		\subsubsection*{Step 4. \textit{Controlling}  \(\mathfrak{Z}_3\)}	

		Similar as in~\eqref{Lemma:Apriori:E13a}, we can bound
	\begin{equation}
		\begin{aligned}\label{Lemma:Supersolu:E7}
			& \mathfrak{Q}(\omega_1 + \omega_2 + \omega_3) \left(\vartheta(\omega_1 + \omega_2 + \omega_3) - \vartheta(\omega_1) - \vartheta(\omega_2) - \vartheta(\omega_3) \right) \\
			& \quad - \mathfrak{Q}(\omega_1 - \omega_2 - \omega_3) \left( \vartheta(\omega_1) - \vartheta(\omega_1 - \omega_2 - \omega_3) - \vartheta(\omega_2) - \vartheta(\omega_3) \right) \\
			=\ & \left[\mathfrak{Q}(\omega_1 + \omega_2 + \omega_3) - \mathfrak{Q}(\omega_1 - \omega_2 - \omega_3)\right] \\
			& \quad \times \left( \int_0^{\omega_1 - \omega_2 - \omega_3} \int_0^{\omega_2 + \omega_3} \mathrm{d}\zeta\, \mathrm{d}\zeta_0\, \partial_{\omega}^2 \vartheta(\zeta + \zeta_0) 
			+ \int_0^{\omega_2} \int_0^{\omega_3} \mathrm{d}\zeta\, \mathrm{d}\zeta_0\, \partial_{\omega}^2 \vartheta(\zeta + \zeta_0) \right) \\
			& \quad + \mathfrak{Q}(\omega_1 + \omega_2 + \omega_3) \int_{\omega_1 - \omega_2 - \omega_3}^{\omega_1} \int_0^{\omega_2 + \omega_3} 
			\mathrm{d}\zeta\, \mathrm{d}\zeta_0\, \partial_{\omega}^2 \vartheta(\zeta + \zeta_0) \ge 0,
		\end{aligned}
	\end{equation}
	since \(\partial_{\omega}^2 \vartheta(\zeta + \zeta_0) \ge 0\), and noting that 
	\(\mathfrak{Q}(\omega_1 + \omega_2 + \omega_3) - \mathfrak{Q}(\omega_1 - \omega_2 - \omega_3) \ge 0\).

Using~\eqref{Lemma:Supersolu:E7}, we can restrict the domain of integration onto the domain of $\omega_1,\omega_2, \omega_3 \in \mathscr{D}_{\rho(t)}^{\Omega_{\ell_j}} $   and  estimate
\begin{equation*} 
	\begin{aligned}
		\mathfrak{Z}_3 \ \ge\; & 3c_{31} \iiint_{ \mathbb{R}_+^3} 
		\mathrm{d}\omega_1\, \mathrm{d}\omega_2\, \mathrm{d}\omega_3\, 
		\mathfrak{Q}(\omega_1)\, \mathfrak{Q}(\omega_2)\, \mathfrak{Q}(\omega_3)\, f(\omega_1)\, f(\omega_2)\, f(\omega_3) \,
		\chi_{\left\{\omega_1, \omega_2, \omega_3 \in \mathscr{D}_{\rho(t)}^{\Omega_{\ell_j}} \right\}} \\
		& \times 	\chi_{\mathscr{P}_{\ell_j}'^{\mathscr{T}_{\ell_j}}}(t) 
	  \mathfrak{Q}(\omega_1 + \omega_2 + \omega_3) 
	  \vartheta(\omega_1 + \omega_2 + \omega_3)   \\
		\ge\; & 3c_{31} \iiint_{ \mathbb{R}_+^3 } 
		\mathrm{d}\omega_1\, \mathrm{d}\omega_2\, \mathrm{d}\omega_3\, 
		\tilde{\mathfrak{Q}}(\omega_1)\, \tilde{\mathfrak{Q}}(\omega_2)\, \tilde{\mathfrak{Q}}(\omega_3)\, 
		\mathfrak{F}_1\cdot\omega_1\, \mathfrak{F}_2\cdot\omega_2\, \mathfrak{F}_3\cdot \omega_3  \\
		& \times \chi_{\left\{ \omega_1,\omega_2, \omega_3 \in \mathscr{D}_{\rho(t)}^{\Omega_{\ell_j}} \right\}} 
			\chi_{\mathscr{P}_{\ell_j}'^{\mathscr{T}_{\ell_j}}}(t)
		  \mathfrak{Q}(\omega_1 + \omega_2 + \omega_3) 
		 \vartheta(\omega_1 + \omega_2 + \omega_3),
	\end{aligned}
\end{equation*}
noticing that $\vartheta(\omega_1-\omega_2-\omega_3) =\vartheta(\omega_1) =\vartheta(\omega_2) =\vartheta(\omega_3)=0$. 

Using the inequality \( \tilde{\mathfrak{Q}}(\omega) \ge C_{\mathfrak{Q}} \omega^{\varpi_3} \), we can now estimate \( \mathfrak{Z}_3 \) as  
\begin{equation} \label{Lemma:Multiscale1:E2}
	\begin{aligned}
		\mathfrak{Z}_3 \ge\; & 3c_{31}\, C_{\mathfrak{Q}}^{3} 
		\iiint_{\mathbb{R}_+^3} 
		\mathrm{d}\omega_1\, \mathrm{d}\omega_2\, \mathrm{d}\omega_3\, 
		\mathfrak{F}_1\cdot \omega_1^{1+\varpi_3}\, 
		\mathfrak{F}_2\cdot \omega_2^{1+\varpi_3}\, 
		\mathfrak{F}_3\cdot \omega_3^{1+\varpi_3} \\
		& \times 
		\chi_{\big\{\omega_1, \omega_2, \omega_3 \in \mathscr{D}_{\rho(t)}^{\Omega_{\ell_j}}\big\}} 
		\chi_{\mathscr{P}_{\ell_j}'^{\mathscr{T}_{\ell_j}}}(t)
		\Big[
		\mathfrak{Q}(\omega_1 + \omega_2 + \omega_3)\,
		\big(\omega_1 + \omega_2 + \omega_3 - \tfrac{7}{4}\Omega_\ell\big)
		\Big].
	\end{aligned}
\end{equation}

We now establish a bound over the domain of integration, using the equality in~\eqref{Lemma:Supersolu:E7},  
\begin{equation*}
	\begin{aligned}
		& \mathfrak{Q}(\omega_1 + \omega_2 + \omega_3) 
		\big(\omega_1 + \omega_2 + \omega_3 - \tfrac{7}{4}\Omega_\ell\big)  \\[0.5em]
		\ge\; & \tfrac{1}{6}\theta\, C_\omega^{-1} C_{\mathfrak{P}} (\omega_1 + \omega_2 + \omega_3)^{3\theta + \varpi_3}\, |\Omega_\ell| 
		\ \ge\ \tfrac{1}{6}\theta\, C_\omega^{-1} C_{\mathfrak{P}} (\omega_1 + \omega_2 + \omega_3)^{3\theta + \varpi_3 + 1},
	\end{aligned}
\end{equation*}
which, in combination with \eqref{Lemma:Multiscale1:E2}, gives

	\begin{equation} \label{Lemma:Mutiscale1:E2:1}
		\begin{aligned}
			\mathfrak{Z}_3 \ \ge\; & 
			\frac12 c_{31} C_{\mathfrak{Q}}^{3}   \theta\, C_\omega^{-1} C_{\mathfrak{P}} |\Omega_\ell|^{3\theta + 4\varpi_3 + 1}
			\left[ \int_{\mathbb{R}_+} \mathrm{d}\omega \,
		\mathfrak{F}\, \omega\,
		\chi_{\big\{\omega \in \mathscr{D}_{\rho(t)}^{\Omega_{\ell_j}}\big\}} \right]^3
		\chi_{\mathscr{P}_{\ell_j}'^{\mathscr{T}_{\ell_j}}}(t).
		\end{aligned}
	\end{equation}

	\subsubsection*{Step 5. \textit{Putting together the controls of $\mathfrak{Z}_1$, $\mathfrak{Z}_2$, and $\mathfrak{Z}_3$}}

We first consider the case \( c_{22} = 0 \).  
Combining \eqref{Lemma:Supersolu:E9}, \eqref{Lemma:Multiscale1:E1:3}, and \eqref{Lemma:Mutiscale1:E2:1}, we obtain the following estimate:
\begin{equation*} \label{Lemma:Supersolu:E9:1}
	\begin{aligned}
		\partial_t \!\left( \int_{\mathbb{R}_+} \mathrm{d}\omega\, \mathfrak{F}\, \vartheta \right)
		\ge\; & 4^{-1}\, \theta\, C_\omega^{-1} C_{\mathfrak{P}}^{3}\, c_{12}\, \lvert \Omega_\ell \rvert^{3\theta + 3\varpi_1 + 1}
		\left[ \int_{\mathbb{R}_+} \mathrm{d}\omega\,
		\mathfrak{F}\, \omega\,
		\chi_{\big\{\omega \in \mathscr{D}_{\rho(t)}^{\Omega_{\ell_j}}\big\}} \right]^2
		\chi_{\mathscr{P}_{\ell_j}'^{\mathscr{T}_{\ell_j}}}(t) \\[0.4em]
		& + \tfrac{1}{2}\, c_{31}\, C_{\mathfrak{Q}}^{3}\, \theta\, C_\omega^{-1} C_{\mathfrak{P}}\, 
		\, \lvert \Omega_\ell \rvert^{3\theta + 4\varpi_3 + 1}
		\left[ \int_{\mathbb{R}_+} \mathrm{d}\omega\,
		\mathfrak{F}\, \omega\,
		\chi_{\big\{\omega \in \mathscr{D}_{\rho(t)}^{\Omega_{\ell_j}}\big\}} \right]^3
		\chi_{\mathscr{P}_{\ell_j}'^{\mathscr{T}_{\ell_j}}}(t).
	\end{aligned}
\end{equation*}
Integrating both sides of the above inequalities yields  
\begin{equation} \label{Lemma:Supersolu:E9:1}
	\begin{aligned}
		\mathfrak{M} + \mathfrak{E} 
		\ge\; & \int_{\mathbb{R}_+} \mathrm{d}\omega\, \mathfrak{F}(\mathscr{T}_{\ell_j})\, \vartheta \\[0.3em]
		\ge\; & \int_0^{\mathscr{T}_{\ell_j}} \mathrm{d}t\,4^{-1}\, \theta\, C_\omega^{-1} C_{\mathfrak{P}}^{3}\, c_{12}\, 
		\lvert \Omega_\ell \rvert^{3\theta + 3\varpi_1 + 1}
		\left[ \int_{\mathbb{R}_+} \mathrm{d}\omega\,
		\mathfrak{F}\, \omega\,
		\chi_{\big\{\omega \in \mathscr{D}_{\rho(t)}^{\Omega_{\ell_j}}\big\}} \right]^2
		\chi_{\mathscr{P}_{\ell_j}'^{\mathscr{T}_{\ell_j}}}(t) \\[0.4em]
		& + \int_0^{\mathscr{T}_{\ell_j}} \mathrm{d}t\,\tfrac{1}{2}\, c_{31}\, C_{\mathfrak{Q}}^{3}\, \theta\, C_\omega^{-1} C_{\mathfrak{P}}\, 
		\lvert \Omega_\ell \rvert^{3\theta + 4\varpi_3 + 1}
		\left[ \int_{\mathbb{R}_+} \mathrm{d}\omega\,
		\mathfrak{F}\, \omega\,
		\chi_{\big\{\omega \in \mathscr{D}_{\rho(t)}^{\Omega_{\ell_j}}\big\}} \right]^3
		\chi_{\mathscr{P}_{\ell_j}'^{\mathscr{T}_{\ell_j}}}(t) \\[0.4em]
		=\; & \mathcal{C}_1\, \Omega_{\ell_j}^{\epsilon} \;\to\; \infty,
	\end{aligned}
\end{equation}
as \( j \to \infty \).  
This yields a contradiction.

	We next consider the case \( c_{22} \ne 0 \).  
Combining \eqref{Lemma:Supersolu:E9} and \eqref{Lemma:Mutiscale1:E8:3}, we obtain the following estimate:
\begin{equation} \label{Lemma:Mutiscale1:E9}
	\begin{aligned}
		\partial_t\left( \int_{\mathbb{R}^{+}} \mathrm{d}\omega\, \mathfrak{F} \vartheta \right)
		\ge\; &  2^{\kappa_2-5-c_{\mathrm{in}}} c_{22} C_{\omega}^{2\theta}   C_{\mathfrak{R}}^2 
		C_{\mathfrak{R}'}  
		\int_{\mathbb{R}_+}
		\mathrm{d}\omega\, 
		\Omega_{\ell_j}^{3\varpi_2 + 2 - 2\theta + \kappa_2 - c_{\mathrm{in}}+\gamma}\,
		\mathfrak{F}\omega \\[0.5em]
		&\times \chi_{\left\{ \omega \in \mathscr{D}_{\rho(t)}^{\Omega_{\ell_j}} \right\}} \,
		\chi_{\mathscr{P}_{\ell_j}'^{\mathscr{T}_{\ell_j}}}(t)\,
		\left( \int_{\mathbb{R}_+} \mathrm{d}\omega\, \mathfrak{F}\, \vartheta(\omega) \right).
	\end{aligned}
\end{equation}

Solving the differential inequality \eqref{Lemma:Mutiscale1:E9} yields
\begin{equation} \label{Lemma:Mutiscale1:E10}
	\begin{aligned}
		\int_{\mathbb{R}^{+}} \mathrm{d}\omega\, \mathfrak{F}(t,\omega)\, \vartheta(\omega)
		\ge\ & \int_{\mathbb{R}^{+}} \mathrm{d}\omega\, \mathfrak{F}(0,\omega)\, \vartheta(\omega) \\
		&\times \exp\Bigg(
		\int_0^t \mathrm{d}s\, \chi_{\mathscr{P}_{\ell_j}'^{\mathscr{T}_{\ell_j}}}(s) \Bigg[2^{\kappa_2-5-c_{\mathrm{in}}} c_{22} C_{\omega}^{2\theta}   C_{\mathfrak{R}}^2 
		C_{\mathfrak{R}'}  
		\int_{\mathscr{D}_{\rho(s)}^{\Omega_{\ell_j}}} \mathrm{d}\omega\, \mathfrak{F}(s,\omega)\, \omega \\
		&\quad \times \Omega_{\ell_j}^{3\varpi_2 + 2 - 2\theta + \kappa_2 - c_{\mathrm{in}}+\gamma}
		\Bigg] \Bigg).
	\end{aligned}
\end{equation}

Using~\eqref{Theorem1:4} and ~\eqref{Lemma:Apriori2:1}, we obtain the bound, for $j$ sufficiently large,
\begin{equation}
	\label{Lemma:Mutiscale1:E11}
	\int_{\mathbb{R}_+} \mathrm{d}\omega\, \mathfrak{F}(0,\omega)\, \vartheta(\omega)   
	\geq \frac{1}{20} \int_{2\Omega_{\ell_j}}^\infty \mathrm{d}\omega\, \mathfrak{F}(0,\omega)\, \omega   
	\geq \frac{1}{20} C_{\mathrm{ini}} \left(2\Omega_{\ell_j} \right)^{-c_{\mathrm{in}}}.
\end{equation}

Substituting~\eqref{Lemma:Mutiscale1:E11} into~\eqref{Lemma:Mutiscale1:E10}, we obtain the following estimate for \( 0 \le t \le \mathscr{T}_{\ell_j} \):  
\begin{equation}\label{Lemma:Multiscale1:E12}
	\begin{aligned}
		\int_{\mathbb{R}^{+}} \mathrm{d}\omega\, \mathfrak{F}(t,\omega)\, \vartheta(\omega)
		\ge\; & \tfrac{1}{20}\, C_{\mathrm{ini}}\, (2\Omega_{\ell_j})^{-c_{\mathrm{in}}}
		\exp\Bigg(
		\int_0^t \mathrm{d}s\, \chi_{\mathscr{P}_{\ell_j}'^{\mathscr{T}_{\ell_j}}}(s)
		\Bigg[ \\[0.4em]
		& \quad 2^{\kappa_2 - 5 - c_{\mathrm{in}}}\, c_{22}\, C_{\omega}^{2\theta}\, C_{\mathfrak{R}}^{2}\,
		C_{\mathfrak{R}'} 
		\int_{\mathscr{D}_{\rho(s)}^{\Omega_{\ell_j}}} \mathrm{d}\omega\, 
		\mathfrak{F}(s,\omega)\, \omega \\[0.4em]
		& \quad \times \Omega_{\ell_j}^{3\varpi_2 + 2 - 2\theta + \kappa_2 - c_{\mathrm{in}} + \gamma}
		\Bigg] \Bigg).
	\end{aligned}
\end{equation}

Combining~\eqref{Lemma:Mutis:E8} and~\eqref{Lemma:Multiscale1:E12}, we obtain  
\begin{equation}\label{Lemma:Mutiscale1:E12a}
	\begin{aligned}
		\int_{\mathbb{R}^{+}} \mathrm{d}\omega\, \mathfrak{F}(\mathscr{T}_{\ell_j},\omega)\, \vartheta(\omega)
		\ge\ & \tfrac{1}{20} C_{\mathrm{ini}} \left(2\Omega_{\ell_j} \right)^{-c_{\mathrm{in}}}
		\exp\Big( \mathcal{C}_1\, \Omega_{\ell_j}^{\epsilon} \Big).
	\end{aligned}
\end{equation}

We next bound, for \( 0 \le t \le \mathscr{T}_{\ell_j} \),
\begin{equation}\label{Lemma:Multiscale1:E12c} 
	\begin{aligned}
		\int_{\mathbb{R}^{+}} \mathrm{d}\omega\, \mathfrak{F}(t,\omega)\, \vartheta(\omega)
		\le\ &   \mathfrak{M} + \mathfrak{E}.
	\end{aligned}
\end{equation}

Inequalities \eqref{Lemma:Mutiscale1:E12a} and \eqref{Lemma:Multiscale1:E12c} yield
\begin{equation}\label{Lemma:Multiscale1:E13}
	\begin{aligned}
		\mathfrak{M} + \mathfrak{E}  
		\ge\ & \tfrac{1}{20} C_{\mathrm{ini}} \left(2\Omega_{\ell_j} \right)^{-c_{\mathrm{in}}}
		\exp\left( \mathcal{C}_1\, \Omega_{\ell_j}^{\epsilon} \right)
		\longrightarrow \infty \quad \text{as } j \to \infty.
	\end{aligned}
\end{equation}

We obtain a contradiction, and hence the assumption leading to~\eqref{Lemma:Mutis:E1a}-\eqref{Lemma:Mutis:E1} cannot hold. This completes the proof of the proposition.

	\end{proof}
	
	\section{Estimating $\mathscr{N}_{\ell}^T$ and $\mathscr{M}_{\ell}^T$}\label{Sec:Second}
	
	\subsection{The general estimate}\label{Sec:DM}

For a constant \( \lambda > 0 \), we define 
\begin{equation}\label{TPump}
	\begin{aligned}
		\Theta_{\lambda,\Omega_{\ell}}^2 &:= \left\{ t \in \Theta \,\middle|\, 
		\forall i \in \{0, \dots, \mathscr{O}_{\Omega_{\ell}} - 1\},\ 
		\int_{\mathscr{D}_{i}^{\Omega_{\ell}}} \mathrm{d}\omega\, \mathfrak{F}(t)\, \omega
		< (1 - \lambda) \int_{[\Omega_{\ell}, \infty)} \mathrm{d}\omega\, \mathfrak{F}(t)\, \omega \right\}, \\[6pt]
		\Theta_{\lambda,\Omega_{\ell}}^1 &:= \Theta \setminus \Theta_{\lambda,\Omega_{\ell}}^2,
	\end{aligned}
\end{equation}
recalling that $\Theta$ is defined in \eqref{Theta}.

	\begin{proposition}
		\label{Propo:Collision}
We assume Assumptions X and Y.	There exists a universal constant \( \mathcal{C}_{22}^o > 0 \) such that the following estimate holds for \( 0 < \Delta_\ell < \frac{1}{10} \) and   
\(
\lambda = 1 - \frac{1}{2^{\sigma}} 
\):
	\begin{equation}\label{Propo:Collision:1}
		\begin{aligned}
			\mathfrak{M} + \mathfrak{E} \ \ge\ 
			&\, c_{12} \mathcal{C}_{22}^o 
			\int_{\Theta_{\lambda,\Omega_{\ell}}^2} \mathrm{d}t 
				\iint_{[\Omega_{\ell}, \infty)^2} \mathrm{d}\omega_1\, \mathrm{d}\omega_2\,
			\mathfrak{F}_1 \omega_1\, \mathfrak{F}_2 \omega_2\,
			(\omega_1+\omega_2)^{3\theta+\varpi_1+\alpha-2}\omega_1^{\varpi_1+1}\omega_2^{\varpi_1+1}   \\[0.5em]
			&+ c_{22} \mathcal{C}_{22}^o \lambda^4\Delta_\ell^2\Omega_\ell^{4\varpi_2- 2 +\alpha+\gamma}		
			\int_{\Theta_{\lambda,\Omega_{\ell}}^2} \mathrm{d}t 
			\left( \int_{[\Omega_{\ell}, \infty)} \mathrm{d}\omega\, \mathfrak{F} \omega \right)^3 \\[0.5em]
			&+ c_{31}\mathcal{C}_{22}^o	\int_{\Theta_{\lambda,\Omega_{\ell}}^2} \mathrm{d}t 
			\iiint_{[\Omega_{\ell}, \infty)^3} \mathrm{d}\omega_1\, \mathrm{d}\omega_2\, \mathrm{d}\omega_3\,
			(\omega_1+\omega_2+\omega_3)^{3\theta+\varpi_3+\alpha-2}
			\omega_1^{\varpi_3}\omega_2^{\varpi_3}\omega_3^{\varpi_3} \\[0.5em]
			&\quad \times \mathfrak{F}_1 \omega_1\, \mathfrak{F}_2 \omega_2\, \mathfrak{F}_3 \omega_3\, 
			(\omega_1 \omega_2 + \omega_2 \omega_3 + \omega_3 \omega_1).
	\end{aligned}
	\end{equation}
	
	\end{proposition}
	
\begin{proof}

We  divide the proof into smaller steps.

\textbf{Step 1: Applying \eqref{Lemma:Apriori:1}.} In this step, we will use \eqref{Lemma:Apriori:1} to obtain some preliminary bounds for $\int_{\Theta_{\lambda,\Omega_{\ell}}^2} \mathrm{d}t 
\left( \int_{[\Omega_{\ell}, \infty)} \mathrm{d}\omega\, \mathfrak{F} \omega \right)^2$ and $\int_{\Theta_{\lambda,\Omega_{\ell}}^2} \mathrm{d}t 
\left( \int_{[\Omega_{\ell}, \infty)} \mathrm{d}\omega\, \mathfrak{F} \omega \right)^3$.

From \eqref{Lemma:Apriori:1}, we obtain the bound:
\begin{equation}\label{Lemma:SecondEstimate:E1}
	\begin{aligned}
		\mathfrak{M} + \mathfrak{E} \ \ge\ 
		& \, c_{12} \int_{\Theta_{\lambda,\Omega_{\ell}}^2} \mathrm{d}t \iint_{[\Omega_{\ell}, \infty)^2} \mathrm{d}\omega_1\, \mathrm{d}\omega_2\,
		\tilde{\mathfrak{P}}_1\, \tilde{\mathfrak{P}}_2\, \mathfrak{F}_1 \omega_1\, \mathfrak{F}_2 \omega_2\, {\mathfrak{P}}(\omega_1 + \omega_2)\, 
		\frac{\alpha(1-\alpha)\omega_1 \omega_2}{(\omega_1 + \omega_2)^{2-\alpha}} \\[0.5em]
		& + c_{22} \int_{\Theta_{\lambda,\Omega_{\ell}}^2} \mathrm{d}t \iiint_{[\Omega_{\ell}, \infty)^3} \mathrm{d}\omega_1\, \mathrm{d}\omega_2\, \mathrm{d}\omega\,
		\mathfrak{F}_1 \omega_1\, \mathfrak{F}_2 \omega_2\, \mathfrak{F} \omega\, \frac{\alpha(1-\alpha)(\omega_{\mathrm{Med}} - \omega_{\mathrm{Inf}})^2}{(2\omega_{\mathrm{Med}} - \omega_{\mathrm{Inf}})^{2-\alpha}\, |k_{\mathrm{Sup}}|\, |k_{\mathrm{Med}}|}  \\[0.5em]
		& \quad \times \tilde{\mathfrak{R}}(\omega_{\mathrm{Sup}})\, \tilde{\mathfrak{R}}(\omega_{\mathrm{Inf}})\, \tilde{\mathfrak{R}}(\omega_{\mathrm{Med}})\,
		\mathfrak{R}(\omega_{\mathrm{Sup}} - \omega_{\mathrm{Inf}} + \omega_{\mathrm{Med}}) \bar{\mathfrak{R}}_o\\[0.5em]
		& + c_{31} \int_{\Theta_{\lambda,\Omega_{\ell}}^2} \mathrm{d}t \iiint_{[\Omega_{\ell}, \infty)^3} \mathrm{d}\omega_1\, \mathrm{d}\omega_2\, \mathrm{d}\omega_3\,
		\mathfrak{Q}(\omega_1 + \omega_2 + \omega_3)\, \tilde{\mathfrak{Q}}(\omega_1)\, \tilde{\mathfrak{Q}}(\omega_2)\, \tilde{\mathfrak{Q}}(\omega_3) \\[0.5em]
		& \quad \times \mathfrak{F}_1 \omega_1\, \mathfrak{F}_2 \omega_2\, \mathfrak{F}_3 \omega_3\,
	\alpha(1-\alpha)	\frac{\omega_1 \omega_2 + \omega_2 \omega_3 + \omega_3 \omega_1}{3(\omega_1 + \omega_2 + \omega_3)^{2-\alpha}}.
	\end{aligned}
\end{equation}
Similar to \eqref{Propo:Glo:E6c:1}-\eqref{Propo:Glo:E6c:2}, we obtain the following bounds for \( \omega_1, \omega_2 \ge 1 \):
\begin{equation}\label{Lemma:SecondEstimate:E1:1}
	\begin{aligned}
		\tilde{\mathfrak{P}}_1\, \tilde{\mathfrak{P}}_2\, \omega_1\, \omega_2\, {\mathfrak{P}}(\omega_1 + \omega_2)
		&\gtrsim (\omega_1+\omega_2)^{3\theta+\varpi_2}\omega_1^{\varpi_1+1}\omega_2^{\varpi_1+1}.
	\end{aligned}
\end{equation}

Furthermore, we have
\begin{equation}\label{Lemma:SecondEstimate:E1:2}
	\begin{aligned}
		&\bar{\mathfrak{R}}_o\,\bar{\mathfrak{R}}(\omega_{\mathrm{Sup}})\, 
		\bar{\mathfrak{R}}(\omega_{\mathrm{Inf}})\, 
		\bar{\mathfrak{R}}(\omega_{\mathrm{Med}})\,
		\mathfrak{R}(\omega_{\mathrm{Sup}} - \omega_{\mathrm{Inf}} + \omega_{\mathrm{Med}})\,
		\frac{1}{|k_{\mathrm{Sup}}||k_{\mathrm{Med}}|}  \\
		&\gtrsim (\omega_{\mathrm{Sup}} - \omega_{\mathrm{Inf}} + \omega_{\mathrm{Med}})^{2\theta+\varpi_2+\gamma}
		\omega_{\mathrm{Sup}}^{\varpi_2+1-\theta}
		\omega_{\mathrm{Inf}}^{\varpi_2+1}
		\omega_{\mathrm{Med}}^{\varpi_2+1-\theta}.
	\end{aligned}
\end{equation}

Finally,
\begin{equation}\label{Lemma:SecondEstimate:E1:3}
	\begin{aligned}
	&	\tilde{\mathfrak{Q}}_1\, \tilde{\mathfrak{Q}}_2\, \tilde{\mathfrak{Q}}_3\, \omega_1\, \omega_2\, \omega_3\, {\mathfrak{Q}}(\omega_1 + \omega_2 + \omega_3)
		&\gtrsim (\omega_1+\omega_2+\omega_3)^{3\theta+\varpi_3}
		\omega_1^{\varpi_3+1}\omega_2^{\varpi_3+1}\omega_3^{\varpi_3+1}.
	\end{aligned}
\end{equation}

Combining \eqref{Lemma:SecondEstimate:E1}-\eqref{Lemma:SecondEstimate:E1:3}, we arrive at the bound:
\begin{equation}\label{Lemma:SecondEstimate:E4}
	\begin{aligned}
		\mathfrak{M} + \mathfrak{E} \ \ge\ 
		& \, c_{12} \mathcal{C}_{22}^o 
		\int_{\Theta_{\lambda,\Omega_{\ell}}^2} \mathrm{d}t 
		\int_{\Theta_{\lambda,\Omega_{\ell}}^2} \mathrm{d}t \iint_{[\Omega_{\ell}, \infty)^2} \mathrm{d}\omega_1\, \mathrm{d}\omega_2\,
		  \mathfrak{F}_1 \omega_1\, \mathfrak{F}_2 \omega_2\, 
		(\omega_1+\omega_2)^{3\theta+\varpi_1+\alpha-2}\omega_1^{\varpi_1+1}\omega_2^{\varpi_1+1}   \\[0.5em]
		& + c_{22} \int_{\Theta_{\lambda,\Omega_{\ell}}^2} \mathrm{d}t \iiint_{[\Omega_{\ell}, \infty)^3} \mathrm{d}\omega_1\, \mathrm{d}\omega_2\, \mathrm{d}\omega\,
		\mathfrak{F}_1 \omega_1\, \mathfrak{F}_2 \omega_2\, \mathfrak{F} \omega\, \frac{(\omega_{\mathrm{Med}} - \omega_{\mathrm{Inf}})^2}{(2\omega_{\mathrm{Med}} - \omega_{\mathrm{Inf}})^{2-\alpha} }\\[0.5em]
		& \quad \times  (\omega_{\mathrm{Sup}} - \omega_{\mathrm{Inf}} + \omega_{\mathrm{Med}})^{2\theta+\varpi_2+\gamma}
		\omega_{\mathrm{Sup}}^{\varpi_2-\theta}\omega_{\mathrm{Inf}}^{\varpi_2}\omega_{\mathrm{Med}}^{\varpi_2-\theta} \\[0.5em]
		& + c_{31}\mathcal{C}_{22}^o\iiint_{[\Omega_{\ell}, \infty)^3} \mathrm{d}\omega_1\, \mathrm{d}\omega_2\, \mathrm{d}\omega_3\,
		(\omega_1+\omega_2+\omega_3)^{3\theta+\varpi_3+\alpha-2}
		\omega_1^{\varpi_3}\omega_2^{\varpi_3}\omega_3^{\varpi_3} \\[0.5em]
		& \quad \times \mathfrak{F}_1 \omega_1\, \mathfrak{F}_2 \omega_2\, \mathfrak{F}_3 \omega_3\, (\omega_1 \omega_2 + \omega_2 \omega_3 + \omega_3 \omega_1).
	\end{aligned}
\end{equation}

\textbf{ Step 2: Subdomain inclusions.} In this step, we will prove some subdomain inclusions using the setting of Section \ref{Sec:DDM}.
 
We  define
\begin{equation}\label{Sec:DM:2}
	\mathcal{D}^{\Omega_{\ell}}_{i,j,l} := \mathcal{D}^{\Omega_{\ell}}_{i} \times \mathcal{D}^{\Omega_{\ell}}_{j} \times \mathcal{D}^{\Omega_{\ell}}_{l}, 
	\quad i, j, l = 0, \dots, \mathscr{O}_{\Omega_{\ell}} - 1,
\end{equation}

\begin{equation}\label{Sec:DM:3}
	\begin{aligned}
		\mathfrak{S}_{i}^{\Omega_{\ell}} &= \{ i - 1, i, i + 1 \}, && \quad i = 1, \dots, \mathscr{O}_{\Omega_{\ell}} - 2, \\
		\mathfrak{S}_{0}^{\Omega_{\ell}} &= \{ 0, 1 \}, && 
		\mathfrak{S}_{\mathscr{O}_{\Omega_{\ell}} - 1}^{\Omega_{\ell}} = \{ \mathscr{O}_{\Omega_{\ell}} - 2, \mathscr{O}_{\Omega_{\ell}} - 1 \},
	\end{aligned}
\end{equation}

\begin{equation}\label{Sec:DM:4}
	\mathfrak{U}_{\Omega_{\ell}} := \left\{ (\omega, \omega_{1}, \omega_{2}) \in [\Omega_{\ell}, \infty)^3 : 
	|\omega_{\mathrm{Med}} - \omega_{\mathrm{Inf}}| \ge 2\Delta_{\ell} \right\},
\end{equation}

and
\begin{equation}\label{Sec:DM:5}
	\mathfrak{U}_{\Omega_{\ell}}' := \left\{ (\omega, \omega_{1}, \omega_{2}) \in [\Omega_{\ell}, \infty)^3 : 
	|\omega_{\mathrm{Med}} - \omega_{\mathrm{Inf}}| \ge \Delta_{\ell} \right\}.
\end{equation}
	
We also set
\begin{equation}\label{Propo:SecondEstimate:2}
	\mathfrak{U}''_{\Omega_{\ell}} =
	\left[
	\bigcup_{\substack{
			\mathscr{O}_{\Omega_{\ell}} \ge \max\{i, j, l\} -1 \\
			> \mathrm{mid}\{i, j, l\} \ge \min\{i, j, l\} 
	}}
	\mathcal{D}^{\Omega_{\ell}}_{i,j,l}
	\right],
\end{equation}
and aim to establish the containment
\begin{equation}\label{Propo:SecondEstimate:3}
		\mathfrak{U}_{\Omega_{\ell}} \subset 
	\mathfrak{U}''_{\Omega_{\ell}} \subset 
		\mathfrak{U}'_{\Omega_{\ell}}.
\end{equation}

To this end, consider a point 
\[
\left( \omega_2, \omega_{1}, \omega \right) \in 	\mathfrak{U}_{\Omega_{\ell}}.
\]
Without loss of generality, we may assume
\[
\omega = \omega_{\mathrm{Inf}} \left( \omega_2, \omega_1, \omega \right)
< \omega_1 = \omega_{\mathrm{Med}} \left( \omega_2, \omega_1, \omega \right)
\le \omega_2 = \omega_{\mathrm{Sup}} \left( \omega_2, \omega_1, \omega \right).
\]

By the definition of $	\mathfrak{U}_{\Omega_{\ell}}$, there exist indices $i$, $j$, and $l$ such that
\[
\omega_2 \in \mathcal D_i^{\Omega_{\ell}}, \quad
\omega_1 \in \mathcal D_j^{\Omega_{\ell}}, \quad
\omega \in \mathcal D_l^{\Omega_{\ell}}, \quad
\text{with } l \ge j \ge i.
\]
Furthermore, since
\[
|\omega_1-\omega| =  |\omega_{\mathrm{Med}} - \omega_{\mathrm{Inf}}| \ge 2 \Delta_{\ell},
\]
it follows that \( l > j + 1 \). Hence, the triple \( (i,j,l) \) satisfies the condition in the union defining $	\mathfrak{U}''_{\Omega_{\ell}}$, and therefore
\[
\left( \omega_2, \omega_1, \omega \right) \in 	\mathfrak{U}''_{\Omega_{\ell}}.
\]
This establishes that
\[
	\mathfrak{U}_{\Omega_{\ell}} \subset 
	\mathfrak{U}''_{\Omega_{\ell}}.
\]

Next, consider a point 
\[
\left( \omega_2, \omega_1, \omega \right) \in \mathfrak{U}''_{\Omega_{\ell}}.
\]
Suppose that 
\[
\omega_2 \in \mathcal D_i^{\Omega_{\ell}}, \quad 
\omega_1 \in \mathcal D_j^{\Omega_{\ell}}, \quad 
\omega \in \mathcal D_l^{\Omega_{\ell}},
\]
with
\[
\mathscr{O}_{\Omega_{\ell}} > l -1 > j \ge i.
\]
Then, by construction,
\[
\left| \omega_{\mathrm{Med}} - \omega_{\mathrm{Inf}} \right| \ge \Delta_{\ell}.
\]
Therefore,
\[
\mathfrak{U}''_{\Omega_{\ell}} \subset \mathfrak{U}'_{\Omega_{\ell}},
\]
which completes the proof of \eqref{Propo:SecondEstimate:3}.

\textbf{Step 3: Construction of the special subdomains.} In this step, we aim to construct special subdomains $\mathcal{O}^{y},  \mathcal{O}^{z,1},  \text{and}  \mathcal{O}^{z,2}$ satisfying estimates \eqref{Propo:SecondEstimate:34a1}-\eqref{Propo:SecondEstimate:34a}.

Define the index set
\begin{equation}\label{Propo:SecondEstimate:6a}
\mathscr{G} := \left\{ (i, j, l) \in \mathbb{Z}^3 \,\middle|\, i, j, l \geq 0,\quad \mathscr{O}_{\Omega_{\ell}} > l-1 > j \ge i  \right\}.
\end{equation}

Let \( t \in \Theta_{\lambda, \Omega_{\ell}}^2 \), and let \( i_t^o \in \{0, \dots, \mathscr{O}_{\Omega_{\ell}} - 1\} \) be an index such that
\begin{equation}\label{Propo:SecondEstimate:6}
	\int_{\mathcal D_{i_t^o}^{\Omega_{\ell}}} \mathrm{d}\omega\, \mathfrak F(t) \omega
	= \max_{i \in \{0, \dots, \mathscr{O}_{\Omega_{\ell}} - 1\}} 
	\left\{ \int_{\mathcal D_i^{\Omega_{\ell}}} \mathrm{d}\omega\, \mathfrak F(t)\omega \right\}.
\end{equation}

We distinguish between two cases:
\begin{equation*}
	\int_{\mathcal D_{i_t^o}^{\Omega_{\ell}}} \mathrm{d}\omega\, \mathfrak F(t) \omega
	< \frac{\lambda}{1000} \int_{[\Omega_{\ell},\infty)} \mathrm{d}\omega\, \mathfrak F(t)\omega \ \ \ \mbox{ and } 	\int_{\mathcal D_{i_t^o}^{\Omega_{\ell}}} \mathrm{d}\omega\, \mathfrak F(t) \omega
	\ge \frac{\lambda}{1000} \int_{[\Omega_{\ell},\infty)} \mathrm{d}\omega\, \mathfrak F(t)\omega.  
\end{equation*}

\medskip
\noindent
\textit{Case (A):} Assume that
\begin{equation}\label{Propo:SecondEstimate:18}
	\int_{\mathcal D_{i_t^o}^{\Omega_{\ell}}} \mathrm{d}\omega\, \mathfrak F(t) \omega
	< \frac{\lambda}{1000} \int_{[\Omega_{\ell},\infty)} \mathrm{d}\omega\, \mathfrak F(t)\omega.
\end{equation}

The analysis of this case is divided into two steps.

\textit{Case (A) -- Step (a): Construction of the first subdomain.}

Let \( \mathfrak{V}_t \) denote the collection of subsets of 
\(\{0, \dots, \mathscr{O}_{\Omega_{\ell}} - 1\}\) satisfying the following properties:

\begin{itemize}
	\item If \(\mathfrak{I}_t \in \mathfrak{V}_t\), then \(i_t^o \in \mathfrak{I}_t\).
	
	\item For \(\mathfrak{I}_t = \{i_1, \dots, i_n\} \in \mathfrak{V}_t\), each \(i_j \in \mathfrak{I}_t\) satisfies
	\begin{equation}\label{Propo:SecondEstimate:6a}
		\mathcal D_{i_j}^{\Omega_{\ell}} \cap 
		\bigcup_{i \in \mathfrak{I}_t \setminus \{i_j\}} \mathscr{D}_i^{\Omega_{\ell}} = \emptyset.
	\end{equation}
	
	\item There exists a constant \(10 > L_* \geq 3\), independent of \(t\), such that for every \(\mathfrak{I}_t \in \mathfrak{V}_t\),
	\begin{equation}\label{Propo:SecondEstimate:7}
		L_* \int_{\bigcup_{i \in \mathcal{I}_t} \mathcal D_i^{\Omega_{\ell}}} \mathrm{d}\omega\, \mathfrak F(t) \omega
		\geq 
		\int_{\bigcup_{i \in \mathcal{I}_t} \mathscr{D}_i^{\Omega_{\ell}}} \mathrm{d}\omega\, \mathfrak F(t)\omega.
	\end{equation}
	
	\item For all \(\mathfrak{I}_t \in \mathfrak{V}_t\),
	\begin{equation}\label{Propo:SecondEstimate:8}
		\int_{\bigcup_{i \in \mathcal{I}_t} \mathscr D_i^{\Omega_{\ell}}} \mathrm{d}\omega\, \mathfrak F(t) \omega
		< 
		(1 - \lambda) \int_{[\Omega_{\ell},\infty)} \mathrm{d}\omega\, \mathfrak F(t)\omega.
	\end{equation}
\end{itemize}

From  \eqref{Propo:SecondEstimate:6}, it follows that
\begin{equation}\label{Propo:SecondEstimate:9}
	3 \int_{\mathcal D_{i_t^o}^{\Omega_{\ell}}} \mathrm{d}\omega\, \mathfrak F(t) \omega
	\geq 
	\int_{\mathscr{D}_{i_t^o}^{\Omega_{\ell}}} \mathrm{d}\omega\, \mathfrak F(t)\omega,
\end{equation}
which implies that the collection \(\mathfrak{V}_t\) is nonempty, since the singleton set \(\{i_t^o\}\) satisfies all conditions.

Next, we select \(\Lambda^z_t = \{i_1, \dots, i_n\} \in \mathfrak{V}_t\) such that:

\begin{itemize}
	\item The first index is \(i_1 = i_t^o\).
	
	\item For each \(j = 2, \dots, n\),
	\begin{equation}\label{Propo:SecondEstimate:10:1}
		\int_{\mathcal D_{i_j}^{\Omega_{\ell}}} \mathrm{d}\omega\, \mathfrak F(t)\omega
		= 
		\max_{i \in \{0, \dots, \mathscr{O}_{\Omega_{\ell}} - 1\} \setminus 
			\left( \mathfrak {S}_{i_1}^{\Omega_{\ell}} \cup \cdots \cup \mathfrak {S}_{i_{j-1}}^{\Omega_{\ell}} \right)} 
		\left\{ \int_{\mathcal D_i^{\Omega_{\ell}}} \mathrm{d}\omega\, \mathfrak F(t)\omega \right\},
	\end{equation}
	where the sets \(\mathfrak {S}_i^{\Omega_{\ell}}\) are defined as in \eqref{Sec:DM:3}.
	
	\item Moreover, for every 
	\[
	j \in \{0, \dots, \mathscr{O}_{\Omega_{\ell}} - 1\} \setminus 
	\left( \mathfrak {S}_{i_1}^{\Omega_{\ell}} \cup \cdots \cup \mathfrak {S}_{i_n}^{\Omega_{\ell}} \right),
	\]
	we have
	\begin{equation}\label{Propo:SecondEstimate:10}
		\int_{\bigcup_{i \in \Lambda^z_t \cup \{j\}} \mathscr{D}_i^{\Omega_{\ell}}} \mathrm{d}\omega\, \mathfrak F(t)\omega
		\geq 
		(1 - \lambda) \int_{[\Omega_{\ell},\infty)} \mathrm{d}\omega\, \mathfrak F(t)\omega.
	\end{equation}
\end{itemize}

We now define the domains
\[
\mathcal{O}^y := [\Omega_{\ell}, \infty) \setminus \bigcup_{i \in \Lambda^z_t} \mathscr{D}_i^{\Omega_{\ell}} 
= \bigcup_{i \in \Lambda^y_t} \mathcal D_i^{\Omega_{\ell}}, 
\quad \text{and} \quad
\mathcal{O}^z := \bigcup_{i \in \Lambda^z_t} \mathcal D_i^{\Omega_{\ell}}.
\]

By inequality \eqref{Propo:SecondEstimate:8}, it follows that
\begin{equation}\label{Propo:SecondEstimate:11}
	\int_{\mathcal{O}^y} \mathrm{d}\omega\, \mathfrak F(t) \omega
	\geq 
	\lambda \int_{[\Omega_{\ell},\infty)} \mathrm{d}\omega\, \mathfrak F(t)\omega.
\end{equation}

The set \(\mathcal{O}^y\) constitutes our first subdomain. Next, we proceed to construct the second and third subdomains.

\textit{Case (A) -- Step (b): Construction of the second and third subdomains.}

Let \( i_t^{*} \in \Lambda^y_t \) be an index satisfying
\[
\mathcal D_{i_t^{*}}^{\Omega_{\ell}} \cap 
\bigcup_{i \in \Lambda^z_t} \mathscr{D}_i^{\Omega_{\ell}} = \emptyset,
\]
and such that
\begin{equation}\label{Propo:SecondEstimate:12}
	\int_{\mathcal D_{i_t^{*}}^{\Omega_{\ell}}} \mathrm{d}\omega\, \mathfrak F(t) \omega
	= \max_{i \in \Lambda^y_t} \left\{ \int_{\mathcal D_i^{\Omega_{\ell}}} \mathrm{d}\omega\, \mathfrak F(t)\omega \right\}.
\end{equation}

\smallskip

We now introduce two disjoint subsets \( \mathcal{Z}_1 \) and \( \mathcal{Z}_2 \) such that
\[
\mathcal{Z}_1 \cup \mathcal{Z}_2 = \mathscr{D}_{i_t^{*}}^{\Omega_{\ell}}, \quad
\mathcal{Z}_1 \cap \mathcal{Z}_2 = \emptyset, \quad
\mathcal{Z}_2 \subset \bigcup_{i \in \Lambda^y_t} \mathscr{D}_i^{\Omega_{\ell}}, \quad
\mathcal{Z}_2 \cap \bigcup_{i \in \Lambda^z_t} \mathscr{D}_i^{\Omega_{\ell}} = \emptyset.
\]

By the construction of the set \( \Lambda^z_t \), we have
\begin{equation}\label{Propo:SecondEstimate:13}
	L_* \int_{\bigcup_{i \in \Lambda^z_t} \mathcal D_i^{\Omega_{\ell}}} \mathrm{d}\omega\, \mathfrak F(t) \omega
	\ge 
	\int_{\bigcup_{i \in \Lambda^z_t} \mathscr{D}_i^{\Omega_{\ell}}} \mathrm{d}\omega\, \mathfrak F(t)\omega,
\end{equation}
and from the choice of \( i_t^{*} \) it follows that
\begin{equation}\label{Propo:SecondEstimate:14}
	L_* \int_{\mathcal D_{i_t^{*}}^{\Omega_{\ell}}} \mathrm{d}\omega\, \mathfrak F(t) \omega
	\ge 
	3 \int_{\mathcal D_{i_t^{*}}^{\Omega_{\ell}}} \mathrm{d}\omega\, \mathfrak F(t)\omega 
	\ge 
	\int_{\mathcal{Z}_2} \mathrm{d}\omega\, \mathfrak F(t)\omega.
\end{equation}

Using inequalities \eqref{Propo:SecondEstimate:13}, \eqref{Propo:SecondEstimate:14}, and \eqref{Propo:SecondEstimate:10}, we conclude that
\begin{equation}\label{Propo:SecondEstimate:15}
	\begin{aligned}
		L_* \int_{\bigcup_{i \in \Lambda^z_t \cup \{i_t^{*}\}} \mathcal D_i^{\Omega_{\ell}}} \mathrm{d}\omega\, \mathfrak F(t) \omega
		&\ge 
		\int_{\bigcup_{i \in \Lambda^z_t} \mathscr{D}_i^{\Omega_{\ell}}} \mathrm{d}\omega\, \mathfrak F(t) \omega
		+ \int_{\mathcal{Z}_2} \mathrm{d}\omega\, \mathfrak F(t)\omega \\
		&\ge 
		\int_{\bigcup_{i \in \Lambda^z_t \cup \{i_t^{*}\}} \mathscr{D}_i^{\Omega_{\ell}}} \mathrm{d}\omega\, \mathfrak F(t)\omega \\
		&\ge 
		(1 - \lambda) \int_{[\Omega_{\ell},\infty)} \mathrm{d}\omega\, \mathfrak F(t)\omega.
	\end{aligned}
\end{equation}

We deduce from \eqref{Propo:SecondEstimate:18} that

\begin{equation}\label{Propo:SecondEstimate:18a}
	\int_{\mathcal D_{i_t^{*}}^{\Omega_{\ell}}} \mathrm{d}\omega\, \mathfrak F(t) \omega
	< \frac{\lambda}{1000} \int_{[\Omega_{\ell},\infty)} \mathrm{d}\omega\, \mathfrak F(t)\omega.
\end{equation}

Inequality \eqref{Propo:SecondEstimate:18a}, combined with \eqref{Propo:SecondEstimate:15}, yields the estimate
\begin{equation}\label{Propo:SecondEstimate:19}
	L_* \int_{\bigcup_{i \in \Lambda^z_t} \mathcal D_i^{\Omega_{\ell}}} \mathrm{d}\omega\, \mathfrak F(t)  \omega
	+ \frac{\lambda L_*}{1000} \int_{[\Omega_{\ell},\infty)} \mathrm{d}\omega\, \mathfrak F(t)\omega  
	\ge (1 - \lambda) \int_{[\Omega_{\ell},\infty)} \mathrm{d}\omega\, \mathfrak F(t)\omega,
\end{equation}
which implies
\begin{equation}\label{Propo:SecondEstimate:20}
	L_* \int_{\bigcup_{i \in \Lambda^z_t} \mathcal D_i^{\Omega_{\ell}}} \mathrm{d}\omega\, \mathfrak F(t) \omega
	\ge \left(1 - \lambda - \frac{\lambda L_*}{1000} \right) \int_{[\Omega_{\ell},\infty)} \mathrm{d}\omega\, \mathfrak F(t)\omega.
\end{equation}

Therefore, we conclude that
\begin{equation}\label{Propo:SecondEstimate:21}
	\int_{\mathcal{O}^z} \mathrm{d}\omega\, \mathfrak F(t)\omega 
	\ge \frac{1 - 2\lambda}{L_*} \int_{[\Omega_{\ell},\infty)} \mathrm{d}\omega\, \mathfrak F(t)\omega
	\ge \frac{\lambda}{10 L_*} \int_{[\Omega_{\ell},\infty)} \mathrm{d}\omega\, \mathfrak F(t)\omega.
\end{equation}

As a result, we obtain
\begin{equation}\label{Propo:SecondEstimate:22}
	\int_{\mathcal{O}^z} \mathrm{d}\omega\, \mathfrak F(t) \omega
	\ge \frac{\lambda}{10 L_*} \int_{[\Omega_{\ell},\infty)} \mathrm{d}\omega\, \mathfrak F(t)\omega.
\end{equation}

We now proceed to further divide the domain \( \mathcal{O}^{z}\) into two subdomains. To do so, we repeat the construction from \eqref{Propo:SecondEstimate:10:1}--\eqref{Propo:SecondEstimate:10}, but replace the constant \( 1 - \lambda \) with \( \frac{\lambda}{20 L_*} \). Let
\[
 \Lambda^{z,1}_t = \{j_1, \dots, j_p\} \subset \Lambda^z_t
\]
denote the subset satisfying the following:

\begin{itemize}
	\item The first index is defined by \( j_1 = i_t^o \).
	
	\item For \( l = 2, \dots, p \),
	\begin{equation}\label{Propo:SecondEstimate:22:1}
		\int_{\mathcal D_{j_l}^{\Omega_{\ell}}} \mathrm{d}\omega\, \mathfrak F(t)\omega
		= 
		\max_{i \in \{0, \dots, \mathscr{O}_{\Omega_{\ell}} - 1\} 
			\setminus \left( \mathfrak {S}_{j_1}^{\Omega_{\ell}} \cup \cdots \cup \mathfrak {S}_{j_{l-1}}^{\Omega_{\ell}} \right)} 
		\left\{ \int_{\mathcal D_i^{\Omega_{\ell}}} \mathrm{d}\omega\, \mathfrak F(t)\omega \right\},
	\end{equation}
	where the notation \( \mathfrak {S}_{j}^{\Omega_{\ell}} \) is as introduced in \eqref{Sec:DM:3}.
	
	\item Moreover, for all 
	\[
	j \in \{0, \dots, \mathscr{O}_{\Omega_{\ell}} - 1\} 
	\setminus \left( \mathfrak {S}_{j_1}^{\Omega_{\ell}} \cup \cdots \cup \mathfrak {S}_{j_p}^{\Omega_{\ell}} \right),
	\]
	we have
	\begin{equation}\label{Propo:SecondEstimate:22:2}
		\int_{\bigcup_{i \in  \Lambda^{z,1}_t \cup \{j\}} \mathscr{D}_i^{\Omega_{\ell}}} \mathrm{d}\omega\, \mathfrak F(t)\omega
		\ge 
		\frac{\lambda}{20 L_*} \int_{[\Omega_{\ell},\infty)} \mathrm{d}\omega\, \mathfrak F(t)\omega,
	\end{equation}
	and additionally,
	\begin{equation}\label{Propo:SecondEstimate:22:3}
		\int_{\bigcup_{i \in  \Lambda^{z,1}_t} \mathcal D_i^{\Omega_{\ell}}} \mathrm{d}\omega\, \mathfrak F(t)\omega 
		< 
		\frac{\lambda}{20 L_*} \int_{[\Omega_{\ell},\infty)} \mathrm{d}\omega\, \mathfrak F(t)\omega.
	\end{equation}
\end{itemize}

It follows from \eqref{Propo:SecondEstimate:6} and \eqref{Propo:SecondEstimate:18} that this set is nonempty.

We now define
\[
 \Lambda^{z,2}_t :=  \Lambda^{z}_t \setminus  \Lambda^{z,1}_t, \quad 
\mathcal{O}^{z,1} := \bigcup_{i \in  \Lambda^{z,1}_t} \mathcal D_i^{\Omega_{\ell}}, \quad 
\mathcal{O}^{z,2}  := \bigcup_{i \in \Lambda^{z,2}_t} \mathcal D_i^{\Omega_{\ell}}.
\]

It follows from \eqref{Propo:SecondEstimate:22} and \eqref{Propo:SecondEstimate:22:3} that
\begin{equation}\label{Propo:SecondEstimate:22:4}
	\int_{\mathcal{O}^{z,2}} \mathrm{d}\omega\, \mathfrak F(t)\omega 
	\ge 
	\frac{\lambda}{20 L_*} \int_{[\Omega_{\ell}, \infty)} \mathrm{d}\omega\, \mathfrak F(t)\omega.
\end{equation}

Repeating the same argument as in 
\eqref{Propo:SecondEstimate:13}--\eqref{Propo:SecondEstimate:22}, with the constant \( 1 - \lambda \) replaced by \( \frac{\lambda}{20 L_*} \), we find
\begin{equation}\label{Propo:SecondEstimate:22:5}
	\int_{\mathcal{O}^{z,1}} \mathrm{d}\omega\, \mathfrak F(t) \omega
	\ge 
	\frac{1}{L_*} \left( \frac{\lambda}{20 L_*}  - \frac{\lambda L_*}{1000} \right) \int_{[\Omega_{\ell},\infty)} \mathrm{d}\omega\, \mathfrak F(t)\omega
	\ge 
	\frac{\lambda}{200L_*^2} \int_{[\Omega_{\ell},\infty)} \mathrm{d}\omega\, \mathfrak F(t)\omega.
\end{equation}

Finally, at the conclusion of this case, we have constructed three non-overlapping subdomains: \( \mathcal{O}^{y} \), \( \mathcal{O}^{z,1} \), and \( \mathcal{O}^{z,2} \), along with the three key inequalities \eqref{Propo:SecondEstimate:11}, \eqref{Propo:SecondEstimate:22:4}, and \eqref{Propo:SecondEstimate:22:5}.

Now consider \( i \in \Lambda^y_t \), \( j \in \Lambda^{z,1}_t \), and \( l \in \Lambda^{z,2}_t \), with all indices distinct. Then it is clear that
\begin{equation}\label{Propo:SecondEstimate:22:5a}
	\min\{i, j, l\} + 2 < 
	\min\big( \{i, j, l\} \setminus \{\min\{i, j, l\}\} \big) + 1 < 
	\max\{i, j, l\}.
\end{equation}

Using the three non-overlapping subdomains  
\( \mathcal{O}^{y} \), \( \mathcal{O}^{z,1} \), and \( \mathcal{O}^{z,2} \),  
together with the index condition~\eqref{Propo:SecondEstimate:22:5a}, we obtain the estimate
\[
\begin{aligned}
	& \sum_{(i,j,l) \in \mathscr{G}} 
	\iiint_{\mathcal D_{i,j,l}^{\Omega_{\ell}}} 
	\mathrm{d}\omega_1 \, \mathrm{d}\omega_2 \, \mathrm{d}\omega \,
	\mathfrak{F}_1 \omega_1 \, \mathfrak{F}_2 \omega_2 \, \mathfrak{F} \omega\,
	\mathbf{1}_{\omega + \omega_1 - \omega_2 \geq 0} \cdot 
	\frac{1}{\Delta_\ell^2} \cdot 
	\frac{1}{\Omega_\ell^{4\varpi_2+\gamma}} 
	\\[0.5em]
	& \quad \times 
	\frac{(\omega_{\mathrm{Med}} - \omega_{\mathrm{Inf}})^2}{(2\omega_{\mathrm{Med}} - \omega_{\mathrm{Inf}})^{2-\alpha}}	
	\cdot (\omega_{\mathrm{Sup}} - \omega_{\mathrm{Inf}} + \omega_{\mathrm{Med}})^{2\theta+\varpi_2+\gamma}
	\omega_{\mathrm{Sup}}^{\varpi_2-\theta}
	\omega_{\mathrm{Inf}}^{\varpi_2}
	\omega_{\mathrm{Med}}^{\varpi_2-\theta} \\[0.75em]
		&\gtrsim \sum_{(i,j,l) \in \mathscr{G}} 
	\iiint_{\mathcal D_{i,j,l}^{\Omega_{\ell}}} 
	\mathrm{d}\omega_1 \, \mathrm{d}\omega_2 \, \mathrm{d}\omega \,
	\mathfrak{F}_1 \omega_1 \, \mathfrak{F}_2 \omega_2 \, \mathfrak{F} \omega\,
	\mathbf{1}_{\omega + \omega_1 - \omega_2 \geq 0} \cdot 
	\frac{1}{\Delta_\ell^2} \cdot 
	\frac{1}{\Omega_\ell^{4\varpi_2+\gamma}} 
	\\[0.5em]
	& \quad \times 
	\frac{(\omega_{\mathrm{Med}} - \omega_{\mathrm{Inf}})^2}{(2\omega_{\mathrm{Med}} - \omega_{\mathrm{Inf}})^{2-\alpha}}	
	\cdot 
	\omega_{\mathrm{Sup}}^{2\varpi_2+\theta+\gamma}
	\omega_{\mathrm{Inf}}^{\varpi_2}
	\omega_{\mathrm{Med}}^{\varpi_2-\theta} \\[0.75em]
	&\gtrsim \sum_{(i,j,l) \in \mathscr{G}} 
	\iiint_{\mathcal D_{i,j,l}^{\Omega_{\ell}}} 
	\mathrm{d}\omega_1 \, \mathrm{d}\omega_2 \, \mathrm{d}\omega \,
	\mathfrak{F}_1 \omega_1 \, \mathfrak{F}_2 \omega_2 \, \mathfrak{F} \omega\,
	\mathbf{1}_{\omega + \omega_1 - \omega_2 \geq 0}    
	\cdot \frac{1}
	{  (2\omega_{\mathrm{Med}} - \omega_{\mathrm{Inf}})^{2-\alpha}} \\[0.75em]
	&\gtrsim \sum_{(i,j,l) \in \mathscr{G}} 
	\iiint_{\mathcal D_{i,j,l}^{\Omega_{\ell}}} 
	\mathrm{d}\omega_1 \, \mathrm{d}\omega_2 \, \mathrm{d}\omega \,
	\mathfrak{F}_1 \omega_1 \, \mathfrak{F}_2 \omega_2 \, \mathfrak{F} \omega\,
	\mathbf{1}_{\omega + \omega_1 - \omega_2 \geq 0}    
	\cdot  \frac{1}
	{  (2\omega_{\mathrm{Med}} - \omega_{\mathrm{Inf}})^{2-\alpha}} \\[0.75em]
	&\geq 
	C_0 	
	\int_{\mathcal{O}^{y}} \mathrm{d}\omega\, \mathfrak{F} \omega \cdot
	\int_{\mathcal{O}^{z,1}} \mathrm{d}\omega\, \mathfrak{F} \omega \cdot
	\int_{\mathcal{O}^{z,2}} \mathrm{d}\omega\, \mathfrak{F} \omega\cdot \frac{1}
	{  (2\omega_{\mathrm{Med}} - \omega_{\mathrm{Inf}})^{2-\alpha}} ,
\end{aligned}
\]
since $2\varpi_2+\theta+\gamma\ge 0$ (see \eqref{X4}), for some constant \( C_0 > 0 \) independent of the parameters and may vary from line to line. In the above estimates, we use the fact that when \( {2\theta+\varpi_2+\gamma} \ge 0 \),
\[
(\omega_{\mathrm{Sup}} - \omega_{\mathrm{Inf}} + \omega_{\mathrm{Med}})^{2\theta+\varpi_2+\gamma} 
\ge 
\omega_{\mathrm{Sup}}^{2\theta+\varpi_2+\gamma},
\]
and when \( {2\theta+\varpi_2+\gamma} < 0 \),
\[
(\omega_{\mathrm{Sup}} - \omega_{\mathrm{Inf}} + \omega_{\mathrm{Med}})^{2\theta+\varpi_2+\gamma} 
\ge 
(2\,\omega_{\mathrm{Sup}})^{2\theta+\varpi_2+\gamma}.
\]

Noting that \( \Omega_\ell \le \omega_{\mathrm{Med}} \le 2\Omega_\ell \), we estimate
\[
\begin{aligned}
	& \sum_{(i,j,l) \in \mathscr{G}} 
	\iiint_{\mathcal D_{i,j,l}^{\Omega_{\ell}}} 
	\mathrm{d}\omega_1 \, \mathrm{d}\omega_2 \, \mathrm{d}\omega \,
	\mathfrak{F}_1 \omega_1 \, \mathfrak{F}_2 \omega_2 \, \mathfrak{F} \omega\,
	\mathbf{1}_{\omega + \omega_1 - \omega_2 \geq 0} \cdot 
	\frac{1}{\Delta_\ell^2} \cdot 
	\frac{1}{\Omega_\ell^{4\varpi_2-2+\alpha+\gamma}} 
	\\[0.5em]
	& \quad \times 
	\frac{(\omega_{\mathrm{Med}} - \omega_{\mathrm{Inf}})^2}{(2\omega_{\mathrm{Med}} - \omega_{\mathrm{Inf}})^{2-\alpha}}	
	\cdot (\omega_{\mathrm{Sup}} - \omega_{\mathrm{Inf}} + \omega_{\mathrm{Med}})^{2\theta+\varpi_2+\gamma}
	\omega_{\mathrm{Sup}}^{\varpi_2-\theta}
	\omega_{\mathrm{Inf}}^{\varpi_2}
	\omega_{\mathrm{Med}}^{\varpi_2-\theta} \\[0.75em]
	&\geq 
	C_0 	
	\int_{\mathcal{O}^{y}} \mathrm{d}\omega\, \mathfrak{F} \omega \cdot
	\int_{\mathcal{O}^{z,1}} \mathrm{d}\omega\, \mathfrak{F} \omega \cdot
	\int_{\mathcal{O}^{z,2}} \mathrm{d}\omega\, \mathfrak{F} \omega,
\end{aligned}
\]
for some constant \( C_0 > 0 \) independent of the parameters, which may vary from line to line.

Combining this with the key inequalities  
\eqref{Propo:SecondEstimate:11}, \eqref{Propo:SecondEstimate:22:4}, and \eqref{Propo:SecondEstimate:22:5}, we obtain the estimate
\begin{equation}\label{Propo:SecondEstimate:34}
	\begin{aligned}
		& \int_{[\Omega_{\ell},\infty)} \mathrm{d}\omega\, \mathfrak{F} \omega 
		\int_{[\Omega_{\ell},\infty)} \mathrm{d}\omega\, \mathfrak{F}  \omega 
		\int_{[\Omega_{\ell},\infty)} \mathrm{d}\omega\, \mathfrak{F} \omega   \\[0.5em]
		& \leq    
		\sum_{(i,j,l) \in \mathscr{G}} 
		\iiint_{\mathcal D_{i,j,l}^{\Omega_{\ell}}} 
		\mathrm{d}\omega_1\, \mathrm{d}\omega_2\, \mathrm{d}\omega\, 
		\mathfrak{F}_1 \omega_1\, \mathfrak{F}_2 \omega_2\, \mathfrak{F} \omega\,
		\mathbf{1}_{\omega + \omega_1 - \omega_2 \geq 0}\cdot 
		\frac{1}{\Omega_\ell^{4\varpi_2-2+\alpha+\gamma}}  \\[0.5em]
		& \quad \times 
		\frac{C_0}{\lambda^4\Delta_\ell^2} 
		\cdot  
		\frac{(\omega_{\mathrm{Med}} - \omega_{\mathrm{Inf}})^2}{(2\omega_{\mathrm{Med}} - \omega_{\mathrm{Inf}})^{2-\alpha}}	
		\cdot (\omega_{\mathrm{Sup}} - \omega_{\mathrm{Inf}} + \omega_{\mathrm{Med}})^{2\theta+\varpi_2+\gamma}
		\omega_{\mathrm{Sup}}^{\varpi_2-\theta}
		\omega_{\mathrm{Inf}}^{\varpi_2}
		\omega_{\mathrm{Med}}^{\varpi_2-\theta},
	\end{aligned}
\end{equation}
for some constant \( C_0 > 0 \) independent of the parameters, which may vary from line to line.
Noting that \( \Omega_\ell \le \omega_{\mathrm{Med}}, \omega_{\mathrm{Inf}} \le 2\Omega_\ell \), we estimate
\[
\begin{aligned}
	& \sum_{(i,j,l) \in \mathscr{G}} 
	\iiint_{\mathcal D_{i,j,l}^{\Omega_{\ell}}} 
	\mathrm{d}\omega_1 \, \mathrm{d}\omega_2 \, \mathrm{d}\omega \,
	\mathfrak{F}_1 \omega_1 \, \mathfrak{F}_2 \omega_2 \, \mathfrak{F} \omega\,
	\mathbf{1}_{\omega + \omega_1 - \omega_2 \geq 0} \cdot 
	\frac{1}{\Delta_\ell^2} \cdot 
	\frac{1}{\Omega_\ell^{4\varpi_2-2+\theta+\alpha+\gamma}} 
	\\[0.5em]
	& \quad \times 
	\frac{(\omega_{\mathrm{Med}} - \omega_{\mathrm{Inf}})^2}{(2\omega_{\mathrm{Med}} - \omega_{\mathrm{Inf}})^{2-\alpha}}	
	\cdot (\omega_{\mathrm{Sup}} - \omega_{\mathrm{Inf}} + \omega_{\mathrm{Med}})^{2\theta+\varpi_2+\gamma}
	\omega_{\mathrm{Sup}}^{\varpi_2-\theta}
	\omega_{\mathrm{Inf}}^{\varpi_2}
	\omega_{\mathrm{Med}}^{\varpi_2-\theta} \\[0.75em]
	&\geq 
	C_0 	
	\int_{\mathcal{O}^{y}} \mathrm{d}\omega\, \mathfrak{F} \omega \cdot
	\int_{\mathcal{O}^{z,1}} \mathrm{d}\omega\, \mathfrak{F} \omega \cdot
	\int_{\mathcal{O}^{z,2}} \mathrm{d}\omega\, \mathfrak{F} \omega,
\end{aligned}
\]
for some constant \( C_0 > 0 \) independent of the parameters, which may vary from line to line.

Combining this with the key inequalities  
\eqref{Propo:SecondEstimate:11}, \eqref{Propo:SecondEstimate:22:4}, and \eqref{Propo:SecondEstimate:22:5}, we obtain the estimate
\begin{equation}\label{Propo:SecondEstimate:34}
	\begin{aligned}
		& \int_{[\Omega_{\ell},\infty)} \mathrm{d}\omega\, \mathfrak{F} \omega 
		\int_{[\Omega_{\ell},\infty)} \mathrm{d}\omega\, \mathfrak{F}  \omega 
		\int_{[\Omega_{\ell},\infty)} \mathrm{d}\omega\, \mathfrak{F} \omega   \\[0.5em]
		& \leq    
		\sum_{(i,j,l) \in \mathscr{G}} 
		\iiint_{\mathcal D_{i,j,l}^{\Omega_{\ell}}} 
		\mathrm{d}\omega_1\, \mathrm{d}\omega_2\, \mathrm{d}\omega\, 
		\mathfrak{F}_1 \omega_1\, \mathfrak{F}_2 \omega_2\, \mathfrak{F} \omega\,
		\mathbf{1}_{\omega + \omega_1 - \omega_2 \geq 0}  \\[0.5em]
		& \quad \times 
		\frac{C_0}{\lambda^4\Delta_\ell^2\Omega_\ell^{4\varpi_2-2+\theta+\gamma}} 
		\cdot  
		\frac{(\omega_{\mathrm{Med}} - \omega_{\mathrm{Inf}})^2}{(2\omega_{\mathrm{Med}} - \omega_{\mathrm{Inf}})^{2-\alpha}}	
		\cdot (\omega_{\mathrm{Sup}} - \omega_{\mathrm{Inf}} + \omega_{\mathrm{Med}})^{2\theta+\varpi_2+\gamma}
		\omega_{\mathrm{Sup}}^{\varpi_2-\theta}
		\omega_{\mathrm{Inf}}^{\varpi_2}
		\omega_{\mathrm{Med}}^{\varpi_2-\theta},
	\end{aligned}
\end{equation}
for some constant \( C_0 > 0 \) independent of the parameters, which may vary from line to line.

\textit{Case (B):} In this case, we assume that
\begin{equation}\label{Propo:SecondEstimate:22:6}
	\int_{\mathcal D_{i_t^o}^{\Omega_{\ell}}} \mathrm{d}\omega\, \mathfrak F(t) \omega
	\geq \frac{\lambda}{1000} \int_{[\Omega_{\ell},\infty)} \mathrm{d}\omega\, \mathfrak F(t)\omega.
\end{equation}

We set 
\[
\mathcal{O}^{y}  := \mathcal D_{i_t^o}^{\Omega_{\ell}}, \quad \text{and} \quad \Lambda^y := \{ i_t^o \},
\]
so that
\begin{equation}\label{Propo:SecondEstimate:22:7}
	\int_{\mathcal{O}^{y} } \mathrm{d}\omega\, \mathfrak F(t) \omega
	\geq \frac{\lambda}{1000} \int_{[\Omega_{\ell},\infty)} \mathrm{d}\omega\, \mathfrak F(t)\omega.
\end{equation}

By our assumption,
\[
\int_{\mathscr{D}_{i_t^o}^{\Omega_{\ell}}} \mathrm{d}\omega\, \mathfrak F(t) \omega
< (1 - \lambda) \int_{[\Omega_{\ell},\infty)} \mathrm{d}\omega\, \mathfrak F(t)\omega,
\]
which implies
\begin{equation}\label{Propo:SecondEstimate:22:8}
	\lambda \int_{[\Omega_{\ell},\infty)} \mathrm{d}\omega\, \mathfrak F(t) \omega
	\leq \int_{[\Omega_{\ell},\infty) \setminus \mathscr{D}_{i_t^o}^{\Omega_{\ell}}} \mathrm{d}\omega\, \mathfrak F(t)\omega.
\end{equation}

We define
\[
[\Omega_{\ell},\infty) \setminus \mathscr{D}_{i_t^o}^{\Omega_{\ell}} 
= \bigcup_{i \in \Lambda^z_t} \mathcal D_i^{\Omega_{\ell}},
\]
and partition the index set as
\[
\Lambda^{z'}_t := \{ i \in \Lambda^z_t \mid i > i_t^o \}, \quad
\Lambda^{z''}_t := \{ i \in \Lambda^z_t \mid i < i_t^o \}.
\]

Since, by \eqref{Propo:SecondEstimate:22:8},
\[
\lambda \int_{[\Omega_{\ell},\infty)} \! \mathrm{d}\omega\, \mathfrak F(t) \omega
\leq \int_{\bigcup_{i \in \Lambda^{z'}_t}\mathcal{D}_i^{\Omega_{\ell}}} \!  \mathrm{d}\omega\, \mathfrak F(t)\omega
+ \int_{\bigcup_{i \in \Lambda^{z''}_t}\mathcal{D}_i^{\Omega_{\ell}}} \!  \mathrm{d}\omega\, \mathfrak F(t)\omega,
\]
it follows that either
\begin{equation}\label{Propo:SecondEstimate:22:9}
	\frac{\lambda}{2} \int_{[\Omega_{\ell},\infty)} \mathrm{d}\omega\, \mathfrak F(t) \omega
	\leq \int_{\bigcup_{i \in \Lambda^{z'}_t}\mathcal{D}_i^{\Omega_{\ell}}} \mathrm{d}\omega\, \mathfrak F(t)\omega,
\end{equation}
or
\begin{equation}\label{Propo:SecondEstimate:22:10}
	\frac{\lambda}{2} \int_{[\Omega_{\ell},\infty)} \mathrm{d}\omega\, \mathfrak F(t) \omega
	\leq \int_{\bigcup_{i \in\Lambda^{z''}_t}\mathcal{D}_i^{\Omega_{\ell}}}\mathrm{d}\omega\, \mathfrak F(t)\omega.
\end{equation}

If \eqref{Propo:SecondEstimate:22:9} holds, we set
\[
\Lambda^{z,1}_t := \Lambda^{z'}_t = \{ i \in \mathscr{Z}_t \mid i > i_t^o \}, \quad
\Lambda^{z,2}_t := \Lambda^{z,1}_t,
\]
and define
\[
\bigcup_{i \in \Lambda^{z,1}_t} \mathcal{D}_i^{\Omega_{\ell}}
= \bigcup_{i \in \Lambda^{z,2}_t} \mathcal{D}_i^{\Omega_{\ell}}
=: \mathcal{O}^{z,1} = \mathcal{O}^{z,2}.
\]

If instead \eqref{Propo:SecondEstimate:22:10} holds, we set
\[
\Lambda^{z,1}_t := \Lambda^{z''}_t = \{ i \in \Lambda^{z}_t \mid i < i_t^o \}, \quad
\Lambda^{z,2}_t := \Lambda^{z,1}_t,
\]
and define
\[
\bigcup_{i \in \Lambda^{z,1}_t} \Omega_i^{\Omega_{\ell}} =: \mathcal{O}^{z,1}, \quad
\bigcup_{i \in \Lambda^{z,2}_t} \Omega_i^{\Omega_{\ell}} =: \mathcal{O}^{z,2}.
\]

In both cases, we have
\begin{equation}\label{Propo:SecondEstimate:22:11}
	C_0 \lambda \int_{[\Omega_{\ell},\infty)} \mathrm{d}\omega\, \mathfrak F(t) \omega
	\leq \int_{\mathcal{O}^{z,1}} \mathrm{d}\omega\, \mathfrak F(t)\omega, \quad
	C_0 \lambda \int_{[\Omega_{\ell},\infty)} \mathrm{d}\omega\, \mathfrak F(t) \omega
	\leq \int_{\mathcal{O}^{z,2}} \mathrm{d}\omega\, \mathfrak F(t)\omega,
\end{equation}
where \( C_0 > 0 \) is a universal constant that may vary between estimates.

Finally, at the conclusion of this case, we identify three subdomains:
\[
\mathcal{O}^{y}, \quad \mathcal{O}^{z,1}, \quad \text{and} \quad \mathcal{O}^{z,2},
\]
along with the key inequalities \eqref{Propo:SecondEstimate:22:7} and \eqref{Propo:SecondEstimate:22:11}.

Consider indices \( i \in \Lambda^y \), \( j \in \Lambda^{z,1} \), and \( l \in \Lambda^{z,2} \). Then, it follows that
\begin{equation}\label{Propo:SecondEstimate:22:12}
	\max\{i, j, l\} - 1 > \max\big( \{i, j, l\} \setminus \{\max\{i, j, l\}\} \big).
\end{equation}

We now utilize the subdomains 
\( \Lambda^y \), \( \Lambda^{z,1} \), and \( \Lambda^{z,2} \), 
along with the index inequality \eqref{Propo:SecondEstimate:22:12}, to derive the following bound, using \eqref{Propo:SecondEstimate:6a}:
\begin{equation}\label{Propo:SecondEstimate:34a1}
	\begin{aligned}
		& \sum_{(i,j,l) \in \mathscr{G}} 
		\iiint_{\mathcal D_{i,j,l}^{\Omega_{\ell}}} 
		\mathrm{d}\omega_1\, \mathrm{d}\omega_2\, \mathrm{d}\omega\, 
		\mathfrak{F}_1 \omega_1\, \mathfrak{F}_2 \omega_2\, \mathfrak{F} \omega\,
		\mathbf{1}_{\omega + \omega_1 - \omega_2 \geq 0} \\[0.5em]
		& \quad \times \frac{C_0}{\lambda \Delta_\ell^2\Omega_\ell^{4\varpi_2-2+\alpha+\gamma}} 
		\cdot  
		\frac{(\omega_{\mathrm{Med}} - \omega_{\mathrm{Inf}})^2}{(2\omega_{\mathrm{Med}} - \omega_{\mathrm{Inf}})^{2-\alpha}}	
		\cdot (\omega_{\mathrm{Sup}} - \omega_{\mathrm{Inf}} + \omega_{\mathrm{Med}})^{2\theta+\varpi_2+\gamma}
		\omega_{\mathrm{Sup}}^{\varpi_2-\theta}
		\omega_{\mathrm{Inf}}^{\varpi_2}
		\omega_{\mathrm{Med}}^{\varpi_2-\theta} \\[0.75em]
			& \ge\ 
		\int_{\mathcal{O}^{y}} \mathrm{d}\omega\, \mathfrak{F} \omega \cdot
		\int_{\mathcal{O}^{z,1}} \mathrm{d}\omega\, \mathfrak{F} \omega \cdot
		\int_{\mathcal{O}^{z,2}} \mathrm{d}\omega\, \mathfrak{F} \omega.
	\end{aligned}
\end{equation}

Combining this with the key inequalities \eqref{Propo:SecondEstimate:22:7} and \eqref{Propo:SecondEstimate:22:11}, we obtain
\begin{equation}\label{Propo:SecondEstimate:34a}
	\begin{aligned}
		\left( \int_{[\Omega_{\ell},\infty)} \mathrm{d}\omega\, \mathfrak{F}\omega \right)^3 
		&\le    
		\sum_{(i,j,l) \in \mathscr{G}} 
		\iiint_{\mathcal D_{i,j,l}^{\Omega_{\ell}}} 
		\mathrm{d}\omega_1\, \mathrm{d}\omega_2\, \mathrm{d}\omega\, 
		\mathfrak{F}_1 \omega_1\, \mathfrak{F}_2 \omega_2\, \mathfrak{F} \omega\,
		\mathbf{1}_{\omega + \omega_1 - \omega_2 \geq 0} \\[0.5em]
		&\quad \times 
		\frac{C_0}{\lambda^4\Delta_\ell^2\Omega_\ell^{4\varpi_2-2+\alpha+\gamma}} \\[0.5em]
		&\quad \times
		\frac{(\omega_{\mathrm{Med}} - \omega_{\mathrm{Inf}})^2}{(2\omega_{\mathrm{Med}} - \omega_{\mathrm{Inf}})^{2-\alpha}}	
		\cdot (\omega_{\mathrm{Sup}} - \omega_{\mathrm{Inf}} + \omega_{\mathrm{Med}})^{2\theta+\varpi_2+\gamma}
		\omega_{\mathrm{Sup}}^{\varpi_2-\theta}
		\omega_{\mathrm{Inf}}^{\varpi_2}
		\omega_{\mathrm{Med}}^{\varpi_2-\theta},
	\end{aligned}
\end{equation}
where \( C_0 > 0 \) is a universal constant that may vary between estimates.

\textbf{Step 3: The final estimate}

Combining \eqref{Propo:SecondEstimate:34}, \eqref{Propo:SecondEstimate:34a}, and \eqref{Lemma:SecondEstimate:E4}, we obtain the following estimate:
\begin{equation}\label{Propo:SecondEstimate:35}
	\begin{aligned}
		\mathfrak{M} + \mathfrak{E} \ \ge\ 
		&\, c_{12} \mathcal{C}_{22}^o 
		\int_{\Theta_{\lambda,\Omega_{\ell}}^2} \mathrm{d}t 
		\iint_{[\Omega_{\ell}, \infty)^2} \mathrm{d}\omega_1\, \mathrm{d}\omega_2\,
		\mathfrak{F}_1 \omega_1\, \mathfrak{F}_2 \omega_2\,
		(\omega_1+\omega_2)^{3\theta+\varpi_1+\alpha-2}\omega_1^{\varpi_1+1}\omega_2^{\varpi_1+1}   \\[0.5em]
		&+ c_{22} \mathcal{C}_{22}^o \lambda^4\Delta_\ell^2\Omega_\ell^{4\varpi_2-2+\alpha+\gamma}		
		\int_{\Theta_{\lambda,\Omega_{\ell}}^2} \mathrm{d}t 
		\left( \int_{[\Omega_{\ell}, \infty)} \mathrm{d}\omega\, \mathfrak{F} \omega \right)^3 \\[0.5em]
		&+ c_{31}\mathcal{C}_{22}^o
		\int_{\Theta_{\lambda,\Omega_{\ell}}^2} \mathrm{d}t 
		\iiint_{[\Omega_{\ell}, \infty)^3} \mathrm{d}\omega_1\, \mathrm{d}\omega_2\, \mathrm{d}\omega_3\,
		(\omega_1+\omega_2+\omega_3)^{3\theta+\varpi_3+\alpha-2}
		\omega_1^{\varpi_3}\omega_2^{\varpi_3}\omega_3^{\varpi_3} \\[0.5em]
		&\quad \times \mathfrak{F}_1 \omega_1\, \mathfrak{F}_2 \omega_2\, \mathfrak{F}_3 \omega_3\, 
		(\omega_1 \omega_2 + \omega_2 \omega_3 + \omega_3 \omega_1),
	\end{aligned}
\end{equation}
for some universal constant \( \mathcal{C}_{22}^o > 0 \).

This completes the proof of the proposition.

\end{proof}

	\subsection{Set estimates}
	
	\begin{proposition} 
		\label{Lemma:Mutiscale2} We assume Assumptions X and Y. Let \( T^* \) be as in~\eqref{T0}.  We choose $0\le T<T^*$ if $T^*>0$ and $ T=0$ if $T^*=0$. By the definitions in 
	~Assumption X,~\eqref{sigma}, and~\eqref{epsilon}, the following statements hold for the sets defined in \eqref{Sec:DD:4} :
		
		\begin{itemize}
			\item[(i)] There exists a universal constant \( C_\mathscr{N} > 0 \), independent of the parameters 
			\[
			\epsilon,\,  \sigma,\, \theta,\, \varpi_1,\, \varpi_2,\, \varpi_3,\,  T,\, T^*,
			\]
			such that the set $\mathscr{N}_{\ell}^T$ defined in \eqref{Sec:DD:4} satisfies the estimate
			\begin{equation}\label{Lemma:Mutiscale2:1}
				\left| \mathscr{N}_{\ell}^T \right| 
				\le C_\mathscr{N} (\mathfrak{M} + \mathfrak{E})\, \Omega_\ell^{-c_\mathscr{N}},
			\end{equation}
			where the exponent \( c_\mathscr{N} \) is given as follows:
			
			\begin{itemize}
				\item If \( c_{12}>0 \), \( c_{22}=c_{31}=0 \), then
				\begin{equation}\label{Lemma:Mutiscale2:2}
					c_\mathscr{N} := -(2\sigma  - 3\varpi_{1}-3\theta-\alpha).
				\end{equation}
				
				\item If \( c_{31}>0 \), \( c_{12}=c_{22}=0 \), then
				\begin{equation}\label{Lemma:Mutiscale2:3}
					c_\mathscr{N} := -(3\sigma - 4\varpi_{3}-3\theta-\alpha).
				\end{equation}
				
				\item If \( c_{22}>0 \), \( c_{12}=c_{31}=0 \), then
				\begin{equation}\label{Lemma:Mutiscale2:4}
					c_\mathscr{N} := -\frac{4\varpi_2+ \alpha+\gamma}{2}+3\sigma.
				\end{equation}
				
				\item If \( c_{12}>0 \), \( c_{31}>0 \), \( c_{22}=0 \), then
				\begin{equation}\label{Lemma:Mutiscale2:5}
					c_\mathscr{N} := \min\left\{ 
					-(2\sigma  - 3\varpi_{1}-3\theta-\alpha),\, 
					-\frac{4\varpi_2+ \alpha+\gamma}{2} +3\sigma
					\right\}.
				\end{equation}
				
				\item If \( c_{12}>0 \), \( c_{22}>0 \), \( c_{31}=0 \), then
				\begin{equation}\label{Lemma:Mutiscale2:6}
					c_\mathscr{N} := \min\left\{ 
					-(2\sigma  - 3\varpi_{1}-3\theta-\alpha),\, 
					-(3\sigma - 4\varpi_{3}-3\theta-\alpha) 
					\right\}.
				\end{equation}
				
				\item If \( c_{22}>0 \), \( c_{31}>0 \), \( c_{12}=0 \), then
				\begin{equation}\label{Lemma:Mutiscale2:7}
					c_\mathscr{N} := \min\left\{ 
					-(3\sigma - 4\varpi_{3}-3\theta-\alpha),\, 
					-\frac{4\varpi_2+ \alpha+\gamma}{2} +3\sigma
					\right\}.
				\end{equation}
				
				\item If \( c_{12}, c_{22}, c_{31} > 0 \), then
				\begin{equation}\label{Lemma:Mutiscale2:8}
					c_\mathscr{N} := \min\left\{ 
					-(2\sigma   - 3\varpi_{1}-3\theta-\alpha),\ 
					-(3\sigma - 4\varpi_{3}-3\theta-\alpha),\ 
					-\frac{4\varpi_2+ \alpha+\gamma}{2} +3\sigma
					\right\}.
				\end{equation}
			\end{itemize}
			
			
			\item[(ii)] Define
			\[
			c_{\mathscr{N}, \mathscr{P}} := \min\left\{ c_\mathscr{N},\, c_\mathscr{P} \right\},
			\]
			where $c_\mathscr{P} $ is defined in Proposition \ref{Lemma:Mutis1}.
			Then, there exists a constant \( \mathfrak{C}_{\mathscr{M},\mathscr{Q}} > 0 \) such that for all \( \ell > \mathfrak{C}_{\mathscr{M},\mathscr{Q}} \), with \( \ell \in \mathbb{N} \), and for all \( T \in [0, T^*) \), the following bound holds:
			\begin{equation} \label{Lemma:Mutiscale2:11}
				\left| 
				\bigcup_{i=\ell}^\infty \left( \mathscr{M}_i^T \setminus \mathscr{Q}_i^T \right) 
				\right| 
				\le C_{\mathscr{M},\mathscr{Q}}\, \Omega_\ell^{-c_{\mathscr{N}, \mathscr{P}}},
			\end{equation}
			where \(\mathfrak{C}_{\mathscr{M},\mathscr{Q}} , C_{\mathscr{M},\mathscr{Q}} > 0 \) are  universal constants, independent of the parameters 
			\( \epsilon, \theta,\) \(  \sigma,\) \(  \varpi_{1}, \varpi_{2}, \varpi_{3}, T, T^* \).
			
			\item[(iii)] Finally, for all \( \sigma > 0 \) satisfying \eqref{sigma}, there exits \( \mathfrak{C}_\mathscr{A}(\sigma) > 0 \) such that for all \( \ell > \mathfrak{C}_\mathscr{A}(\sigma) \), we have
			\begin{equation} \label{Lemma:Mutiscale3:1}
				\left| \mathscr{M}_\ell^T \right| 
				\le C_{\mathscr{M}} \left( \Omega_\ell \right)^{-c_{\mathscr{N}, \mathscr{P}}},
			\end{equation}
			where \( C_{\mathscr{M}} > 0 \) is a universal constant, independent of the parameters 
			\( \epsilon, \theta,\) \(  \sigma,\)  \(  \varpi_{1}, \varpi_{2}, \varpi_{3}, T, T^* \).
		\end{itemize}

	\end{proposition}

	\begin{proof}
		We divide the proof into several parts.
		
		{\bf Part 1: Proof of~\eqref{Lemma:Mutiscale2:1}.  }

To this end, we apply Proposition~\ref{Propo:Collision}, 
where \( \sigma \) is defined in~\eqref{sigma}. This yields the estimate
\begin{equation} \label{Propo:Mutiscale2:E1}
	\begin{aligned}
		\mathfrak{M} + \mathfrak{E} \ \ge\ 
		&\, c_{12} \mathcal{C}_{22}^o 
		\int_{0}^{T} \mathrm{d}t \,\chi_{\mathscr{N}_{\ell}^T}(t) 
		\iint_{[\Omega_{\ell}, \infty)^2} \mathrm{d}\omega_1\, \mathrm{d}\omega_2\,
		\mathfrak{F}_1 \omega_1\, \mathfrak{F}_2 \omega_2\,
		(\omega_1+\omega_2)^{3\theta+\varpi_1+\alpha-2}\omega_1^{\varpi_1+1}\omega_2^{\varpi_1+1}   \\[0.5em]  
		& + c_{22} \Delta_\ell^2\Omega_\ell^{4\varpi_2-2+\alpha+\gamma}	
				\int_{0}^{T} \mathrm{d}t \,\chi_{\mathscr{N}_{\ell}^T}(t) 
		\left( \int_{[\Omega_{\ell}, \infty)} \mathrm{d}\omega\, \mathfrak{F} \omega \right)^3 \\[0.5em]
		& + c_{31}\mathcal{C}_{22}^o 
		\int_{0}^{T} \mathrm{d}t \,\chi_{\mathscr{N}_{\ell}^T}(t) 
			\iiint_{[\Omega_{\ell}, \infty)^3} \mathrm{d}\omega_1\, \mathrm{d}\omega_2\, \mathrm{d}\omega_3\,
		(\omega_1+\omega_2+\omega_3)^{3\theta+\varpi_3+\alpha-2}
		\omega_1^{\varpi_3}\omega_2^{\varpi_3}\omega_3^{\varpi_3} \\[0.5em]
		&\quad \times \mathfrak{F}_1 \omega_1\, \mathfrak{F}_2 \omega_2\, \mathfrak{F}_3 \omega_3\, 
		(\omega_1 \omega_2 + \omega_2 \omega_3 + \omega_3 \omega_1).
	\end{aligned}
\end{equation}

We consider the following cases.

\begin{itemize}
	
	\item If \( c_{12} > 0 \), then using~\eqref{Sec:DD:2}, we obtain, for \( t \in \mathscr{N}_{\ell}^T \) and recalling the definitions in Assumption~X,
	\[
	\int_{[\Omega_{\ell}, \Omega_{\ell+1})} \mathrm{d}\omega\, \mathfrak{F}(t, \omega) 
	\ge c_o\, \Omega_{\ell+1}^{-\sigma},
	\]
	which implies	
	\begin{equation} \label{Propo:Mutiscale2:E2}
		\begin{aligned}
			\mathfrak{M} + \mathfrak{E} \ \ge\ 
			c_{12} \mathcal{C}_{22}^o\, \Omega_\ell^{3\theta+3\varpi_1+\alpha}
			\int_{0}^{T} \mathrm{d}t\, \chi_{\mathscr{N}_{\ell}^T}(t) 
			\left( c_o\, \Omega_{\ell+1}^{-\sigma} \right)^2,
		\end{aligned}
	\end{equation}
	and consequently,
	\begin{equation} \label{Propo:Mutiscale2:E3}
		\begin{aligned}
			C_0 \left( \mathfrak{M} + \mathfrak{E} \right)\, \Omega_{\ell}^{2\sigma - 3\varpi_{1}-3\theta-\alpha} 
			\ \ge\ \left| \mathscr{N}_{\ell}^T \right|,
		\end{aligned}
	\end{equation}
	for a universal constant \( C_0 > 0 \), which may vary from line to line.

	\item If \( c_{31} > 0 \), then using~\eqref{Sec:DD:2}, we bound
	\begin{equation} \label{Propo:Mutiscale2:E4}
		\begin{aligned}
			\mathfrak{M} + \mathfrak{E} \ \ge\ 
			c_{31} \mathcal{C}_{22}^o\Omega_\ell^{3\theta+4\varpi_3+\alpha}
			\int_{0}^{T} \mathrm{d}t\, \chi_{\mathscr{N}_{\ell}^T}(t) 
			\left( c_o\, \Omega_{\ell+1}^{-\sigma} \right)^3,
		\end{aligned}
	\end{equation}
	which implies
	\begin{equation} \label{Propo:Mutiscale2:E5}
		\begin{aligned}
			C_0 \left( \mathfrak{M} + \mathfrak{E} \right)\, \Omega_{\ell}^{3\sigma - 4\varpi_{3}-3\theta-\alpha} 
			\ \ge\ \left| \mathscr{N}_{\ell}^T \right|,
		\end{aligned}
	\end{equation}
	for a universal constant \( C_0 > 0 \), which may vary from line to line.
	
	\item If \( c_{22} > 0 \), then using~\eqref{Sec:DD:2}, we bound
	\begin{equation} \label{Propo:Mutiscale2:E6}
		\begin{aligned}
			\mathfrak{M} + \mathfrak{E} \ \ge\ 
			c_{22} \mathcal{C}_{22}^o \lambda^4 
			\cdot \Delta_\ell^2\Omega_\ell^{4\varpi_2-2+ \alpha+\gamma}
			\int_{0}^{T} \mathrm{d}t\, \chi_{\mathscr{N}_{\ell}^T}(t) 
			\left( c_o\, \Omega_{\ell}^{-\sigma} \right)^3,
		\end{aligned}
	\end{equation}
	which implies
	\begin{equation*}
		\begin{aligned}
			\frac{C_0 \left( \mathfrak{M} + \mathfrak{E} \right) }{\Delta_{\ell}^2\Omega_\ell^{-3\sigma+4\varpi_2-2+ \alpha+\gamma}}
			\ \ge\ \left| \mathscr{N}_{\ell}^T \right|,
		\end{aligned}
	\end{equation*}
	for a universal constant \( C_0 > 0 \), which may vary from line to line.
	
	Using~\eqref{Sec:DD:0}, we bound
	\begin{equation*}
		C_0 \left( \mathfrak{M} + \mathfrak{E} \right)\, \Omega_\ell^{3\sigma-(4\varpi_2+ \alpha+\gamma)} \cdot 2^{2\Upsilon_\ell }
		\ \ge\ \left| \mathscr{N}_{\ell}^T \right|,
	\end{equation*}
	which implies
	\begin{equation} \label{Propo:Mutiscale2:E7}
		C_0 \left( \mathfrak{M} + \mathfrak{E} \right)\, \Omega_\ell^{3\sigma-\frac{4\varpi_2+ \alpha+\gamma}{2}} 
		\ \ge\ \left| \mathscr{N}_{\ell}^T \right|.
	\end{equation}
	
\end{itemize}

Combining all the cases discussed above, we conclude the bound~\eqref{Lemma:Mutiscale2:1}.

			{\bf Part 2: Proof of~\eqref{Lemma:Mutiscale2:11}.  }	

	To this end, we combine~\eqref{Lemma:Mutis1:1} and~\eqref{Lemma:Mutiscale2:1} to estimate
	\begin{equation*}
		\begin{aligned}
			\left| \mathscr{N}_\ell^T \cup \mathscr{P}_{\ell}^T \right|  
			&\le C_{\mathscr{P}}\, \Omega_{\ell}^{-c_{\mathscr{P}}} 
			+ C_{\mathscr{N}}\, \Omega_{\ell}^{-c_{\mathscr{N}}} 
			\le C_0\, \Omega_{\ell}^{-c_{\mathscr{N},\mathscr{P}}},
		\end{aligned}
	\end{equation*}
	where \( C_0 > 0 \) is a universal constant that may vary from line to line.
	
	Now, we bound
	\begin{equation*}
		\left\vert \bigcup_{i=\ell}^\infty \left( \mathscr{M}_i^T \setminus \mathscr{Q}_i^T \right) \right\vert 
		\le \left\vert \bigcup_{i=\ell}^\infty \left( \mathscr{N}_i^T \cup \mathscr{P}_i^T \right) \right\vert
		\le \sum_{i=\ell}^{\infty} \left| \mathscr{N}_i^T \cup \mathscr{P}_i^T \right| 
		\le C_0 \sum_{i=\ell}^{\infty} \Omega_i^{-c_{\mathscr{N},\mathscr{P}}} 
		\le 2{C_0}\, \Omega_{\ell}^{-c_{\mathscr{N},\mathscr{P}}},
	\end{equation*}
	which completes the proof of~\eqref{Lemma:Mutiscale2:11}.

%
%
%
%
%
%
	
				{\bf Part 3: Proof of~\eqref{Lemma:Mutiscale3:1}.  }	

We aim to prove that for all   \( \sigma > 0 \) satisfying \eqref{sigma}, there exits  \( \mathfrak{C}_\mathscr{A} > 0 \) such that for all \( \ell > \mathfrak{C}_\mathscr{A}(\sigma) \), the following inclusion holds:
\begin{equation} \label{Lemma:Mutiscale3:E1}
	\mathscr{M}_{\ell}^T \subset \bigcup_{i=\ell}^\infty \left( \mathscr{M}_i^T \setminus \mathscr{Q}_i^T \right).
\end{equation}
By definition, we observe that
\[
\bigcup_{i=\ell}^\infty \left( \mathscr{M}_i^T \setminus \mathscr{M}_{i+1}^T \right) \subset \bigcup_{i=\ell}^\infty \left( \mathscr{M}_i^T \setminus \mathscr{Q}_i^T \right).
\]
Hence, to prove \eqref{Lemma:Mutiscale3:E1}, it suffices to establish the existence of  \( \mathfrak{C}_\mathscr{A}(\sigma) > 0 \) such that for all \( \ell > \mathfrak{C}_\mathscr{A}(\sigma) \), we have
\begin{equation} \label{Lemma:Mutiscale3:E2}
	\mathscr{M}_{\ell}^T \setminus \left( \bigcup_{i=\ell}^\infty \left( \mathscr{M}_i^T \setminus \mathscr{M}_{i+1}^T \right) \right) = \emptyset.
\end{equation}

We proceed by contradiction. Suppose there exist $\sigma$ satisfying  \eqref{sigma} and a sequence
\[
\ell(\sigma) = \ell_0(\sigma) < \ell_1(\sigma) < \ell_2(\sigma) < \cdots
\]
such that for each \( j \in \mathbb{N} \), the set
\[
\mathscr{M}_{\ell_j(\sigma)}^T \setminus \left( \bigcup_{i=\ell_j(\sigma)}^\infty \left( \mathscr{M}_i^T \setminus \mathscr{M}_{i+1}^T \right) \right)
\]
is nonempty.

We divide the proof into several steps.

\textit{Step 1:}  
Let \( t(\sigma) \in \mathscr{M}_{\ell(\sigma)}^T \setminus \left( \bigcup_{i =\ell(\sigma)}^\infty \left(  \mathscr{M}_i^T \setminus  \mathscr{M}_{i+1}^T \right) \right) \subset [0, T] \).  
We claim that
\begin{equation} \label{Lemma:Mutiscale3:E3}
	\int_{[\Omega_i, \infty)} \mathrm{d}\omega \, \mathfrak{F}\left( t(\sigma), \omega \right)\, \omega \geq c_o\, \Omega_i^{-\sigma}, 
\end{equation}
for all \( i \ge \ell(\sigma) \), where \( c_o > 0 \) is a fixed constant.

To prove this, suppose for contradiction that there exists \( i_0 \ge \ell(\sigma) \) such that the bound \eqref{Lemma:Mutiscale3:E3} fails to hold at \( i = i_0 \). However, since \( t(\sigma) \in \mathscr{M}_{\ell(\sigma)}^T \), the condition \eqref{Lemma:Mutiscale3:E3} must hold at least for \( i = \ell(\sigma) \).

Therefore, we may assume that \eqref{Lemma:Mutiscale3:E3} holds for all indices
\[
i = \ell(\sigma), \, \ell(\sigma)+1, \, \dots, \, i_0 - 1,
\]
and fails only at \( i = i_0 \). This would imply that
\[
t(\sigma) \in \mathscr{M}_{i_0 - 1}^T \setminus \mathscr{M}_{i_0}^T,
\]
which contradicts the assumption that \( t(\sigma) \notin \mathscr{M}_i^T \setminus \mathscr{M}_{i+1}^T \) for all \( i \ge \ell(\sigma) \).

We conclude that \eqref{Lemma:Mutiscale3:E3} must indeed hold for all \( i \ge \ell(\sigma) \).

\textit{Step 2:}  
We replace the index sequence
\[
\ell(\sigma) = \ell_0(\sigma) < \ell_1(\sigma) < \cdots
\]
with
\[
\ell(\sigma) = \ell_j(\sigma) < \ell_{j+1}(\sigma) < \cdots
\]
for \( j \) sufficiently large, so that we may assume \( \ell(\sigma) \) is arbitrarily large.

It follows from \eqref{Lemma:Apriori2:1} and \eqref{Lemma:Mutiscale3:E3} that
\begin{equation} \label{Lemma:Mutiscale3:E4}
	\int_{[4\Omega_i, \infty)} \mathrm{d}\omega \, \mathfrak{F}(t, \omega)\, \omega \geq \frac{c_o}{3} \, \Omega_i^{-\sigma}, \quad \text{for all } t > t(\sigma),
\end{equation}
and for all \( i \ge \ell(\sigma) \).

\medskip
\noindent
Next, define
\[
t_* = \sup_{0 < \sigma < \frac15\min\left\{  \frac{4\varpi_{3}+3\theta+\alpha}{3},\ \frac{3\varpi_{1}+3\theta+\alpha}{2},\ \frac{4\varpi_2+ \alpha+\gamma}{6},\    3\varpi_2 + 2 - 2\theta + \kappa_2 - c_{\mathrm{in}} + \gamma  \right\}} t(\sigma),
\]
so that \( 0 \le t_*  < T^* \). Then, for all \(T^*> t > t_* \) and \( i \ge \ell(\sigma) \), we have
\begin{equation} \label{Lemma:Multiscale3:E15}
	\int_{[4\Omega_i, \infty)} \mathrm{d}\omega \, \mathfrak{F}(t, \omega)\, \omega 
	\ge \frac{c_0}{3} \, \Omega_i^{-\sigma} 
	= \frac{c_0}{4^{-\sigma} \cdot 3} (4\Omega_i)^{-\sigma}.
\end{equation}

\medskip
\noindent
Fix \( \sigma \), and let \( \ell^* = \ell^*(\sigma) \) be such that
\begin{equation} \label{Lemma:Mutiscale3:E16}
	\int_{[\Omega_j, \infty)} \mathrm{d}\omega\, \mathfrak{F}(t, \omega)\, \omega > c_o' \, \Omega_j^{-\sigma},
\end{equation}
for all \( j > \ell^*(\sigma) \) and for all \( t \in (t_*, T^*) \), where \( c_o' > 0 \) is independent of \( j \) and \( t \).

\medskip
\noindent
Now fix \( t_* < T_1 < T^* \), and let \( j > \ell^*(\sigma) \). For \( t \in (t_*, T_1) \), we aim to prove the existence of an unbounded set \( \mathfrak{Y}(t) \subset \mathbb{N} \) such that for all \( j \in \mathfrak{Y}(t) \),
\begin{equation} \label{Lemma:Mutiscale3:E17}
	\int_{[\Omega_j, \Omega_{j+1})} \mathrm{d}\omega\, \mathfrak{F}(t, \omega)\, \omega 
	> \frac{2^\sigma - 1}{2^\sigma} \cdot \frac{c_o'}{2} \Omega_j^{-\sigma},
\end{equation}
while for \( j \notin \mathfrak{Y}(t) \), the inequality \eqref{Lemma:Mutiscale3:E17} fails to hold.

This can be shown via a contradiction argument.  
Suppose, to the contrary, that there exists \( \mathfrak{M}_0 > 0 \) such that for all \( j \ge \mathfrak{M}_0 \),
\[
\int_{[\Omega_{j}, \Omega_{j+1})} \mathrm{d}\omega\, \mathfrak{F}(t, \omega)\, \omega 
\le \frac{2^\sigma - 1}{2^\sigma} \cdot \frac{c_o'}{2} \, \Omega_j^{-\sigma}.
\]
Then summing over all such \( j \), we obtain
\[
\int_{\bigcup_{j = \mathfrak{M}_0}^\infty [\Omega_{j}, \Omega_{j+1})} 
\mathrm{d}\omega\, \mathfrak{F}(t, \omega)\, \omega 
\le \frac{2^\sigma - 1}{2^\sigma} \cdot \frac{c_o'}{2} \sum_{j = \mathfrak{M}_0}^\infty \Omega_j^{-\sigma}.
\]
Since \( \Omega_j = \Omega_{\mathfrak{M}_0} \cdot 2^{j - \mathfrak{M}_0} \), we have
\[
\Omega_j^{-\sigma} = \Omega_{\mathfrak{M}_0}^{-\sigma} \cdot 2^{-(j - \mathfrak{M}_0)\sigma},
\]
and hence,
\[
\sum_{j = \mathfrak{M}_0}^\infty \Omega_j^{-\sigma} 
= \Omega_{\mathfrak{M}_0}^{-\sigma} \sum_{j = 0}^\infty 2^{-j\sigma} 
= \Omega_{\mathfrak{M}_0}^{-\sigma} \cdot \frac{1}{1 - 2^{-\sigma}} 
= \Omega_{\mathfrak{M}_0}^{-\sigma} \cdot \frac{2^\sigma}{2^\sigma - 1}.
\]
Therefore,
\[
\int_{\bigcup_{j = \mathfrak{M}_0}^\infty [\Omega_{j}, \Omega_{j+1})} 
\mathrm{d}\omega\, \mathfrak{F}(t, \omega)\, \omega 
\le \Omega_{\mathfrak{M}_0}^{-\sigma} \cdot \frac{c_o'}{2}.
\]

On the other hand, recall from earlier (cf. \eqref{Lemma:Mutiscale3:E16}) that
\[
\int_{[\Omega_{\mathfrak{M}_0}, \infty)} \mathrm{d}\omega\, \mathfrak{F}(t, \omega)\, \omega 
> c_o' \, \Omega_{\mathfrak{M}_0}^{-\sigma},
\]
which yields a contradiction, since
\[
\int_{[\Omega_{\mathfrak{M}_0}, \infty)} \mathrm{d}\omega\, \mathfrak{F}(t, \omega)\, \omega 
>  c_o' \, \Omega_{\mathfrak{M}_0}^{-\sigma} > \Omega_{\mathfrak{M}_0}^{-\sigma} \cdot \frac{c_o'}{2} \ge \int_{[\Omega_{\mathfrak{M}_0}, \infty)} \mathrm{d}\omega\, \mathfrak{F}(t, \omega)\, \omega.
\]
Thus, our assumption must be false. Therefore, there exists an unbounded set \( \mathfrak{Y}(t) \subset \mathbb{N} \) such that for all \( j \in \mathfrak{Y}(t) \), the inequality \eqref{Lemma:Mutiscale3:E17} holds.

\textit{Step 3:}

For \( \ell \in \mathbb{N} \), define \( \mathfrak{X}_{\ell} \subset (t_*, T_1) \) as the set of all times \( t \) such that \( \ell \in \mathfrak{Y}(t) \), that is,
\[
\bigcup_{\ell \in \mathbb{N}} \mathfrak{X}_{\ell} = (t_*, T_1).
\]

We claim that there exists \( \ell_0 \in \mathbb{N} \) such that
\[
\left| \mathfrak{X}_{\ell_0} \right| \ge |T_1 - t_*| \cdot 2^{-\ell_0  \sigma} \cdot \frac{1}{2} \cdot \frac{2^ \sigma - 1}{2^ \sigma}.
\]
We proceed by contradiction. Suppose that
\[
\left| \mathfrak{X}_{\ell} \right| \le |T_1 - t_*| \cdot 2^{-\ell  \sigma} \cdot \frac{1}{2} \cdot \frac{2^ \sigma - 1}{2^ \sigma}
\quad \text{for all } \ell \in \mathbb{N}.
\]
Then summing over all \( \ell \in \mathbb{N} \), we find
\[
\left| \bigcup_{\ell \in \mathbb{N}} \mathfrak{X}_{\ell} \right| 
\le \frac{1}{2} \cdot \frac{2^ \sigma - 1}{2^ \sigma} \cdot |T_1 - t_*| \cdot \sum_{\ell=0}^\infty 2^{-\ell  \sigma}
= \frac{1}{2} \cdot \frac{2^ \sigma - 1}{2^ \sigma} \cdot |T_1 - t_*| \cdot \frac{1}{1 - 2^{- \sigma}}.
\]
Simplifying the geometric sum gives
\[
\left| \bigcup_{\ell \in \mathbb{N}} \mathfrak{X}_{\ell} \right| 
\le |T_1 - t_*| \cdot \frac{1}{2} < |T_1 - t_*| 
= \left| \bigcup_{\ell \in \mathbb{N}} \mathfrak{X}_{\ell} \right|,
\]
which is a contradiction. Therefore, there must exist some \( \ell_0 \in \mathbb{N} \) such that
\[
\left| \mathfrak{X}_{\ell_0} \right| 
\ge |T_1 - t_*| \cdot 2^{-\ell_0  \sigma} \cdot \frac{1}{2} \cdot \frac{2^ \sigma - 1}{2^ \sigma}.
\]

Since for each \( t \in (t_*, T_1) \), the set \( \mathfrak{Y}(t) \) is unbounded, it follows that the subset \( \{ \ell \in \mathfrak{Y}(t) \mid \ell > \ell_0 \} \) is also unbounded. Consequently,
\[
\bigcup_{\ell \in \mathbb{N},\, \ell > \ell_0} \mathfrak{X}_\ell = (t_*, T_1).
\]
Applying the same argument as before, we deduce the existence of \( \ell_1 > \ell_0 \) such that
\[
|\mathfrak{X}_{\ell_1}| \ge |T_1 - t_*| \cdot 2^{-\ell_1  \sigma} \cdot \frac{1}{2} \cdot \frac{2^ \sigma - 1}{2^ \sigma}.
\]

Repeating this procedure inductively, we construct an unbounded sequence \( \{ \ell_0, \ell_1, \ell_2, \dots \} \subset \mathbb{N} \) such that for each \( \ell \) in the sequence,
\[
|\mathfrak{X}_{\ell}| \ge |T_1 - t_*| \cdot 2^{-\ell  \sigma} \cdot \frac{1}{2} \cdot \frac{2^ \sigma - 1}{2^ \sigma}.
\]

Now, let \( \ell \in \{ \ell_0, \ell_1, \dots \} \) be large enough. By the pigeonhole principle and \eqref{Lemma:Mutiscale3:E17}, there exists a set \( \mathscr{D}_{i}^{\Omega_\ell} \subset \left[ \Omega_{\ell}, \Omega_{\ell+1} \right) \) such that for all \( t \in \mathfrak{X}_\ell \),
\begin{equation} \label{Lemma:Mutiscale3:E18}
	\int_{\mathscr{D}_{i}^{{\Omega_\ell}}} \mathrm{d}\omega \, \mathfrak{F}(t, \omega)\omega 
	> \frac{2^\sigma - 1}{2^\sigma} \cdot \frac{c_o'}{100} \cdot \frac{\Omega_\ell^{-\sigma}}{\Upsilon_\ell}
	\ge c_o'' \, \Omega_\ell^{-\sigma -  \min\left\{ \frac{1}{4} \left( 4\varpi_2+ \alpha+\gamma \right),\ \frac{\epsilon}{16} \right\} },
\end{equation}
for some constant \( c_o'' > 0 \) (which may vary line to line). This implies the following bounds:
	\begin{equation} \label{Lemma:Mutis:E19a}
	\begin{aligned}
		& \int_0^{T} \mathrm{d}t\, 4^{-1}\, \theta\, C_\omega^{-1} C_{\mathfrak{P}}^{3}\, c_{12}\, 
		\lvert \Omega_{\ell} \rvert^{3\theta + 3\varpi_1 + 1}
		\left[ \int_{\mathbb{R}_+} \mathrm{d}\omega\,
		\mathfrak{F}\, \omega\,
		\chi_{\big\{\omega \in \mathscr{D}_{i}^{\Omega_{\ell}}\big\}} \right]^2
		\chi_{\mathfrak{X}_\ell}(t) \\[0.4em]
		& \quad + \int_0^{T} \mathrm{d}t\, \tfrac{1}{2}\, c_{31}\, C_{\mathfrak{Q}}^{3}\, \theta\, C_\omega^{-1} C_{\mathfrak{P}}\, 
		\lvert \Omega_{\ell} \rvert^{3\theta + 4\varpi_3 + 1}
		\left[ \int_{\mathbb{R}_+} \mathrm{d}\omega\,
		\mathfrak{F}\, \omega\,
		\chi_{\big\{\omega \in \mathscr{D}_{i}^{\Omega_{\ell}}\big\}} \right]^3
	\chi_{\mathfrak{X}_\ell}(t)\\
		&> \mathcal{C}_1'' \left[
		c_{12} \Omega_{\ell}^{3\theta + 3\varpi_1 + 1-2\sigma-\frac{\epsilon}{4}} +
		c_{31} \Omega_{\ell}^{3\theta + 4\varpi_3 + 1 - 3\sigma -\frac{\epsilon}{4}}
		\right]\cdot|T_1 - t_*| \cdot  \Omega_{\ell}^{-\sigma} \\
		&> \mathcal{C}_1' \Omega_{\ell}^{\frac{\epsilon}{4}},
	\end{aligned}
\end{equation}
when \( c_{22} = 0 \), and
\begin{equation} \label{Lemma:Mutis:E19}
	\begin{aligned}
		& \int_0^{T} \mathrm{d}t\,  \chi_{\mathfrak{X}_\ell}(s)
		\Bigg[
		2^{\kappa_2 - 5 - c_{\mathrm{in}}}\, c_{22}\, C_{\omega}^{2\theta}\, C_{\mathfrak{R}}^{2}\,
		C_{\mathfrak{R}'}  
		\int_{\mathscr{D}_{\rho(s)}^{\Omega_{\ell_j}}} \mathrm{d}\omega\, \mathfrak{F}(s,\omega)\, \omega \\[0.4em]
		& \quad \times \Omega_{\ell}^{3\varpi_2 + 2 - 2\theta + \kappa_2 - c_{\mathrm{in}} + \gamma}
		\Bigg]
	\\
	&> \mathcal{C}_1'' \left[
	c_{22} \Omega_{\ell}^{3\varpi_2 + 2 - 2\theta + \kappa_2 - c_{\mathrm{in}} + \gamma}
	\right]\cdot|T_1 - t_*| \cdot  \Omega_{\ell}^{-\sigma} \\
	&> \mathcal{C}_1' \Omega_{\ell}^{\frac{\epsilon}{4}},
	\end{aligned}
\end{equation}
when \( c_{22} \ne 0 \)
for some constants \( \mathcal{C}_1'', \mathcal{C}_1' > 0 \). It is important to note that the above bound makes use of the inequality \eqref{epsilon}.

There exists \( T_1 < T \) such that

	\begin{equation} \label{Lemma:Mutis:E20a}
	\begin{aligned}
		& \int_0^{T_1} \mathrm{d}t\, 4^{-1}\, \theta\, C_\omega^{-1} C_{\mathfrak{P}}^{3}\, c_{12}\, 
		\lvert \Omega_{\ell} \rvert^{3\theta + 3\varpi_1 + 1}
		\left[ \int_{\mathbb{R}_+} \mathrm{d}\omega\,
		\mathfrak{F}\, \omega\,
		\chi_{\big\{\omega \in \mathscr{D}_{i}^{\Omega_{\ell}}\big\}} \right]^2
		\chi_{\mathfrak{X}_\ell}(t) \\[0.4em]
		& \quad + \int_0^{T} \mathrm{d}t\, \tfrac{1}{2}\, c_{31}\, C_{\mathfrak{Q}}^{3}\, \theta\, C_\omega^{-1} C_{\mathfrak{P}}\, 
		\lvert \Omega_{\ell} \rvert^{3\theta + 4\varpi_3 + 1}
		\left[ \int_{\mathbb{R}_+} \mathrm{d}\omega\,
		\mathfrak{F}\, \omega\,
		\chi_{\big\{\omega \in \mathscr{D}_{i}^{\Omega_{\ell}}\big\}} \right]^3
		\chi_{\mathfrak{X}_\ell}(t)
	\\
		& = \mathcal{C}_1' \Omega_{\ell}^{\frac{\epsilon}{4}},
	\end{aligned}
\end{equation}
when \( c_{22} = 0 \), and
\begin{equation} \label{Lemma:Mutis:E20}
	\begin{aligned}
		& \int_0^{T_1} \mathrm{d}t\,  \chi_{\mathfrak{X}_\ell}(s)
		\Bigg[
		2^{\kappa_2 - 5 - c_{\mathrm{in}}}\, c_{22}\, C_{\omega}^{2\theta}\, C_{\mathfrak{R}}^{2}\,
		C_{\mathfrak{R}'}  
		\int_{\mathscr{D}_{i}^{\Omega_{\ell}}} \mathrm{d}\omega\, \mathfrak{F}(s,\omega)\, \omega \\[0.4em]
		& \quad \times \Omega_{\ell}^{3\varpi_2 + 2 - 2\theta + \kappa_2 - c_{\mathrm{in}} + \gamma}
		\Bigg]
	\\
		&= \mathcal{C}_1' \Omega_{\ell}^{\frac{\epsilon}{4}},
	\end{aligned}
\end{equation}

Similarly to~\eqref{Lemma:Supersolu:E9:1} and \eqref{Lemma:Multiscale1:E12}, we estimate

\begin{equation} \label{Lemma:Supersolu:E21a}
	\begin{aligned}
		\mathfrak{M} + \mathfrak{E} 
		\ge\; & \int_{\mathbb{R}_+} \mathrm{d}\omega\, \mathfrak{F}(T_1)\, \vartheta \\[0.3em]
		\ge\; & \int_0^{T_1} \mathrm{d}t\,4^{-1}\, \theta\, C_\omega^{-1} C_{\mathfrak{P}}^{3}\, c_{12}\, 
		\lvert \Omega_\ell \rvert^{3\theta + 3\varpi_1 + 1}
		\left[ \int_{\mathbb{R}_+} \mathrm{d}\omega\,
		\mathfrak{F}\, \omega\,
		\chi_{\big\{\omega \in \mathscr{D}_{i}^{\Omega_{\ell}}\big\}} \right]^2
			\chi_{\mathfrak{X}_\ell}(t) \\[0.4em]
		& + \int_0^{T_1} \mathrm{d}t\,\tfrac{1}{2}\, c_{31}\, C_{\mathfrak{Q}}^{3}\, \theta\, C_\omega^{-1} C_{\mathfrak{P}}\, 
		\lvert \Omega_\ell \rvert^{3\theta + 4\varpi_3 + 1}
		\left[ \int_{\mathbb{R}_+} \mathrm{d}\omega\,
		\mathfrak{F}\, \omega\,
		\chi_{\big\{\omega \in \mathscr{D}_{i}^{\Omega_{\ell}}\big\}} \right]^3
		\chi_{\mathfrak{X}_\ell}(t) \\[0.4em]
		\ge\; & \mathcal{C}_1'\, \Omega_{\ell}^{\frac\epsilon4},
	\end{aligned}
\end{equation}
and
\begin{equation}\label{Lemma:Multiscale1:E21}
	\begin{aligned}
		\int_{\mathbb{R}^{+}} \mathrm{d}\omega\, \mathfrak{F}(t,\omega)\, \vartheta(\omega)
		\ge\; & \tfrac{1}{20}\, C_{\mathrm{ini}}\, (2\Omega_{\ell})^{-c_{\mathrm{in}}}
		\exp\Bigg(
		\int_0^t \mathrm{d}s\, 	\chi_{\mathfrak{X}_\ell}(t)(s)
		\Bigg[ \\[0.4em]
		& \quad 2^{\kappa_2 - 5 - c_{\mathrm{in}}}\, c_{22}\, C_{\omega}^{2\theta}\, C_{\mathfrak{R}}^{2}\,
		C_{\mathfrak{R}'} 
		\int_{\mathscr{D}_{i}^{\Omega_{\ell}}} \mathrm{d}\omega\, 
		\mathfrak{F}(s,\omega)\, \omega \\[0.4em]
		& \quad \times \Omega_{\ell}^{3\varpi_2 + 2 - 2\theta + \kappa_2 - c_{\mathrm{in}} + \gamma}
		\Bigg] \Bigg).
	\end{aligned}
\end{equation}

By the same argument as in~\eqref{Lemma:Supersolu:E9:1} and \eqref{Lemma:Multiscale1:E13}, we reach a contradiction as \( \ell \to \infty \). Hence,~\eqref{Lemma:Mutiscale3:E1} holds, and consequently,~\eqref{Lemma:Mutiscale3:1} follows.

\end{proof}

\section{Energy cascade}\label{Sec:Third}
\begin{proposition} 
	\label{Propo:Cascade1} We assume Assumptions X and Y. 
	
	\begin{itemize}
		\item[(I)]
Suppose
	\begin{equation}\label{Propo:Cascade1:1} 
		0<	c_{\mathrm{ini}} <  		\frac15\min\left\{  \frac{4\varpi_{3}+3\theta+\alpha}{3},\ \frac{3\varpi_{1}+3\theta+\alpha}{2},\ \frac{4\varpi_2+ \alpha+\gamma}{6},  3\varpi_2 + 2 - 2\theta + \kappa_2 - c_{\mathrm{in}} + \gamma \right\},
	\end{equation}

	then $$	T^* =0.$$
	
			\item[(II)] Suppose only that 	\begin{equation} \label{Propo:Cascade1:2}
	0 < c_{\mathrm{in}} <  3\varpi_2 + 2 - 2\theta + \kappa_2   + \gamma.
			\end{equation}  Then
			\[
			T^* < \infty.
			\]
	
		\end{itemize} 
\end{proposition}

\begin{proof} We first prove (I). Choose  $ T\in[0,T^*)$.
We begin by selecting parameters \( c_{\mathrm{ini}} = \sigma > 0 \) and \( \epsilon \) such that
\begin{equation*}
	\begin{aligned}
		& 0 < c_{\mathrm{ini}} = \sigma < 	\frac15\min\left\{  \frac{4\varpi_{3}+3\theta+\alpha}{3},\ \frac{3\varpi_{1}+3\theta+\alpha}{2},\ \frac{4\varpi_2+ \alpha+\gamma}{6},  3\varpi_2 + 2 - 2\theta + \kappa_2 - c_{\mathrm{in}} + \gamma \right\}, 
	\end{aligned}
\end{equation*}

and 
\begin{equation*}
	\begin{aligned}
		10(\epsilon + \sigma)  &< 3\theta + 3\varpi_1 + 1, \\
	10(\epsilon + \sigma)   &< 3\varpi_2 + 2 - 2\theta + \kappa_2 - c_{\mathrm{in}} +\gamma, \\
	10(\epsilon + \sigma)   &< 3\theta + 4\varpi_3 + 1.
	\end{aligned}
\end{equation*}

Let us choose the test function \(\Xi(\omega) = (\omega - \Omega_\ell)_+\). Then, we obtain the following estimate, following~\eqref{Lemma:Apriori2:1}:
\begin{equation*}
	\begin{aligned}
		\int_{[\Omega_\ell, \infty)} \, \mathrm{d}\omega \omega\, \mathfrak{F}(t, \omega)
		&\ge \int_{\mathbb{R}^+} \, \mathrm{d}\omega (\omega - \Omega_\ell)_+\, \mathfrak{F}(t, \omega)\,  \\
		&\ge \int_{\mathbb{R}^+} \, \mathrm{d}\omega (\omega - \Omega_\ell)_+\, \mathfrak{F}(0, \omega)\,  \\
		&\ge \frac{1}{2} \int_{[2\Omega_\ell, \infty)}\, \mathrm{d}\omega  \omega\, \mathfrak{F}(0, \omega) \\
		&\ge \frac{1}{2} C_{\mathrm{in}} (2\Omega_\ell)^{-c_{\mathrm{in}}},
	\end{aligned}
\end{equation*}
for all time \( t \) prior to the  time   \( T^* \).  
This implies that
\begin{equation*}
	\int_{[\Omega_\ell, \infty)} \, \mathrm{d}\omega \omega\, \mathfrak{F}(t, \omega)\, 
	\ge 2^{-1 - c_{\mathrm{in}}} C_{\mathrm{in}}\, \Omega_\ell^{-c_{\mathrm{in}}},
	\qquad \text{for all } t \in [0, T].
\end{equation*}

Hence, applying equation~\eqref{Sec:DD:2} with the choices \( c_o = 2^{-1- c_{\mathrm{in}}} C_{\mathrm{in}} \) and \( \sigma = c_{\mathrm{in}} \), we conclude that
\begin{equation} \label{Propo:Cascade1:E1}
	\mathscr{M}_\ell^T = [0, T],
\end{equation}
for all $\ell$ sufficiently large.

By equation~\eqref{Lemma:Mutiscale3:1}, we obtain the following upper bound:
\begin{equation} \label{Propo:Cascade1:E2}
	\left| \mathscr{M}_\ell^T \right| 
	\le C_{\mathscr{M}} \left( \Omega_\ell \right)^{-c_{\mathscr{N}, \mathscr{P}}},
\end{equation}
for all 
\[
T \in \left[0, T^*\right).
\]

Combining \eqref{Propo:Cascade1:E1} and \eqref{Propo:Cascade1:E2}, we deduce that
\begin{equation*}
	T^* \leq C_{\mathscr{M}} \left( \Omega_\ell \right)^{-c_{\mathscr{N}, \mathscr{P}}} \longrightarrow 0,
\end{equation*}
as \( \ell \to \infty \).

We next prove (II). Choose \( T \in [0, T^*) \). Let \( \varrho \) be a sufficiently small constant satisfying
	\begin{equation}\label{Propo:Cascade1:1} \begin{aligned}
		0 < \varrho < &	\frac15\min\left\{  \frac{4\varpi_{3}+3\theta+\alpha}{3},\ \frac{3\varpi_{1}+3\theta+\alpha}{2},\ \frac{4\varpi_2+ \alpha+\gamma}{6},  3\varpi_2 + 2 - 2\theta + \kappa_2 - c_{\mathrm{in}} + \gamma \right\} .\end{aligned}
	\end{equation}
	
	Let \( \mathcal{N}_0 \) be a sufficiently large natural number. We choose a constant \( \mathcal{C}_*^o > 0 \) such that
	\begin{equation} \label{Propo:Cascade1:2} 
		\int_{[\Omega_{\mathcal{N}_0}, \infty)} \mathrm{d}\omega\, \mathfrak{F}(0, \omega)\, \omega \ge C_{\mathrm{in}}\, \Omega_{\mathcal{N}_0}^{-c_{\mathrm{in}}} = \mathcal{C}_*^o \Omega_{\mathcal{N}_0}^{-\varrho}.
	\end{equation}
	
	From \eqref{Propo:Cascade1:1}, there exists \( \varepsilon > 0 \) such that
	\begin{equation*}
		\begin{aligned}
			10(\varepsilon + \sigma)  &< 3\theta + 3\varpi_1 + 1, \\
		10(\varepsilon + \sigma)   &< 3\varpi_2 + 2 - 2\theta + \kappa_2 - c_{\mathrm{in}} +\gamma, \\
		10(\varepsilon + \sigma)   &< 3\theta + 4\varpi_3 + 1
		\end{aligned}
	\end{equation*}
	
	We set \( \sigma = \varrho \) and \( \epsilon = \varepsilon \).  
	Applying equation~\eqref{Sec:DD:2} with the choices \( c_o = 2^{-1 + c_{\mathrm{in}}} C_{\mathrm{in}} \) and \( \sigma = c_{\mathrm{in}} \), we conclude that
	\begin{equation*}
		\mathscr{M}_{\mathcal{N}_0}^T = [0, T],
	\end{equation*}
	for all
	\[
	T \in [0, T^*).
	\]

By equation~\eqref{Lemma:Mutiscale3:1}, we also have the estimate
\begin{equation*}
	\left| \mathscr{M}_{\mathcal{N}_0}^T \right| 
	\le C_{\mathscr{M}} \left( \Omega_{\mathcal{N}_0} \right)^{-c_{\mathscr{N}, \mathscr{P}}},
\end{equation*}
valid for all
\[
T \in \left[0, T^*\right).
\]

Therefore, we deduce the upper bound:
\begin{equation*}
	T^* \leq C_{\mathscr{M}} \left( \Omega_{\mathcal{N}_0} \right)^{-c_{\mathscr{N}, \mathscr{P}}}.
\end{equation*}

\end{proof}

\section{Proof of the main Theorem \ref{Theorem1}}\label{Sec:Proof}
The global existence result is established by Proposition~\ref{Propo:Glo}. Moreover, Item~(i) of the main theorem is a direct consequence of Item~(I) of Proposition~\ref{Propo:Cascade1}, while Item~(ii) follows from Item~(II) of Proposition~\ref{Propo:Cascade1}.

\bibliographystyle{plain}

\bibliography{WaveTurbulence}

\def\cprime{$'$} \def\cprime{$'$} \def\cprime{$'$} \def\cprime{$'$}
  \def\cprime{$'$} \def\cprime{$'$} \def\cprime{$'$}
\begin{thebibliography}{10}

\bibitem{AlonsoGambaBinh}
R.~Alonso, I.~M. Gamba, and M.-B. Tran.
\newblock The {C}auchy problem and {BEC} stability for the quantum
  {B}oltzmann-{G}ross-{P}itaevskii system for bosons at very low temperature.
\newblock {\em arXiv preprint arXiv:1609.07467}, 2016.

\bibitem{ampatzoglou2024scattering}
I.~Ampatzoglou and T.~L{\'e}ger.
\newblock Scattering theory for the inhomogeneous kinetic wave equation.
\newblock {\em arXiv preprint arXiv:2408.05818}, 2024.

\bibitem{ampatzoglou2025ill}
I.~Ampatzoglou and T.~L{\'e}ger.
\newblock On the ill-posedness of kinetic wave equations.
\newblock {\em Nonlinearity}, 38(11):115004, 2025.

\bibitem{ampatzoglou2025optimal}
I.~Ampatzoglou and T.~L{\'e}ger.
\newblock On the optimal local well-posedness of the wave kinetic equation in
  lr.
\newblock {\em arXiv preprint arXiv:2511.15587}, 2025.

\bibitem{ampatzoglou2025inhomogeneous}
I.~Ampatzoglou, J.~K. Miller, N.~Pavlovi{\'c}, and M.~Taskovi{\'c}.
\newblock Inhomogeneous wave kinetic equation and its hierarchy in polynomially
  weighted spaces.
\newblock {\em Communications in Partial Differential Equations}, pages 1--43,
  2025.

\bibitem{WiemanCornell}
M.H. Anderson, J.R. Ensher, M.R. Matthews, C.E. Wieman, and E.A. Cornell.
\newblock Observation of {B}ose-{E}instein {C}ondensation in a dilute atomic
  vapor.
\newblock {\em Science}, 269(5221):198--201, 1995.

\bibitem{Ketterle}
M.~R. Andrews, C.~G. Townsend, H.-J. Miesner, D.~S. Durfee, D.~M. Kurn, and
  W.~Ketterle.
\newblock Observation of interference between two {B}ose condensates.
\newblock {\em Science}, 275 (5300):637--641, 1997.

\bibitem{banasiak2019analytic}
J.~Banasiak, W.~Lamb, and P.~Lauren{\c{c}}ot.
\newblock {\em Analytic Methods for Coagulation-Fragmentation Models, Volume
  I}.
\newblock Chapman and Hall/CRC, 2019.

\bibitem{benney1969random}
D.~J. Benney and A.~C. Newell.
\newblock Random wave closures.
\newblock {\em Studies in Applied Mathematics}, 48(1):29--53, 1969.

\bibitem{benney1966nonlinear}
D.~J. Benney and P.~G. Saffman.
\newblock Nonlinear interactions of random waves in a dispersive medium.
\newblock {\em Proc. R. Soc. Lond. A}, 289(1418):301--320, 1966.

\bibitem{bradley1995evidence}
Cl.~C. Bradley, C.~A. Sackett, J.~J. Tollett, and R.~G. Hulet.
\newblock Evidence of bose-einstein condensation in an atomic gas with
  attractive interactions.
\newblock {\em Physical review letters}, 75(9):1687, 1995.

\bibitem{brout1956statistical}
R.~Brout and I.~Prigogine.
\newblock Statistical mechanics of irreversible processes part viii: general
  theory of weakly coupled systems.
\newblock {\em Physica}, 22(6-12):621--636, 1956.

\bibitem{CaiLu2019_RateStrongConvBE}
S.~Cai and X.~Lu.
\newblock The spatially homogeneous boltzmann equation for {B}ose--{E}instein
  particles: Rate of strong convergence to equilibrium.
\newblock {\em Journal of Statistical Physics}, 176(2):289--350, 2019.

\bibitem{collot2024stability}
C.~Collot, H.~Dietert, and P.~Germain.
\newblock Stability and cascades for the kolmogorov--zakharov spectrum of wave
  turbulence.
\newblock {\em Archive for Rational Mechanics and Analysis}, 248(1):7, 2024.

\bibitem{cortes2020system}
E.~Cort{\'e}s and M.~Escobedo.
\newblock On a system of equations for the normal fluid-condensate interaction
  in a bose gas.
\newblock {\em Journal of Functional Analysis}, 278(2):108315, 2020.

\bibitem{CraciunBinh}
G.~Craciun and M.-B. Tran.
\newblock A reaction network approach to the convergence to equilibrium of
  quantum {B}oltzmann equations for {B}ose gases.
\newblock {\em ESAIM: Control, Optimisation and Calculus of Variations}, 2021.

\bibitem{das2025energy}
A.~Das and M.-B. Tran.
\newblock An energy cascade finite volume scheme for a mixed 3-and 4-wave
  kinetic equation arising from the theory of finite-temperature trapped bose
  gases.
\newblock {\em arXiv preprint arXiv:2511.13064}, 2025.

\bibitem{das2024numerical}
Arijit Das and Minh-Binh Tran.
\newblock Numerical schemes for a fully nonlinear coagulation--fragmentation
  model coming from wave kinetic theory.
\newblock {\em Proceedings of the Royal Society A}, 481(2316):20250197, 2025.

\bibitem{dolce2024convergence}
Michele Dolce, Ricardo Grande, et~al.
\newblock On the convergence rates of discrete solutions to the wave kinetic
  equation.
\newblock {\em MATHEMATICS IN ENGINEERING}, 6(4):536--558, 2024.

\bibitem{escobedo2023linearized1}
M.~Escobedo.
\newblock On the linearized system of equations for the condensate--normal
  fluid interaction at very low temperature.
\newblock {\em Studies in Applied Mathematics}, 150(2):448--456, 2023.

\bibitem{escobedo2023linearized}
M.~Escobedo.
\newblock On the linearized system of equations for the condensate-normal fluid
  interaction near the critical temperature.
\newblock {\em Archive for Rational Mechanics and Analysis}, 247(5):92, 2023.

\bibitem{escobedo2025local}
M.~Escobedo.
\newblock Local classical solutions of a kinetic equation for three waves
  interactions in presence of a dirac measure at the origin.
\newblock {\em arXiv preprint arXiv:2505.00267}, 2025.

\bibitem{escobedo2024entropy}
M.~Escobedo, P.~Germain, J.~La, and A.~Menegaki.
\newblock Entropy maximizers for kinetic wave equations set on tori.
\newblock {\em arXiv preprint arXiv:2412.16026}, 2024.

\bibitem{escobedo2024instability}
M.~Escobedo and A.~Menegaki.
\newblock Instability of singular equilibria of a wave kinetic equation.
\newblock {\em arXiv preprint arXiv:2406.05280}, 2024.

\bibitem{escobedo2003homogeneous}
M.~Escobedo, S.~Mischler, and M.~A. Valle.
\newblock Homogeneous boltzmann equation in quantum relativistic kinetic
  theory.
\newblock {\em Electronic Journal of Differential Equations}, 2003, 2003.

\bibitem{escobedo2007fundamental}
M.~Escobedo, S.~Mischler, and J.~J.~L. Velazquez.
\newblock On the fundamental solution of a linearized uehling--uhlenbeck
  equation.
\newblock {\em Archive for rational mechanics and analysis}, 186(2):309--349,
  2007.

\bibitem{EscobedoBinh}
M.~Escobedo and M.-B. Tran.
\newblock Convergence to equilibrium of a linearized quantum {B}oltzmann
  equation for bosons at very low temperature.
\newblock {\em Kinetic and Related Models}, 8(3):493--531, 2015.

\bibitem{EscobedoVelazquez:2015:FTB}
M.~Escobedo and J.~J.~L. Vel{\'a}zquez.
\newblock Finite time blow-up and condensation for the bosonic {N}ordheim
  equation.
\newblock {\em Invent. Math.}, 200(3):761--847, 2015.

\bibitem{EscobedoVelazquez:2015:OTT}
M.~Escobedo and J.~J.~L. Vel{\'a}zquez.
\newblock On the theory of weak turbulence for the nonlinear {S}chr\"odinger
  equation.
\newblock {\em Mem. Amer. Math. Soc.}, 238(1124):v+107, 2015.

\bibitem{EPV}
Miguel Escobedo, Federica Pezzotti, and Manuel Valle.
\newblock Analytical approach to relaxation dynamics of condensed {B}ose gases.
\newblock {\em Ann. Physics}, 326(4):808--827, 2011.

\bibitem{GambaSmithBinh}
I.~M. Gamba, L.~M. Smith, and M.-B. Tran.
\newblock On the wave turbulence theory for stratified flows in the ocean.
\newblock {\em M3AS: Mathematical Models and Methods in Applied Sciences. Vol.
  30, No. 1 105-137}, 2020.

\bibitem{GermainIonescuTran}
P.~Germain, A.~D. Ionescu, and M.-B. Tran.
\newblock Optimal local well-posedness theory for the kinetic wave equation.
\newblock {\em Journal of Functional Analysis}, 279(4):108570, 2020.

\bibitem{germain2024stability}
P.~Germain, J.~La, and A.~Menegaki.
\newblock Stability of rayleigh-jeans equilibria in the kinetic fpu equation.
\newblock {\em arXiv e-prints}, pages arXiv--2409, 2024.

\bibitem{germain2023local}
P.~Germain, J.~La, and K.~Z. Zhang.
\newblock Local well-posedness for the kinetic mmt model.
\newblock {\em arXiv preprint arXiv:2310.11893}, 2023.

\bibitem{giri2011continuous}
A.~K. Giri, J.~Kumar, and G.~Warnecke.
\newblock The continuous coagulation equation with multiple fragmentation.
\newblock {\em Journal of mathematical analysis and applications},
  374(1):71--87, 2011.

\bibitem{GriffinNikuniZaremba:BCG:2009}
A.~Griffin, T.~Nikuni, and E.~Zaremba.
\newblock {\em Bose-condensed gases at finite temperatures}.
\newblock Cambridge University Press, Cambridge, 2009.

\bibitem{ReichlGust:2012:CII}
E.~D Gust and L.~E. Reichl.
\newblock Collision integrals in the kinetic equations ofdilute bose-einstein
  condensates.
\newblock {\em arXiv:1202.3418}, 2012.

\bibitem{halpern2009nonlinear}
L.~Halpern and J.~Szeftel.
\newblock Nonlinear nonoverlapping schwarz waveform relaxation for semilinear
  wave propagation.
\newblock {\em Mathematics of Computation}, pages 865--889, 2009.

\bibitem{hasselmann1962non}
K.~Hasselmann.
\newblock On the non-linear energy transfer in a gravity-wave spectrum part 1.
  general theory.
\newblock {\em Journal of Fluid Mechanics}, 12(04):481--500, 1962.

\bibitem{hasselmann1974spectral}
K.~Hasselmann.
\newblock On the spectral dissipation of ocean waves due to white capping.
\newblock {\em Boundary-Layer Meteorology}, 6(1-2):107--127, 1974.

\bibitem{josserand2006self}
C.~Josserand, Y.~Pomeau, and S.~Rica.
\newblock Self-similar singularities in the kinetics of condensation.
\newblock {\em Journal of Low Temperature Physics}, 145(1):231--265, 2006.

\bibitem{KD1}
T.~R. Kirkpatrick and J.~R. Dorfman.
\newblock Transport theory for a weakly interacting condensed {B}ose gas.
\newblock {\em Phys. Rev. A (3)}, 28(4):2576--2579, 1983.

\bibitem{KD2}
T.~R. Kirkpatrick and J.~R. Dorfman.
\newblock Transport in a dilute but condensed nonideal bose gas: Kinetic
  equations.
\newblock {\em J. Low Temp. Phys.}, 58:301--331, 1985.

\bibitem{LiLu2019_AnisotropicBE}
W.~Li and X.~Lu.
\newblock Global existence of solutions of the boltzmann equation for
  {B}ose--{E}instein particles with anisotropic initial data.
\newblock {\em Journal of Functional Analysis}, 276(1):231--283, 2019.

\bibitem{Lions:1989:OSA}
P.-L. Lions.
\newblock On the {S}chwarz alternating method. {II}. {S}tochastic
  interpretation and order properties.
\newblock In {\em Domain decomposition methods ({L}os {A}ngeles, {CA}, 1988)},
  pages 47--70. SIAM, Philadelphia, PA, 1989.

\bibitem{Lu2000_ModifiedBoltzmannBE}
X.~Lu.
\newblock A modified boltzmann equation for {B}ose--{E}instein particles:
  Isotropic solutions and long-time behavior.
\newblock {\em Journal of Statistical Physics}, 98(5-6):1335--1394, 2000.

\bibitem{Lu2004_IsotropicDistributionalBE}
X.~Lu.
\newblock On isotropic distributional solutions to the boltzmann equation for
  {B}ose--{E}instein particles.
\newblock {\em Journal of Statistical Physics}, 116(5-6):1597--1649, 2004.

\bibitem{Lu2005_VelocityConcBE}
X.~Lu.
\newblock The boltzmann equation for {B}ose--{E}instein particles: Velocity
  concentration and convergence to equilibrium.
\newblock {\em Journal of Statistical Physics}, 119(5-6):1027--1067, 2005.

\bibitem{Lu2013_CondensationFiniteTime}
X.~Lu.
\newblock The boltzmann equation for {B}ose--{E}instein particles: Condensation
  in finite time.
\newblock {\em Journal of Statistical Physics}, 150(6):1138--1176, 2013.

\bibitem{Lu2014_RegularityCondensation}
X.~Lu.
\newblock The boltzmann equation for {B}ose--{E}instein particles: Regularity
  and condensation.
\newblock {\em Journal of Statistical Physics}, 156(3):493--545, 2014.

\bibitem{Lu2016_LongTimeBEC}
X.~Lu.
\newblock Long time convergence of the {B}ose--{E}instein condensation.
\newblock {\em Journal of Statistical Physics}, 162(3):652--670, 2016.

\bibitem{Lu2018_LongTimeStrongBE}
X.~Lu.
\newblock Long time strong convergence to {B}ose--{E}instein distribution for
  low temperature.
\newblock {\em Kinetic and Related Models}, 11(4):715--734, 2018.

\bibitem{menegaki20222}
A.~Menegaki.
\newblock L2-stability near equilibrium for the 4 waves kinetic equation.
\newblock {\em arXiv preprint arXiv:2210.11189}, 2022.

\bibitem{mischler2008singular}
S.~Mischler, M.~Escobedo, and J.~J.~L. Velazquez.
\newblock Singular solutions for the uehling-uhlenbeck equation.
\newblock {\em Proceedings of the Royal Society of Edinburgh: Section A,
  Mathematics}, 138(1):67--107, 2008.

\bibitem{nguyen2017quantum}
T.~T. Nguyen and M.-B. Tran.
\newblock On the {K}inetic {E}quation in {Z}akharov's {W}ave {T}urbulence
  {T}heory for {C}apillary {W}aves.
\newblock {\em SIAM J. Math. Anal.}, 50(2):2020--2047, 2018.

\bibitem{Nordheim}
L.W. Nordheim.
\newblock Transport phenomena in einstein-bose and fermi- dirac gases.
\newblock {\em Proc. Roy. Soc. London A}, 119:689, 1928.

\bibitem{Peierls:1993:BRK}
R.~Peierls.
\newblock Zur kinetischen theorie der warmeleitung in kristallen.
\newblock {\em Annalen der Physik}, 395(8):1055--1101, 1929.

\bibitem{Peierls:1960:QTS}
R.~E. Peierls.
\newblock Quantum theory of solids.
\newblock In {\em Theoretical physics in the twentieth century ({P}auli
  memorial volume)}, pages 140--160. Interscience, New York, 1960.

\bibitem{PomeauBinh}
Y.~Pomeau and M.-B. Tran.
\newblock Statistical physics of non equilibrium quantum phenomena.
\newblock {\em Lecture Notes in Physics, Springer}, 2019.

\bibitem{reichl2019kinetic}
L.~E. Reichl and M.-B. Tran.
\newblock A kinetic equation for ultra-low temperature bose--einstein
  condensates.
\newblock {\em Journal of Physics. A, Mathematical and Theoretical (Online)},
  52(6), 2019.

\bibitem{rumpf2021wave}
B.~Rumpf, A.~Soffer, and M.-B. Tran.
\newblock On the wave turbulence theory: ergodicity for the elastic beam wave
  equation.
\newblock {\em Mathematische Zeitschrift}, 310(2):1--41, 2025.

\bibitem{soffer2018dynamics}
A.~Soffer and M.-B. Tran.
\newblock On the dynamics of finite temperature trapped bose gases.
\newblock {\em Advances in Mathematics}, 325:533--607, 2018.

\bibitem{soffer2019energy}
A.~Soffer and M.-B. Tran.
\newblock On the energy cascade of 3-wave kinetic equations: beyond
  kolmogorov--zakharov solutions.
\newblock {\em Communications in Mathematical Physics}, pages 1--48, 2019.

\bibitem{soffer2020energy}
A.~Soffer and M.-B. Tran.
\newblock On the energy cascade of 3-wave kinetic equations: beyond
  kolmogorov--zakharov solutions.
\newblock {\em Communications in Mathematical Physics}, 376(3):2229--2276,
  2020.

\bibitem{Spohn:2010:KOT}
H.~Spohn.
\newblock Kinetics of the bose-einstein condensation.
\newblock {\em Physica D}, 239:627--634, 2010.

\bibitem{staffilani2024condensation}
G.~Staffilani and M.-B. Tran.
\newblock Condensation and non-condensation times for 4-wave kinetic equations.
\newblock {\em arXiv preprint arXiv:2407.18533}, 2024.

\bibitem{staffilani2024energy}
G.~Staffilani and M.-B. Tran.
\newblock On the energy transfer towards large values of wavenumbers for
  solutions of 4-wave kinetic equations.
\newblock {\em arXiv preprint arXiv:2407.18508}, 2024.

\bibitem{staffilani2025condensate}
G.~Staffilani and M.-B. Tran.
\newblock Evolution of finite temperature bose-einstein condensates: Some
  rigorous studies on condensate growth.
\newblock {\em preprint}, 2025.

\bibitem{staffilani2025formation}
G.~Staffilani and M.-B. Tran.
\newblock Formation of condensations for non-radial solutions to 3-wave kinetic
  equations.
\newblock {\em arXiv preprint arXiv:2503.17066}, 2025.

\bibitem{stewart1989global}
I.~W. Stewart and E.~Meister.
\newblock A global existence theorem for the general coagulation--fragmentation
  equation with unbounded kernels.
\newblock {\em Mathematical Methods in the Applied Sciences}, 11(5):627--648,
  1989.

\bibitem{toselli2004domain}
A.~Toselli and O.~Widlund.
\newblock {\em Domain decomposition methods-algorithms and theory}, volume~34.
\newblock Springer Science \& Business Media, 2004.

\bibitem{tran2020reaction}
M.-B. Tran, G.~Craciun, L.~M. Smith, and S.~Boldyrev.
\newblock A reaction network approach to the theory of acoustic wave
  turbulence.
\newblock {\em Journal of Differential Equations}, 269(5):4332--4352, 2020.

\bibitem{tran2020boltzmann}
M.-B. Tran and Y.~Pomeau.
\newblock Boltzmann-type collision operators for bogoliubov excitations of
  bose-einstein condensates: A unified framework.
\newblock {\em Physical Review E}, 101(3):032119, 2020.

\bibitem{tran2021thermal}
M.-B. Tran and Y.~Pomeau.
\newblock On a thermal cloud--bose-einstein condensate coupling system.
\newblock {\em The European Physical Journal Plus}, 136(5):1--11, 2021.

\bibitem{walton2023numerical}
S.~Walton and M.-B. Tran.
\newblock A numerical scheme for wave turbulence: 3-wave kinetic equations.
\newblock {\em SIAM Journal on Scientific Computing}, 45(4):B467--B492, 2023.

\bibitem{walton2022deep}
S.~Walton, M.-B. Tran, and A.~Bensoussan.
\newblock A deep learning approximation of non-stationary solutions to wave
  kinetic equations.
\newblock {\em Applied Numerical Mathematics}, 2022.

\bibitem{walton2024numerical}
Steven Walton and Minh-Binh Tran.
\newblock Numerical schemes for 3-wave kinetic equations: A complete treatment
  of the collision operator.
\newblock {\em Journal of Computational Physics}, page 114147, 2025.

\bibitem{PomeauBrachetMetensRica}
S.~M'etens Y.~Pomeau, M.A.~Brachet and S.~Rica.
\newblock Th\'eorie cin\'etique d'un gaz de bose dilu\'e avec condensat.
\newblock {\em C. R. Acad. Sci. Paris S'er. IIb M'ec. Phys. Astr.},
  327:791--798, 1999.

\bibitem{zakharov2012kolmogorov}
V.~E. Zakharov, V.~S. L'vov, and G.~Falkovich.
\newblock {\em Kolmogorov spectra of turbulence I: Wave turbulence}.
\newblock Springer Science \& Business Media, 2012.

\bibitem{ZarembaNikuniGriffin:1999:DOT}
E.~Zaremba, T.~Nikuni, and A.~Griffin.
\newblock Dynamics of trapped bose gases at finite temperatures.
\newblock {\em J. Low Temp. Phys.}, 116:277--345, 1999.

\bibitem{zaslavskii1967limits}
G.~M. Zaslavskii and R.~Z. Sagdeev.
\newblock Limits of statistical description of a nonlinear wave field.
\newblock {\em Soviet physics JETP}, 25:718--724, 1967.

\end{thebibliography}

\end{document}